\documentclass[acmsmall,screen]{acmart}
\AtBeginDocument{%
  }

\setcopyright{cc}
\setcctype{by}
\acmDOI{10.1145/3763182}
\acmYear{2025}
\acmJournal{PACMPL}
\acmVolume{9}
\acmNumber{OOPSLA2}
\acmArticle{404}
\acmMonth{10}
\received{2025-03-26}
\received[accepted]{2025-08-12}

\DeclareFontFamily{OT1}{pzc}{}
\DeclareFontShape{OT1}{pzc}{m}{it}%
{<-> s * [1.15] pzcmi7t}{}
\DeclareMathAlphabet{\mathpzc}{OT1}{pzc}{m}{it}

\usepackage{urwchancal}
\usepackage{listings}
\usepackage{xspace}
\usepackage{multirow}
\usepackage{wrapfig}
\usepackage[subrefformat=parens,labelformat=parens]{subcaption}
\usepackage{tikz}
\usepackage{MnSymbol}
\usepackage[most]{tcolorbox}
\usepackage{adjustbox}   
\usepackage{extarrows}
\usepackage{todonotes}
\usepackage{csquotes}
\usepackage[os=mac]{menukeys}
\usepackage{slashbox}

\usepackage{enumitem}
\setlist{
  parsep=2pt,
  leftmargin=20pt,
  itemsep=2pt
}

\newcommand{\seckern}{\kern-2pt}

\renewcommand{\figureautorefname}{Fig.\kern-1pt}

\usepackage{mathtools}
\definecolor{expgreen}{HTML}{E0E0E0}
\definecolor{patblue}{HTML}{ABDCEC}
\definecolor{typpurp}{HTML}{C9B3EB}

\usetikzlibrary{calc, arrows, arrows.meta, backgrounds, fit, positioning, shapes.misc, graphs}

\makeatletter
\newcommand*{\RightArrow}{\mathrel{\mathpalette{\@RightArrow}{}}}
\newcommand*{\LeftArrow}{\mathrel{\mathpalette{\@LeftArrow}{}}}
\newcommand*{\@RightArrow}[1]{%
    \edef\@LineWidth{%
        \the\fontdimen8
        \ifx#1\displaystyle\textfont
        \else\ifx#1\textstyle\textfont
        \else\ifx#1\scriptstyle\scriptfont
        \else\scriptscriptfont
        \fi\fi\fi
        3}
    \edef\@ScaleWidth{%
        \ifx#1\displaystyle0.39
        \else\ifx#1\textstyle0.39
        \else\ifx#1\scriptstyle0.35
        \else0.31
        \fi\fi\fi}
    \text{$\tikz
    \draw[double equal sign distance, line width=\@LineWidth,
    -{Implies[sep=-0.5ex] . Computer Modern Rightarrow[scale width=\@ScaleWidth, scale length=0.8]}]
    (0,0) -- (1.5em,0);$}}
\newcommand*{\@LeftArrow}[1]{%
    \edef\@LineWidth{%
        \the\fontdimen8
        \ifx#1\displaystyle\textfont
        \else\ifx#1\textstyle\textfont
        \else\ifx#1\scriptstyle\scriptfont
        \else\scriptscriptfont
        \fi\fi\fi
        3}
    \edef\@ScaleWidth{%
        \ifx#1\displaystyle0.39
        \else\ifx#1\textstyle0.39
        \else\ifx#1\scriptstyle0.35
        \else0.31
        \fi\fi\fi}
    \text{$\tikz
    \draw[double equal sign distance, line width=\@LineWidth,
    -{Implies[sep=-0.5ex] . Computer Modern Rightarrow[scale width=\@ScaleWidth, scale length=0.8]}]
    (1.5em,0) -- (0,0);$}}

\makeatother

\definecolor{tg}{HTML}{3BA64C}
\colorlet{tgreen}{tg!75!darkgray}
\colorlet{eorange}{orange!75!gray}
\definecolor{mp}{HTML}{9B1D78}
\colorlet{mpurp}{mp!90!gray}

\newcommand{\pindent}{12pt}
\newcommand{\rlabelcolor}{mpurp}
\newcommand{\slabelcolor}{eorange}

\colorlet{sbg}{\slabelcolor!22}
\colorlet{rbg}{\rlabelcolor!18}
\colorlet{rbgTwo}{\rlabelcolor!12}

\newcommand{\cradius}{4pt}
\tikzset{
  node distance = 4pt and 4pt,
  box/.style = {rectangle, inner sep=0pt},
  SHIFT/.style args = {#1}{rectangle, rounded corners=\cradius,
                  fill=sbg, inner xsep=\pindent, inner ysep=2pt, fit=#1,
                  node contents={}, outer sep=0pt},
  REDUCE/.style args = {#1}{rectangle, rounded corners=\cradius,
                    fill=rbg, inner xsep=\pindent, inner ysep=2pt, fit=#1,
                    node contents={}, outer sep=0pt},
  REDUCE2/.style args = {#1}{rectangle, rounded corners=\cradius,
                    fill=rbgTwo, inner xsep=\pindent, inner ysep=2pt, fit=#1,
                    node contents={}, outer sep=0pt},
  STUCK/.style args = {#1}{
    rectangle, rounded corners=\cradius,
    fill=gray!30,
    inner xsep=\pindent,
    inner ysep=2pt,
    fit=#1,
    node contents={}
  },
  SLABEL/.style = {
    anchor=north east,
    text=\slabelcolor,
    font=\ttfamily\bfseries\footnotesize,
  },
  RLABEL/.style = {
    anchor=north east,
    text=\rlabelcolor,
    font=\ttfamily\bfseries\footnotesize,
  },
  STLABEL/.style = {
    anchor=north east,
    text=gray!80,
    font=\ttfamily\bfseries\footnotesize,
  },
  SARROW/.style = {
    double=sbg,
    double distance=1pt,
    \slabelcolor,
    line width=1.25pt,
    >=stealth
  },
  RARROW/.style = {
    double=rbg,
    double distance=1pt,
    \rlabelcolor,
    line width=1.25pt,
    >=stealth
  },
  RARROWTWO/.style = {
    double=rbgTwo,
    double distance=1pt,
    \rlabelcolor,
    line width=1.25pt,
    >=stealth
  }
}
\newcommand{\sarrow}[1]{
  \draw[SARROW] ([xshift=1.75pt]#1.north west) -- ([xshift=1.75pt]#1.south west)
}
\newcommand{\rarrow}[1]{
  \draw[RARROW] ([xshift=1.75pt]#1.north west) -- ([xshift=1.75pt]#1.south west)
}
\newcommand{\rarrowtwo}[1]{
  \draw[RARROWTWO] ([xshift=1.75pt]#1.north west) -- ([xshift=1.75pt]#1.south west)
}
\newcommand{\shiftbox}[2]{
  \node (#1) [SHIFT={#2}];
  \node [SLABEL] at (#1.north east) {Shift};
  \sarrow{#1};
}
\newcommand{\reducebox}[3][none]{
  \node (#2) [REDUCE={#3},draw={#1}];
  \node [RLABEL] at (#2.north east) {Reduce};
  \rarrow{#2};
}

\newcommand{\degroutbox}[3][none]{
  \node (#2) [REDUCE2={#3},draw={#1}];
  \node [RLABEL] at (#2.north east) {Degrout};
  \rarrowtwo{#2};
}

\usepackage{diagbox}

\usepackage{stmaryrd}
\usepackage{textgreek}
\usepackage{pifont}

\def\OPTIONConf{1}%
\usepackage{janadunfield}

\newtcbox{\gbox}[1]
{
  on line,
  nobeforeafter,
  boxsep=1pt,left=1pt,right=1pt,top=0.5pt,bottom=0.5pt,
  colframe=white,colback=gray!#1,opacityframe=0,
  highlight math style={enhanced},
  bottomrule=0pt,leftrule=0.5pt,
  enlarge left by=-2pt,
  enlarge bottom finally by=2pt,
}

\newlength{\glen}

\newcommand{\tylr}{\texttt{tall} \texttt{tylr}\@\xspace}
\newcommand{\tylrcore}{\textsf{meldr}\@\xspace}

\newcommand{\tinytylr}{\texttt{tiny} \texttt{tylr}\@\xspace}
\newcommand{\teentylr}{\texttt{teen} \texttt{tylr}\@\xspace}
\newcommand{\talltylr}{\texttt{tall} \texttt{tylr}\@\xspace}

\newcommand{\gram}{\mathcal{G}}

\newcommand{\labl}{t}
\newcommand{\lbls}{\mathcal{T}}
\newcommand{\sort}{s}
\newcommand{\sortr}{r}
\newcommand{\sorts}{\mathcal{S}}
\newcommand{\sym}{x}
\newcommand{\symy}{y}

\newcommand{\rgx}{g}

\DeclareMathOperator{\ror}{\,\textcolor{black!90}{\pmb{\mid}}\,}
\newcommand{\ralt}[2]{#1 \ror #2}
\newcommand{\ralttight}[2]{#1 \!\ror\! #2}
\newcommand{\rone}{\epsilon}
\newcommand{\rand}{\cdot}
\newcommand{\rseq}[2]{#1\rand #2}
\newcommand{\rstar}[1]{#1^*}

\newcommand{\plt}{\prec}
\newcommand{\pleq}{\preccurlyeq}
\newcommand{\pgt}{\succ}
\newcommand{\pgeq}{\succcurlyeq}

\newcommand{\opscriptsize}{\footnotesize}
\newcommand{\ltop}[1]{\lessdot_{\text{\opscriptsize $#1$}}}
\newcommand{\eqop}[1]{\doteq_{\text{\opscriptsize $#1$}}}
\newcommand{\gtop}[1]{\gtrdot_{\text{\opscriptsize $#1$}}}

\newcommand{\eqrel}[3]{#1\eqop{#2}#3}
\newcommand{\ltrel}[3]{#1\ltop{#2}#3}
\newcommand{\gtrel}[3]{#1\gtop{#2}#3}

\newcommand{\consistent}{\sim}
\newcommand{\reduce}{\Leftarrow}
\newcommand{\produce}{\Rightarrow}
\newcommand{\Reduce}{\Lleftarrow}
\newcommand{\Produce}{\Rrightarrow}

\newcommand{\synelab}[3]{#3 \reduce #2}

\newcommand{\greduce}{\reduce}
\newcommand{\syninj}[3]{#3 \greduce #2}

\newcommand{\pruduce}{\produce}

\newcommand{\ogdots}{\dots}
\renewcommand{\dots}{\:\!...\:\!}
\newcommand{\dsty}[1]{\displaystyle{#1}}

\newcommand{\adjop}{\smallsquare}
\newcommand{\adj}[2]{#1\adjop#2}

\newcommand{\lmsymop}{\smalltriangleleft}
\newcommand{\lmsym}[2]{#1\lmsymop#2}
\newcommand{\lmsymStar}[2]{#1\produce^* {#2}\dots}

\newcommand{\rmsymop}{\smalltriangleright}
\newcommand{\rmsym}[2]{#1\rmsymop#2}
\newcommand{\rmsymStar}[2]{#2\produce^* \dots{#1}}

\newcommand{\wf}[1]{#1~\mathsf{wf}}

\renewcommand{\o}[1]{#1?}
\newcommand{\none}[1][gray]{\textcolor{#1}{\circ}}
\newcommand{\some}[1]{\textcolor{gray}{\bullet}{#1}}

\newcommand{\s}[1]{\overline{#1}}
\newcommand{\snil}{\cdot}

\newcommand{\dlroot}{{\leftY}}
\newcommand{\drroot}{{\rightY}}

\newcommand{\startsort}[1]{\hat{#1}}
\newcommand{\startSym}{\startsort{\sort}}

\newcommand{\groot}{\gamma}
\newcommand{\cft}{\tau}
\newcommand{\cfs}{\sigma}
\newcommand{\cfsr}{\rho}
\newcommand{\cfx}{\chi}
\newcommand{\bddsort}[3]{{
  \prescript{{\vphantom{\mathstrut}{#1}}}{{\vphantom{x}}}{{#2}}^{{\vphantom{\mathstrut}{#3}}}_{{\vphantom{x}}}
}}
\newcommand{\topsort}[1]{\bddsort{\top}{#1}{\top}}
\newcommand{\botsort}[1]{\bddsort{\bot}{#1}{\bot}}
\newcommand{\zerosort}[1]{\bddsort{\precz}{#1}{\precz}}

\newcommand{\nat}[1]{#1~\mathsf{nat}}

\newcommand{\psq}{\bddsort{\precp}{\sort}{\precq}}
\newcommand{\prq}{\bddsort{\precp}{\sortr}{\precq}}
\newcommand{\msn}{\bddsort{\precm}{\sort}{\precn}}
\newcommand{\mrn}{\bddsort{\precm}{\sortr}{\precn}}

\newcommand{\cfg}{\mathcal{H}}

\newcommand{\yields}[2]{#1\smalltriangleup#2}
\newcommand{\yieldsStar}[2]{#1\produce^* \dots{#2}\dots}

\newcommand{\promotesym}[1]{\left\lceil#1\right\rceil}

\newcommand{\head}[1]{\texttt{hd}(#1)}

\newcommand{\code}[1]{\textcolor{purple}{\texttt{#1}}}
\newcommand{\ghostcode}[1]{\textcolor{lightgray}{\texttt{#1}}}

\newcommand{\keypress}[1]{%
  \raisebox{0.6ex}{\adjustbox{valign=M}{\fcolorbox{gray!80}{gray!20}{\color{black!75}\texttt{#1\vphantom{|}}}}}\xspace%
}
\newcommand{\cmpr}{\odot}
\newcommand{\cmprleq}{\leqdot}

\newcommand{\tnode}[1]{#1}

\newcommand{\snode}[1]{\left\{#1\right\}}
\newcommand{\nnnode}[1]{\left\{#1\right\}}

\newcommand{\precm}{m}
\newcommand{\precn}{n}
\newcommand{\precp}{p}
\newcommand{\precq}{q}

\newcommand{\precz}{0}
\newcommand{\precs}{\mathcal{P}}

\newcommand{\term}{\mathbb{S}}
\newcommand{\termr}{\mathbb{R}}
\newcommand{\nod}{\mathbb{X}}
\newcommand{\stack}{\mathbb{K}}

\newcommand{\pick}[3]{#1 \rcurvearrowne #2~=~#3}
\newcommand{\pfill}[2]{#1 \rcurvearrowne #2}
\newcommand{\push}[4]{#1\xleftarrow[#2]{}#3\,=\,#4}
\newcommand{\pushj}[4][\ ]{{#2}#1\xleftarrow[#3]{}#1{#4}}
\newcommand{\pushl}[3]{#1\xleftarrow[#2]{}#3}

\newcommand{\parsel}[2]{\texttt{parse}\left(#1,#2\right)}
\newcommand{\parse}[3]{\parsel{#1}{#2}\,=\,#3}

\newcommand{\iseq}[2]{\textcolor{gray}{\left[\,\textcolor{black}{#1}\,\right]_{\textcolor{black}{#2}}}}

\newlength{\tilelen}
\newcommand{\theight}{1.3ex}
\newcommand{\tiplen}{0.45ex}

\newcommand{\expsort}{\textsc{e}}
\newcommand{\patsort}{\textsc{p}}
\newcommand{\typsort}{\textsc{t}}

\colorlet{litegray}{gray!25}
\newcommand{\expcolor}{expgreen}
\newcommand{\patcolor}{patblue}
\newcommand{\pc}{\patcolor}
\newcommand{\typcolor}{typpurp}
\newcommand{\yc}{\typcolor}

\newcommand{\mixgray}[1]{#1!50!black}

\newcommand{\vartile}[2][\expcolor]{%
  \settowidth{\tilelen}{$#2$}%
  \begin{tikzpicture}[baseline=(text.base)]%
    \draw[line width=0.75pt,{\mixgray{#1}},fill={#1}]
      (0,-\theight)
      -- (-\tiplen,0)
      -- (0,\theight)
      -- (\tilelen,\theight)
      -- ($ (\tilelen+\tiplen,0) $)
      -- (\the\tilelen,-\theight)
      -- cycle;
    \node[anchor=center] (text) at ($ (\tilelen/2,0) $) {$#2\vphantom{X}$};
  \end{tikzpicture}%
}

\newcommand{\soptile}[2][\expcolor]{%
  \settowidth{\tilelen}{$\texttt{#2}$}%
  \begin{tikzpicture}[baseline=(text.base)]%
    \draw[line width=0.75pt,{\mixgray{#1}},fill={#1}]
      (0,-\theight)
      -- (-\tiplen,0)
      -- (0,\theight)
      -- (\tilelen,\theight)
      -- ($ (\tilelen+\tiplen,0) $)
      -- (\the\tilelen,-\theight)
      -- cycle;
    \node[anchor=center] (text) at ($ (\tilelen/2,0) $) {$\texttt{#2}\vphantom{X}$};
  \end{tikzpicture}%
}
\newcommand{\spretile}[2][\expcolor]{%
  \settowidth{\tilelen}{\texttt{#2}}%
  \begin{tikzpicture}[baseline=(text.base)]%
    \draw[line width=0.75pt,{\mixgray{#1}},fill={#1}]
      (0,-\theight)
      -- (-\tiplen,0)
      -- (0,\theight)
      -- ($ (\tilelen+\tiplen,\theight) $)
      -- (\the\tilelen,0)
      -- ($ (\tilelen+\tiplen,-\theight) $)
      -- cycle;
    \node[anchor=center] (text) at ($ (\tilelen/2,0) $) {\texttt{#2}\vphantom{X}};
  \end{tikzpicture}%
}
\newcommand{\sposttile}[2][\expcolor]{%
  \settowidth{\tilelen}{\texttt{#2}}%
  \begin{tikzpicture}[baseline=(text.base)]%
    \draw[line width=0.75pt,{\mixgray{#1}},fill={#1}]
    (-\tiplen,-\theight)
      -- (0,0)
      -- (-\tiplen,\theight)
      -- (\tilelen,\theight)
      -- ($ (\tilelen+\tiplen,0) $)
      -- (\the\tilelen,-\theight)
      -- cycle;
    \node[anchor=center] (text) at ($ (\tilelen/2,0) $) {\texttt{#2}\vphantom{X}};
  \end{tikzpicture}%
}
\newcommand{\sbintile}[2][\expcolor]{%
  \settowidth{\tilelen}{\texttt{#2}}%
  \begin{tikzpicture}[baseline=(text.base)]%
    \draw[line width=0.75pt,{\mixgray{#1}},fill={#1}]
      (-\tiplen,-\theight)
      -- (0,0)
      -- (-\tiplen,\theight)
      -- ($ (\tilelen+\tiplen,\theight) $)
      -- (\the\tilelen,0)
      -- ($ (\tilelen+\tiplen,-\theight) $)
      -- cycle;
    \node[anchor=center] (text) at ($ (\tilelen/2,0) $) {\texttt{#2}\vphantom{X}};
  \end{tikzpicture}%
}

\newcommand{\holeheight}{0.45ex}
\newcommand{\holewidth}{0.25ex}
\newcommand{\holetip}{0.18ex}

\newcommand{\tophole}{%
\settowidth{\tilelen}{$x$}%
\vcenter{\hbox{%
\begin{tikzpicture}%
\draw[line width=0.75pt,black!40!]
  (-\tiplen,0)
  -- (0,\theight)
  -- (\tilelen,\theight)
  -- ($ (\tilelen+\tiplen,0) $)
  -- (\the\tilelen,-\theight)
  -- (0,-\theight)
  -- cycle;%
 \draw[line width=0.75pt,gray]
  ($ (\tilelen/2-\holewidth-\holetip,0) $)
  -- ($ (\tilelen/2-\holewidth,\holeheight) $)
  -- ($ (\tilelen/2+\holewidth,\holeheight) $)
  -- ($ (\tilelen/2+\holewidth+\holetip,0) $)
  -- ($ (\tilelen/2+\holewidth,-\holeheight) $)
  -- ($ (\tilelen/2-\holewidth,-\holeheight) $)
  -- cycle;%
\end{tikzpicture}%
}}%
}

\newcommand{\groutcolor}{black!80}
\newcommand{\ophole}[1][none]{%
\begin{tikzpicture}[baseline=0]%
\draw[line width=0.75pt,\groutcolor,fill=#1]
  (-0.25ex,0.5ex)
  -- (0,1ex)
  -- (0.5ex,1ex)
  -- (0.75ex,0.5ex)
  -- (0.5ex,0)
  -- (0,0)
  -- cycle;%
\end{tikzpicture}%
}
\newcommand{\sophole}[1][\expcolor]{%
\begin{tikzpicture}[baseline=0]%
\draw[line width=0.75pt,{\mixgray{#1}},fill={#1}]
  (-0.25ex,0.5ex)
  -- (0,1ex)
  -- (0.5ex,1ex)
  -- (0.75ex,0.5ex)
  -- (0.5ex,0)
  -- (0,0)
  -- cycle;%
\end{tikzpicture}%
}

\newcommand{\prehole}{%
\begin{tikzpicture}[baseline=0]%
\draw[line width=0.75pt,\groutcolor]
  (-0.25ex,0.5ex)
  -- (0,1ex)
  -- (0.75ex,1ex)
  -- (0.5ex,0.5ex)
  -- (0.75ex,0)
  -- (0,0)
  -- cycle;%
\end{tikzpicture}%
}
\newcommand{\sprehole}[1][\expcolor]{%
\begin{tikzpicture}[baseline=0]%
\draw[line width=0.75pt,{\mixgray{#1}},fill={#1}]
  (-0.25ex,0.5ex)
  -- (0,1ex)
  -- (0.75ex,1ex)
  -- (0.5ex,0.5ex)
  -- (0.75ex,0)
  -- (0,0)
  -- cycle;%
\end{tikzpicture}%
}

\newcommand{\poshole}{%
\begin{tikzpicture}[baseline=0]%
\draw[line width=0.75pt,\groutcolor]
  (0,0.5ex)
  -- (-0.25ex,1ex)
  -- (0.5ex,1ex)
  -- (0.75ex,0.5ex)
  -- (0.5ex,0)
  -- (-0.25ex,0)
  -- cycle;%
\end{tikzpicture}%
}
\newcommand{\sposhole}[1][\expcolor]{%
\begin{tikzpicture}[baseline=0]%
\draw[line width=0.75pt,{\mixgray{#1}},fill={#1}]
  (0,0.5ex)
  -- (-0.25ex,1ex)
  -- (0.5ex,1ex)
  -- (0.75ex,0.5ex)
  -- (0.5ex,0)
  -- (-0.25ex,0)
  -- cycle;%
\end{tikzpicture}%
}

\newcommand{\binhole}{%
\begin{tikzpicture}[baseline=0]%
\draw[line width=0.75pt,\groutcolor]
  (0,0.5ex)
  -- (-0.25ex,1ex)
  -- (0.75ex,1ex)
  -- (0.5ex,0.5ex)
  -- (0.75ex,0)
  -- (-0.25ex,0)
  -- cycle;%
\end{tikzpicture}%
}
\newcommand{\sbinhole}[1][\expcolor]{%
\begin{tikzpicture}[baseline=0]%
\draw[line width=0.75pt,{\mixgray{#1}},fill={#1}]
  (0,0.5ex)
  -- (-0.25ex,1ex)
  -- (0.75ex,1ex)
  -- (0.5ex,0.5ex)
  -- (0.75ex,0)
  -- (-0.25ex,0)
  -- cycle;%
\end{tikzpicture}%
}

\newcommand{\tbinhole}{%
\settowidth{\tilelen}{$x$}%
\vcenter{\hbox{%
\begin{tikzpicture}%
\draw[line width=0.75pt,black!40!]
  (0,0)
  -- (-\tiplen,\theight)
  -- ($ (\tilelen+\tiplen,\theight) $)
  -- (\the\tilelen,0)
  -- ($ (\tilelen+\tiplen,-\theight) $)
  -- (-\tiplen,-\theight)
  -- cycle;%
\draw[line width=0.65pt,gray]
  ($ (\tilelen/2-\holewidth,0) $)
  -- ($ (\tilelen/2-\holewidth-\holetip,\holeheight) $)
  -- ($ (\tilelen/2+\holewidth+\holetip,\holeheight) $)
  -- ($ (\tilelen/2+\holewidth,0) $)
  -- ($ (\tilelen/2+\holewidth+\holetip,-\holeheight) $)
  -- ($ (\tilelen/2-\holewidth-\holetip,-\holeheight) $)
  -- cycle;%
\end{tikzpicture}%
}}%
}

\newcommand{\glangle}{\textcolor{gray}{\bigl\langle}}
\newcommand{\grangle}{\textcolor{gray}{\bigr\rangle}}

\newcommand{\camelplus}[2]{$#1$\texttt{ + }$#2$}
\newcommand{\camelap}[2]{#1\texttt{(}#2\texttt{)}}

\newcommand{\tylrseqimg}[2]{
\vcenter{\hbox{%
    \sbox0{\includegraphics[height=1.6em]{#1}}%
    \usebox0
    \kern-\wd0
    \makebox[\wd0]{%
      \hfil
      \smash{\raisebox{2.4em}{{#2}}}
      \hfil
    }%
  }}
}

\newcommand{\tylrseqimgtwosm}[2]{
\vcenter{\hbox{%
    \sbox0{\includegraphics[height=2.8em]{#1}}%
    \usebox0
    \kern-\wd0
    \makebox[\wd0]{%
      \hfil
      \smash{\raisebox{3.4em}{{#2}}}
      \hfil
    }%
  }}
}

\newcommand{\tylrseqimgtwo}[2]{
\vcenter{\hbox{%
    \sbox0{\includegraphics[height=3.2em]{#1}}%
    \usebox0
    \kern-\wd0
    \makebox[\wd0]{%
      \hfil
      \smash{\raisebox{3.8em}{{#2}}}
      \hfil
    }%
  }}
}

\newcommand{\tylrseqarrow}[1]{\hspace{5pt}\xrightarrow{#1}\hspace{5pt}}

\renewcommand{\sortname}[1]{\textsf{#1}}

\newcommand{\ghazel}{\gram_{\textsf{HZ}}}
\newcommand{\hhazel}{\cfg_{\textsf{HZ}}}

\newcommand{\mc}[1]{\textcolor{mpurp}{#1}}
\newcommand{\ec}[1]{\textcolor{eorange}{#1}}

\newcommand{\mlangle}{\mc{\bigl\langle}}
\newcommand{\mrangle}{\mc{\bigr\rangle}}

\newcommand{\imghole}{\includegraphics[height=\fontcharht\font`\A]{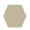}}
\newcommand{\imginfixhole}{
\includegraphics[height=\fontcharht\font`\A]{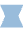}
}
\newcommand{\imgsortx}{\includegraphics[height=\fontcharht\font`\A]{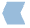}}

\newcommand{\tightencolor}{eorange}
\newcommand{\subsumecolor}{mpurp}

\begin{document}

\title{Syntactic Completions with Material Obligations}
\titlenote{A version of this paper with errata and a complete appendix is available at \url{https://arxiv.org/abs/2508.16848}.}

\author{David Moon}
\orcid{0000-0002-1081-2235}
\affiliation{%
  \institution{University of Michigan}
  \city{Ann Arbor}
  \country{USA}
}
\email{dmoo@umich.edu}

\author{Andrew Blinn}
\orcid{0000-0001-6938-7379}
\affiliation{%
  \institution{University of Michigan}
  \city{Ann Arbor}
  \country{USA}
}
\email{blinnand@umich.edu}

\author{Thomas J. Porter}
\orcid{0009-0000-1809-8382}
\affiliation{%
  \institution{University of Michigan}
  \city{Ann Arbor}
  \country{USA}
}
\email{thomasjp@umich.edu}

\author{Cyrus Omar}
\orcid{0000-0003-4502-7971}
\affiliation{%
  \institution{University of Michigan}
  \city{Ann Arbor}
  \country{USA}
}
\email{comar@umich.edu}


\begin{abstract}
Code editors provide essential services that help developers understand,
navigate, and modify programs. However, these services often fail in the
presence of syntax errors. Existing syntax error recovery techniques, like panic
mode and multi-option repairs, are either too coarse, \eg in deleting large
swathes of code, or lead to a proliferation of possible completions.
This paper introduces \tylr, an error-handling parser and editor generator
that completes malformed code with \emph{syntactic obligations}
that abstract over many possible completions.
These obligations generalize the familiar notion of holes in structure editors
to cover missing operands, operators, delimiters, and sort transitions.

\tylr is backed by a novel theory of tile-based parsing, conceptually organized
around a \emph{molder} that turns tokens into tiles and a \emph{melder}
that completes and parses tiles into terms using an error-handling
generalization of operator-precedence parsing.
We formalize melding as a parsing calculus, \tylrcore, that completes input tiles
with additional obligations such that it can be parsed into a well-formed term,
with success guaranteed over all inputs.
We further describe how \tylr implements molding and completion-ranking using
the principle of \emph{minimizing obligations}.

Obligations offer a useful way to scaffold internal program representations,
but in \tylr we go further to investigate the potential of
\emph{materializing} these obligations visually to the programmer.
We conduct a user study to evaluate the extent to which an editor like \tylr that materializes syntactic obligations might be usable and useful, finding both points of positivity and interesting new avenues for future work.
\end{abstract}



\maketitle

\section{Introduction} \label{sec:intro}

Programmers rely on editor services like syntax highlighting, code completion, and go-to-definition for help with comprehending, modifying, and navigating programs in various stages of completion.
These services require analyzing the syntactic structure of the program being edited.
The problem is that, for typical grammars, most textual edit states do not successfully parse \cite{yoonLongitudinalStudyProgrammers2014,omarHazelnutBidirectionallyTyped2017}. For example, consider the malformed program in \autoref{fig:syntax-error}A, written in an OCaml-like expression language (that uses parentheses for function application, rather than spaces).
A naively implemented parser would simply stop upon encountering the first
unexpected token, \code{p2}, at the end of the first line, leaving the program
unparsed and unanalyzable by downstream editor services.
To address this problem, modern parsers attempt to recover from and continue
parsing around malformed syntax, but often leave much to be desired.
Let us consider how contemporary methods might recover from the unexpected token \code{p2}.

A simple recovery method known as ``panic mode'' \cite{ahoCompilersPrinciplesTechniques2007,gruneParsingTechniquesPractical2008} drops tokens heuristically around the error until parsing can resume---in this case, as shown in \autoref{fig:syntax-error}B, a simple panicking parser might drop the first four lines of code because of the various parse errors on those lines, then perhaps recover more granularly on the final line by ignoring the dangling minus sign.
While better than nothing and relatively easy to implement, this approach is liable to ignore large windows around error locations, leaving the programmer with limited or incorrect assistance where they may need it most~\cite{diekmannDonPanicBetter2020}. For example, the dropped lines in \autoref{fig:syntax-error}B would lead
a type error reporting service to mark the uses of \code{x1}, \code{y1}, and \code{y2} unbound (in contrast to what a human would likely conclude).

More sophisticated error recovery methods consider a range of possible repairs around the error location.
\autoref{fig:syntax-error}C shows some possible repairs for the first line of code in \autoref{fig:syntax-error}A.
The first three repairs show how this method reduces dropped input (\ie deleted tokens) compared to a panicking parser.
On the other hand, the next four completion-only repairs show how this method must enumerate all tokens that play similar structural roles---in this case, infix operators on patterns---which can lead to a combinatorial explosion as additional insertions, deletions, and larger repair windows are considered
\cite{considineSyntaxRepairIdempotent,diekmannDonPanicBetter2020}.
To select from these possible repairs, most parsers use \emph{ad hoc} heuristics \cite{grahamPracticalSyntacticError1975,fischerLocallyLeastCostLRError1979}. The heuristically chosen repair is generally not communicated to the programmer, leaving them only indirect clues in the behavior of downstream editor services.
Having the programmer disambiguate between possible repairs interactively can cause information overload due to the number of possible completions \cite{koFrameworkMethodologyStudying2005}.


\begin{figure}[t]
  \begin{minipage}[t]{0.57\textwidth}
    \includegraphics[width=\textwidth]{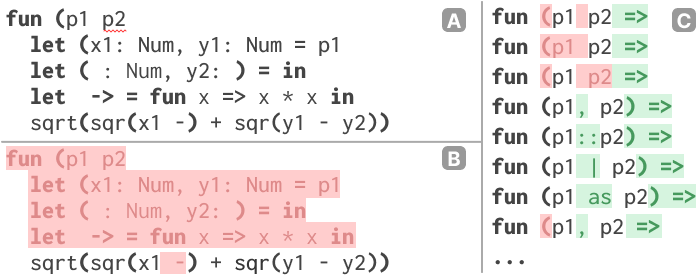}
  \end{minipage}\hfill
  \begin{minipage}[t]{0.4\textwidth}
    \vspace*{-3.2cm}
    \caption{(A) A malformed editor state, with the first unexpected token underlined in red. (B) Regions skipped by a simple panicking parser, highlighted in red. (C) Some possible textual repairs for the first line generated by a conventional error-correcting parser.}
    \label{fig:syntax-error}
  \end{minipage}
\end{figure}

This paper introduces \tylr, a parser and editor generator that performs syntax repair by \emph{syntactic completion} (but not deletion) in a grammar enriched with \emph{obligations}.
This approach can be used to determine the internal program representation within an editor or language server (and indeed we expect this to be a common application of these ideas), but in this paper we go further to investigate the potential of \emph{materializing} these obligations visually to the programmer.

\begin{wrapfigure}{r}{0.55\textwidth}
  \vskip-1em
  \centering
  \includegraphics[width=0.55\textwidth]{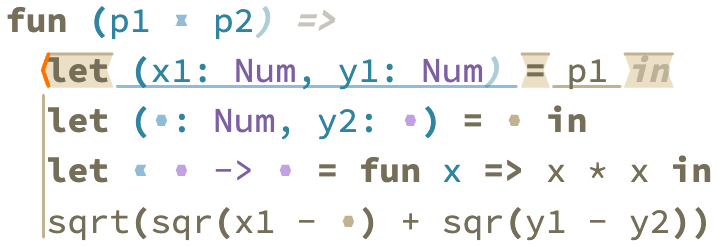}
  \vspace{-10pt}
    \caption{The syntactic completion of the program from \autoref{fig:syntax-error} in \tylr, our tile-based editor.
    }
    \vspace{-10pt}
    \label{fig:obligations}
  \end{wrapfigure}

The screenshot in \autoref{fig:obligations} shows how \tylr repairs the program from \autoref{fig:syntax-error}A.
On Line 1, \tylr completes the user-inserted tokens \code{fun (p1 p2} by materializing three obligations. The first obligation, \imginfixhole, is an \emph{infix obligation}, i.e. it ranges over many possible infix operators in pattern position. The remaining two obligations, \ghostcode{)} and \ghostcode{=>}, as well as those on Line 2, are \emph{ghost obligations} that serve to complete the partially written syntactic forms.
The user can accept these suggested locations by placing the cursor on these faded out tokens and pressing the \texttt{Tab} key or typing over them explicitly. If they wish to place them elsewhere, the ghost obligations can be ignored and the user can type these obligations elsewhere; the ghost obligations are removed when no longer necessary.
On Lines 3 and 5, the user has omitted operands of various syntactic sorts: one pattern, one type, and two expressions. \tylr materializes \emph{operand obligations} (a.k.a. \emph{holes}), written \imghole{} and colored by syntactic sort, to stand for the missing operands. Finally, on Line 4, \imgsortx, is a \emph{sort transition obligation} that indicates that there is a missing transition from the pattern sort (in blue) to the type sort (in purple), because in this language the token \code{->} can appear only in types. This collection of obligations captures the different ways a program may be incomplete. We expand on this visual taxonomy with additional examples from the user's perspective in \autoref{sec:design-overview}.

Each token and obligation in \tylr has a color and a shape, collectively a
\emph{mold}. We refer to a token or obligation equipped with a mold as a
\emph{tile}. In particular, the color indicates the syntactic sort of term being
considered. The shape indicates the hierarchical relationship between a token
and its neighbors. For example, the convex tip on the left of the \code{let}
token in \autoref{fig:obligations} indicates that it is the beginning of a term,
and the concave tip on its right indicates that a child term is expected to its
right. Tips are visualized only for the term at the cursor (which is shown as a
red angle in \autoref{fig:obligations} to conform to the shape of its adjacent
token) to avoid the visual clutter associated with nested block-based
visualizations like those in systems like Scratch
\cite{maloneyScratchProgrammingLanguage2010,cohenCodeStyleSheets2025}.




Underlying this visual taxonomy is a novel theory of parsing that we dub
\emph{tile-based parsing}. Tile-based parsing departs from the predominant
item-based approach of the LL/LR methods and instead builds on
the token-based perspective of operator-precedence (OP) parsing as first
described by \citet{floydSyntacticAnalysisOperator1963}.
OP parsing enjoys the \emph{bounded context} property
\cite{gruneParsingTechniquesPractical2008} that makes it possible to
maximally parse any subrange of input knowing only its single-token delimiters,
an attractive property for modeling and analyzing program edit states.
On the other hand, OP parsing is also known for its limited grammar class, owing
to difficulties reusing the same token in different structural roles
(\eg \code{-} for both infix subtraction and unary negation).
With tile-based parsing, we propose splitting the overall problem of parsing
into a top-down, context-dependent \emph{molder} that molds tokens into tiles,
thereby distinguishing one structural use of a token from another;
and a bottom-up, bounded-context \emph{melder} of tiles that
extends OP parsing with obligation-based syntactic completions.

Where error handling is typically an afterthought in existing parsing methods,
it emerges in tile-based parsing as a natural generalization of the core
OP parsing method.
In particular, we generalize the single-step precedence comparisons
in OP parsing to multi-step precedence \emph{walks} in melding, where the
intermediate steps between the comparands constitute possible completions
between them.
In \autoref{sec:tylrcore}, we precisely specify melding as a parsing calculus
called \tylrcore.
In addition to precedence walks, we describe in \autoref{sec:injecting-grout} how \tylrcore ``injects''
the given grammar with additional grout forms that buffer
the various inconsistencies that may arise between bottom-up reductions
and top-down expectations.
Along the way, we present in \autoref{sec:elab-prec} a new parser-independent
semantics for precedence annotations that generalizes and unifies prior
accounts.

\tylrcore describes a nondeterministic parser of tiles, leaving many decisions
up to the implementation regarding how tiles are molded and completions
are chosen.
In \autoref{sec:implementation}, we describe the principle of \emph{minimizing
obligations} and additional heuristics that guide these decisions in \talltylr.
Finally, we evaluate our overall design with a user study in
\autoref{sec:user-study}, investigating both code insertion and code
modification tasks.
We discover our design of materialized obligations has both promise and demand,
but more design work is needed to give the programmer more control over their
placement and removal, especially when modifying existing code.

\section{Design Overview} \label{sec:design-overview}
We begin with a user-facing summary of how \tylr operates in various common editing scenarios that demonstrate each form of obligation and how it is materialized to the user.

\tylr is a parser and editor generator, i.e. it can be instantiated with various
grammars. In this section, we will write programs in a simple
expression-oriented programming language, Hazel \cite{omarSemanticFoundationsProgram2017}. Hazel is a near-subset of
OCaml. One notable deviation is the use of postfix parentheses,
$\camelap{e}{e}$, instead of infix space, $e~e$, for function application. This
allows us to demonstrate how \tylr handles adjacency when whitespace is not
accepted by the grammar as an infix operator.

\subsection{Operand Obligations}
\begin{figure}[h]
\vspace{-0.4em}
\[
  \tylrseqimg{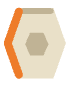}{(a)}
    \tylrseqarrow{\keypress{2}}
  \tylrseqimg{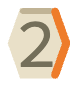}{(b)}
    \tylrseqarrow{\keypress{\textvisiblespace}}
  \tylrseqimg{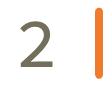}{(c)}
    \tylrseqarrow{\keypress{+}}
  \tylrseqimg{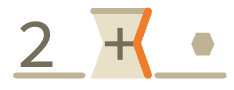}{(d)}
    \tylrseqarrow{\keypress{\textvisiblespace}}
  \tylrseqimg{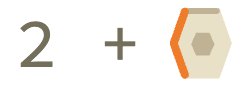}{(e)}
    \tylrseqarrow{\keypress{3}\keypress{\textvisiblespace}\keypress{*}}
  \tylrseqimg{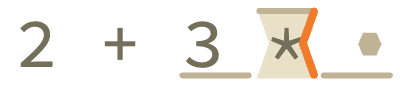}{(f)}
\]
\vspace{-2em}
\caption{Basic expression insertion in \tylr, demonstrating operand obligations and term decorations.}
\label{fig:operand-obligations}
\end{figure}
\vspace{-.8em}

\noindent
We begin with an empty editor buffer in \autoref{fig:operand-obligations}(a). The root sort of Hazel is \sortname{expression}, and no expression has been entered, so \tylr repairs the empty buffer to a single \emph{operand obligation}, or simply \emph{hole}, of that sort. Holes have convex tips on both sides,  and the user's caret (in red) appears angled when on either side of the hole to emphasize its shape.

We next type the character \keypress{2}, which causes the hole to be ``filled''
with the number literal \code{2} in \autoref{fig:operand-obligations}(b). Atomic
operands also have convex tips on both sides. The sort (here,
\sortname{expression}) and the shape, i.e. the convexity of the tips on either
side, are collectively called a \emph{mold} and a token or obligation equipped
with a mold is called a \emph{tile}. Visually, the editor indicates the sort of
a tile using color (here, expressions are grey)
and the shape as shown above when the caret is on the tile.
Next, we type a space, \keypress{\textvisiblespace}, causing a space character to be inserted and the caret to shift right in \autoref{fig:operand-obligations}(c). The caret is no longer on a tile, so no visual indications appear and the caret straightens out. Note that \tylr only supports whitespace-insensitive grammars as of this writing.

Next, we type \keypress{+}. In the Hazel grammar, the \code{+} token is only used as an infix operator, so the molder assigns it a shape with concave tips on both sides, as shown visually in \autoref{fig:operand-obligations}(d).
\tylr must then perform syntax repair, because \code{2 +} is not accepted by the grammar. To do so, \tylr performs a \emph{precedence walk}. We will describe precedence walks precisely later in the paper, but for now, let us develop some intuition. The idea is that we need to walk from the current token, \code{+}, to the following token, which in this case is an implicitly included end-of-buffer token. The only walk which allows this is one that traverses the right operand of the form \camelplus{e}{e}. Consequently, \tylr repairs the syntax by inserting an expression-sorted operand obligation, i.e. hole, as shown. For convenience, \tylr also automatically inserts the space between the operator and the hole. When we subsequently type \keypress{\textvisiblespace}, this automatically inserted space is ``consumed'' rather than causing the insertion of a second space, as shown in \autoref{fig:operand-obligations}(e).

Finally, we continue typing as shown, resulting in \autoref{fig:operand-obligations}(f). Operator sequences are parsed according to Hazel's precedences and associativities. Notice in both \autoref{fig:operand-obligations}(d) and \autoref{fig:operand-obligations}(f) that \tylr underlines the associated operands when the caret is on a tile to visually communicate the structure of the overall term. Notice also that completed terms are always convex on both sides. Indeed, a user's mental model can simply be that \tylr inserts obligations to maintain visual convexity.

\subsection{Infix Obligations}
\noindent
Starting from the editor state in \autoref{fig:infix-obligations}(a),
we press backspace, $\backdel$, deleting the \code{+} tile. Textually, this would result in the operands \code{2} and \code{3} appearing adjacent to one another, which is not accepted by the Hazel grammar. There are many possible walks between adjacent terms---one for each of the infix operators---so \tylr abstracts over them by inserting an \emph{infix obligation}, a.k.a. an \emph{operator hole}, as shown in \autoref{fig:infix-obligations}(b). Infix obligations have the lowest precedence.

\begin{wrapfigure}[6]{r}{0.47\textwidth}
  \centering
  \vskip-.4em
  $\tylrseqimg{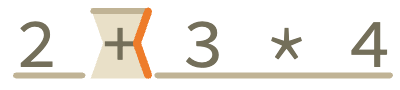}{(a)}
  \tylrseqarrow{\keypress{\backdel}}
  \tylrseqimg{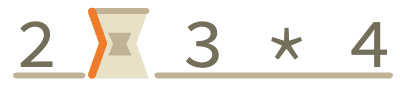}{(b)}$
  \vspace{-.8em}
  \caption{Adjacent operands are connected by infix obligations in \tylr.}
  \label{fig:infix-obligations}
\end{wrapfigure}
One way to think about this mechanism is that operand obligations arise when one term is expected but zero terms appear, whereas infix obligations arise when one term is expected but many adjacent terms appear, as is often the case transiently during edits.

\subsection{Molding Ambiguity}
\noindent
\begin{wrapfigure}[6]{r}{0.4\textwidth}
  \vskip-.4em
  \centering
      $\tylrseqimg{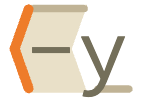}{(a)}
      \tylrseqarrow{\keypress{x}\keypress{\textvisiblespace}}
      \tylrseqimg{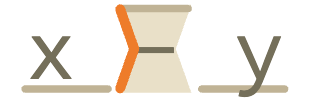}{(b)}$
  \vspace{-.6em}
  \caption{The minus sign has multiple molds. The mold is chosen to minimize obligations.}
  \label{fig:molding-ambiguity}
\end{wrapfigure}
The situation becomes more interesting if we use the \code{-} token, because in the Hazel grammar this token can appear both as an infix operator (subtraction) and as a prefix operator (negation). These correspond to different molds. For example, in \autoref{fig:molding-ambiguity}(a), the \code{-} token before \code{y} is molded as a prefix operator, visualized with a convex tip on the left and a concave tip on the right as shown.

If in this position, we type \keypress{x} followed by \keypress{\textvisiblespace} (to move the caret over the automatically inserted space), \tylr remolds the token into the corresponding infix operator as shown in \autoref{fig:molding-ambiguity}(b). The reason is \tylr's novel  approach to disambiguation: \tylr always selects the mold which locally minimizes the number of obligations that must be inserted. Retaining the prefix mold would have required also inserting an infix obligation, whereas the infix mold requires no obligations.

Although this approach requires considering alternative token moldings as tokens are encountered, we note that there are generally only a few tokens in a typical grammar which can have multiple possible moldings. In the Hazel grammar, only \code{-} and \code{(} have this property. Traditional operator precedence parsing cannot handle such grammars, but using a molder separate from the core parsing algorithm that makes decisions based on repair costs allows us to overcome this expressiveness limitation while retaining a relatively simple core parsing algorithm.

Note that formally, a mold is not simply a shape and sort, but rather a zipper into the grammar, so the molder is also responsible for resolving other parsing ambiguities that might arise as well, e.g. the famous ``dangling else'' problem in imperative languages. This is in contrast to approaches where the parser resolves these ambiguities, e.g. by favoring shifts over reduces. We leave to future work the problem of declaratively specifying disambiguation policies in this setting. We only work with unambiguous grammars in the remainder of the paper.



\subsection{Ghost Obligations} \label{sec:ghost-obligations}
\begin{figure}[h]
\vskip0pt
\[
    \tylrseqimgtwo{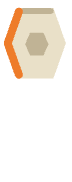}{(a)}
    \tylrseqarrow{\keypress{l}\keypress{e}\keypress{t}\keypress{\textvisiblespace}}
    \tylrseqimgtwo{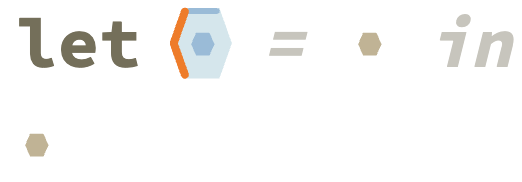}{(b)}
    \tylrseqarrow{\keypress{x}\keypress{\textvisiblespace}\keypress{= \textrm{or} \tab}}
    \tylrseqimgtwo{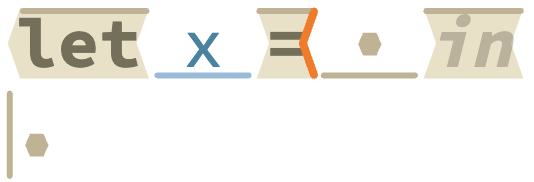}{(c)}
\]
\vspace{-1.4em}
\caption{Ghost obligations are inserted for mixfix forms in \tylr.}
\label{fig:ghost-obligations-let-insertion}
\end{figure}
\noindent
So far, our examples have only used infix operators. Introducing \emph{mixfix operators} requires enriching our language of obligations to handle mixfix delimiters that have not yet been inserted.

For example, starting from an empty buffer in \autoref{fig:ghost-obligations-let-insertion}(a), we can insert a let expression by  typing $\keypress{l}\keypress{e}\keypress{t}\keypress{\textvisiblespace}$. This causes insertion of \emph{ghost obligations}, shown in gray in \autoref{fig:ghost-obligations-let-insertion}(b). These obligations are again determined by computing a precedence walk from the inserted token to the next token. When walking over a token that is not explicitly in the editor state, we can include it as a ghost obligation in the corresponding completion. We again choose the completion that minimizes obligations.
The user can continue by entering a variable to fill the pattern hole at the caret and, when they reach the \code{=}, they can either press tab, \keypress{\tab}, or type over the ghost character(s), here by entering \keypress{=}. Either choice will result in the state shown in \autoref{fig:ghost-obligations-let-insertion}(c). Note that the term structure is visualized despite the missing delimiter.

When inserting a let expression in the middle of an existing program, \tylr
needs to heuristically decide where to place the ghost obligations. The
heuristic we use is is based primarily on newline placement in the buffer,
summarized by the example in
\autoref{fig:ghost-obligations-obligation-placement}.
If we insert the let expression
\begin{wrapfigure}[16]{r}{0.56\textwidth}
  \centering
  \vskip0.4em
      $\tylrseqimgtwosm{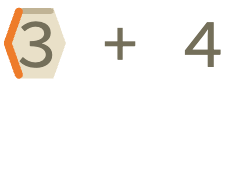}{(a)}
      \tylrseqarrow{\keypress{l}\keypress{e}\keypress{t}\keypress{\textvisiblespace}}
      \tylrseqimgtwosm{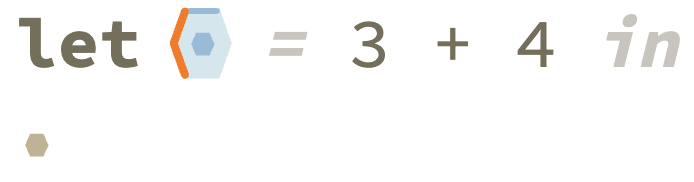}{(b)}$
    \\[1.5em]
        $\tylrseqimgtwosm{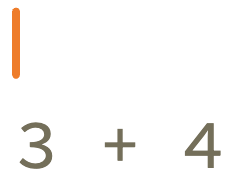}{(c)}
        \tylrseqarrow{\keypress{l}\keypress{e}\keypress{t}\keypress{\textvisiblespace}}
        \tylrseqimgtwosm{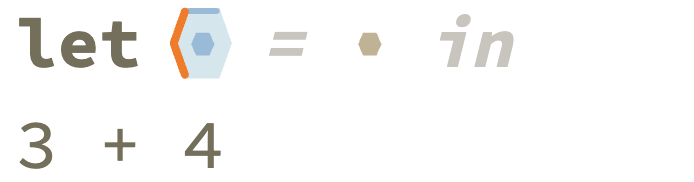}{(d)}$
    \vspace{-0.5em}
    \caption{Ghost obligation placement is chosen heuristically, here based on newline locations.}
    \label{fig:ghost-obligations-obligation-placement}
  \vspace{2em}
      $\tylrseqimgtwosm{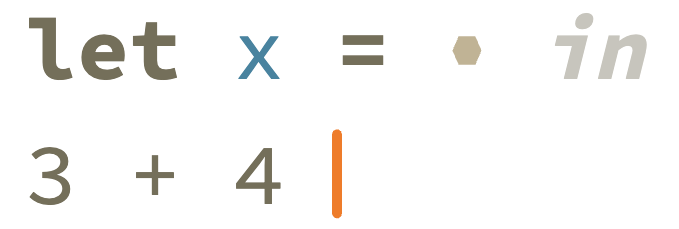}{(a)}
      \hspace{-5pt}
      \tylrseqarrow{\keypress{i}\keypress{n}\keypress{\textvisiblespace}}
      \hspace{-5pt}
      \tylrseqimgtwosm{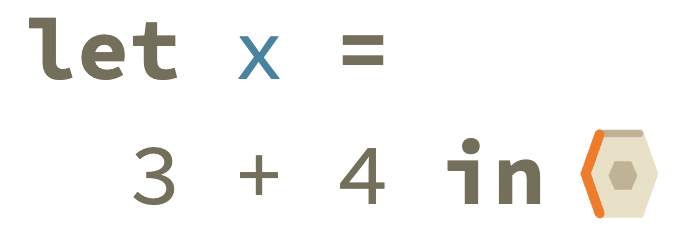}{(b)}$
  \vspace{-0.5em}
  \caption{Ghost obligations can be ignored and are cleaned up if entered elsewhere.}
  \label{fig:ghost-obligations-in}
\end{wrapfigure}
immediately before existing code on the same
line, that code is placed in the first child position of the same sort as shown in \autoref{fig:ghost-obligations-obligation-placement}(a-b). If instead we enter the let expression on a blank line, subsequent lines are placed in the last child position of the same short as shown in \autoref{fig:ghost-obligations-obligation-placement}(c-d).

If the user's intent differs from this heuristic placement, they can ignore the ghost obligations and insert the obligation explicitly where they intend. Given the state in \autoref{fig:ghost-obligations-in}(a), if the user enters \code{in} at the end of the buffer, \tylr would clean up the ghost \code{in} and restructure the code as shown in \autoref{fig:ghost-obligations-in}(b). Note that \tylr automatically manages indentation.


\subsection{Sort Transition Obligations}
\noindent
Some syntactic forms are legal only when entering terms of a particular sort. For example, in Hazel, the arrow operator, \code{->}, can only appear in types. If we enter the arrow in pattern position, as shown in \autoref{fig:sort-transition-obligations}(a), \tylr inserts obligations  indicating that there needs to be a sort transition from the pattern sort to the type sort, as shown in \autoref{fig:sort-transition-obligations}(b). If there were additional text on the right that could be parsed as a pattern, a sort transition ``back'' on the right side would also appear.

\begin{figure}[h]
\vspace{.4em}
\[
    \tylrseqimgtwo{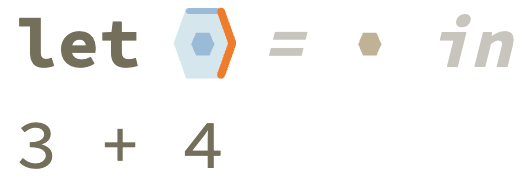}{(a)}
    \tylrseqarrow{\keypress{-}\keypress{>}\keypress{\textvisiblespace}}
    \tylrseqimgtwo{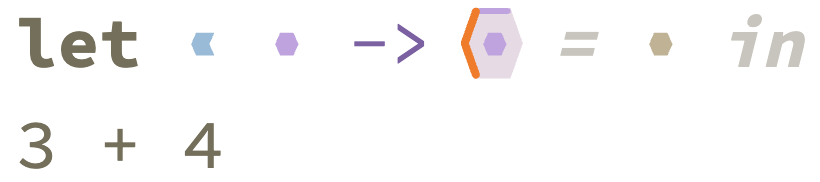}{(b)}
\]
\vspace{-1.2em}
\caption{Sort transition obligations are needed when entering forms that are not sort-correct.}
\label{fig:sort-transition-obligations}
\end{figure}

\subsection{Unmolded Tokens}
\begin{wrapfigure}[3]{r}{0.52\textwidth}
\vskip-1em
\begin{minipage}[t]{0.075\textwidth}
\vskip0pt
$\tylrseqimg{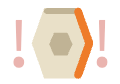}{}$
\end{minipage}
\hfill
\begin{minipage}[t]{0.45\textwidth}
\vskip-12pt
\caption{Unrecognized tokens are left unmolded, and therefore cannot fulfill obligations.}
\label{fig:unmolded-tokens}
\end{minipage}
\end{wrapfigure}
Finally, some tokens are not recognized by the grammar at all. To handle these, the molder marks them as unmolded tokens and treats them like whitespace or comments, i.e. they do not have a shape and so do not participate in obligation insertion. For example, in Hazel, there is no \code{!} token, so \tylr simply marks it in red and ignores it as shown in \autoref{fig:unmolded-tokens}. Notice that no matter where the \code{!} token appears, the operand obligation remains unfilled.

\newcommand{\chpath}{tylrcore}

\newtheorem{assumption}{Assumption}

\newcommand{\ggtrdot}{\textcolor{gray}{\gtrdot}}
\newcommand{\gvert}{\textcolor{gray}{\vert}}
\newcommand{\glp}{\textcolor{gray}{(}}
\newcommand{\glr}{\textcolor{gray}{)}}
\newcommand{\gplus}[1]{\glp#1\textcolor{gray}{\glr^+}}
\newcommand{\gstr}[1]{\glp#1\textcolor{gray}{\glr^*}}
\newcommand{\gopt}[1]{\glp#1\textcolor{gray}{\glr ?}}
\newcommand{\g}[1]{\textcolor{gray}{#1}}

\newcommand{\gbl}{\textcolor{gray}{\{}}
\newcommand{\gbr}{\textcolor{gray}{\}}}
\newcommand{\gsc}{\textcolor{gray}{;}}

\section{\tylrcore} \label{sec:tylrcore}

We now present an error-handling parser calculus, called \tylrcore,
that describes how to complete token sequences with additional tokens
such that they can be parsed into grammatical terms.
Given a language grammar,
whose terminal symbols we call \emph{tiles},
we ``inject'' it with additional \emph{grout} forms
that either stand in for missing terms
or else wrap sort-inconsistent and extraneous terms.
From this grout-injected grammar,
we generate an error-handling parser of tile sequences
that completes its input with grout and additional requisite tiles
(manifesting as ghosts in \talltylr (\autoref{sec:ghost-obligations}))
such that it can be parsed.
By first relaxing grammaticality with grout,
we ensure that the generated tile parser is total over all inputs
(\autoref{theorem:totality}).

As a substrate for these ideas, we generalize and unify two prior accounts
of operator precedence:
Aasa's semantics for precedence annotations in grammars
\cite{aasaPrecedencesSpecificationsImplementations1995}
and Floyd's seminal introduction of operator-precedence parsing (OP parsing)
\cite{floydSyntacticAnalysisOperator1963}.
The two works have related but complementary scopes:
Aasa describes how \emph{precedence annotations} act as filters
on the set of valid derivation trees of the underlying grammar,
as well as how to elaborate the annotated grammar into an unannotated one;
meanwhile, Floyd begins with an unannotated grammar and describes
how to derive a set of \emph{precedence relations}
between terminal symbols that can be used to steer a bottom-up parser.
In \autoref{sec:elab-prec},
we specify a new elaboration
from annotated grammars $\gram$ to unannotated grammars $\cfg$
that simplifies Aasa's version
and generalizes it to allow for arbitrary mixfix forms of varying sorts
in $\gram$.
To demonstrate correctness, we show in \autoref{sec:derive-prec}
that Floyd's precedence relations derived from $\cfg$
cohere as expected with the annotations in $\gram$
(\autoref{theorem:prec-cohere}).

Next, we generalize OP parsing to handle errors using completion-only repair.
After reviewing Floyd's original parsing method
and noting the various ways that it can ``go wrong''
in \autoref{sec:op-parsing-errors},
we present our error-handling variation in
\autoref{sec:op-parsing-error-handling}.
Among other things, our approach generalizes
the single-step precedence comparisons between neighboring input tokens
that steer an OP parser to multi-step precedence \emph{walks},
where the intermediate steps constitute a possible completion between the
tokens.

This approach alone is not quite sufficient to guarantee a successful parse
across all grammars and inputs---moreover, in practice, it would require the
parser to make many heuristic choices between structurally identical tile
completions.
To remedy these issues, we describe in \autoref{sec:injecting-grout}
how to inject grout forms into the translated grammar $\cfg$,
which serve as fallbacks when no tile-only completion
exists, and also as natural defaults when there are many possible choices.
Given these fallbacks, we show that the generated parser is sound and total over
all inputs (\autoref{theorem:totality}).

\begin{wrapfigure}[3]{r}{0.3\textwidth}
  \vskip-1em
  \begin{tabular}{rlcl}
    \text{\textcolor{gray}{option}} & $\o{\alpha}$ & \textcolor{gray}{::=}
    & $\none ~\gvert~ \some{\alpha}$ \\
    \text{\textcolor{gray}{sequence}} & $\s{\alpha}$ & \textcolor{gray}{::=}
    & $\snil ~\gvert~ \alpha\s{\alpha}$
  \end{tabular}
\end{wrapfigure}
\paragraph{Notation}
Throughout this section, we will use
the notation on the right for options and sequences given an element type $\alpha$.
Given a judgment form $J~\alpha$,
we will write $J~\o{\alpha}$ to mean
either $\o{\alpha} = \none$
or else $\o{\alpha} = \some{\alpha}$
such that $J~\alpha$ holds.
Similarly, given $J~\alpha~\beta$,
we will write $J~\o{\alpha}~\o{\beta}$ to mean
either $\o{\alpha} = \none$ and $\o{\beta} = \none$
or else $\o{\alpha} = \some{\alpha}$ and $\o{\beta} = \some{\beta}$
such that $J~\alpha~\beta$ holds.




\subsection{Elaborating Precedence Annotations} \label{sec:elab-prec}

\subsubsection{Precedence-Bounded Grammars}
Our calculus is parametrized by a language grammar $\gram$ in EBNF form with
precedence annotations, what we call in this work
a \emph{precedence-bounded grammar (PBG)}.
Compared to ordinary context-free grammars,
where precedence must be encoded
in tedious towers of dependent production rules,
PBGs allow language forms of the same semantic \emph{sort}
(\eg expressions vs patterns vs types)
to be organized under a single named entity,
leading to more natural and concise grammar definitions.
By having the author explicitly specify the language's sorts, PBGs also help us
generate a minimal set of semantically meaningful grout forms.

\begin{figure}
  \centering
  ${}$\hfill
  \begin{minipage}[b]{0.54\textwidth}
  \[\arraycolsep=3pt\begin{array}{rlcl}
    \text{\textcolor{gray}{tile}} & \labl & \textcolor{gray}{\in} & \lbls \\
    \text{\textcolor{gray}{sort}} & \sortr,\sort & \textcolor{gray}{\in} & \sorts\supseteq\{\startSym\}  \\
    \text{\textcolor{gray}{symbol}} & \sym,\symy & \textcolor{gray}{::=} & \labl ~\gvert~ \sort \\
    \text{\textcolor{gray}{regex}} & \rgx & \textcolor{gray}{::=} & \rone ~\gvert~ \sym  ~\gvert~ \ralttight{\rgx}{\rgx} ~\gvert~ \rseq{\rgx}{\rgx} ~\gvert~ \rstar{\rgx} \\
    \text{\textcolor{gray}{precedence}} & \precm,\precn,\precp,\precq &
    \textcolor{gray}{\in} & \precs =~ \mathbb{N}\sqcup\{\bot,\top\} \\
    \text{\textcolor{gray}{PBG}} & \gram & \textcolor{gray}{\in} & \sorts \rightarrow \precs \rightarrow \rgx
  \end{array}\]
  \vspace{-0.25cm}
  \caption{Syntax of precedence-bounded grammars}\label{fig:pbg}
  \end{minipage}
  \hfill
  \begin{minipage}[b]{0.365\textwidth}
    \captionsetup{justification=raggedleft}
    \[\arraycolsep=3pt\begin{array}{rlcl}
      \text{\textcolor{gray}{terminal}} & \cft & \textcolor{gray}{::=} & \dlroot~\gvert~
      \labl
      ~\gvert~ \drroot \\
      \text{\textcolor{gray}{nonterminal}} & \cfsr,\cfs & \textcolor{gray}{::=} & \psq \\
      \text{\textcolor{gray}{symbol}} & \cfx & \textcolor{gray}{::=} & \cft ~\gvert~ \cfs \\
      \text{\textcolor{gray}{CFG}} & \cfg & \textcolor{gray}{\in} & \{\cfs\pruduce\s{\cfx}\}
    \end{array}\]
    \vspace{-0.25cm}
    \caption{Syntax of elaborated \\ context-free grammars}\label{fig:cfg}
  \end{minipage}
  \hfill${}$
\end{figure}

\newcommand{\thazel}{\lbls_{\textsf{HZ}}}
\newcommand{\shazel}{\sorts_{\textsf{HZ}}}

\newcommand{\leftarg}{\textcolor[HTML]{67AFA2}{\bddsort{\bot}{\expsort}{2}}}
\newcommand{\oleftarg}{\textcolor[HTML]{67AFA2}{\bddsort{\overline{\bot}}{\expsort}{\overline{2}}}}
\newcommand{\rightarg}{\textcolor[HTML]{9B5377}{\bddsort{3}{\expsort}{\bot}}}
\newcommand{\orightarg}{\textcolor[HTML]{9B5377}{\bddsort{\overline{3}}{\expsort}{\overline{\bot}}}}

\newcommand{\phnt}{\vphantom{\mathstrut}}

\begin{figure}[h]
  \begin{minipage}[b]{0.5\textwidth}
  \vskip 1pt
  \setlength{\tabcolsep}{3pt}
  \setlength{\arraycolsep}{2pt}
  \[\begin{array}{rcl}
    \thazel &=& \left\{
      \spretile{let}\,\sbintile{=}\,\sbintile{in}
      \,...\,
      \spretile{(}\,\sbintile{,}\,\sposttile{)}
    \right\} \\[1mm]
    &\sqcup& \left\{
      \sbintile[\pc]{:}\,
      \vartile[\pc]{var}\,
      \spretile[\pc]{(}\,\sbintile[\pc]{,}\,\sposttile[\pc]{)}
    \right\} \\[1mm]
    &\sqcup& \left\{
      \soptile[\yc]{Num}\,
      \spretile[\yc]{(}\,\sbintile[\yc]{,}\,\sposttile[\yc]{)}
    \right\}
    \hspace{.9cm}
    \shazel = \{\startsort{\expsort}, \patsort, \typsort\}
  \end{array}\]
  \vspace{8pt}

  $\ghazel = \left\{
    \begin{array}{rclrl}
      \expsort & \mapsto & 0 & \mapsto & \spretile{let}
      {\rand}\,\patsort\,{\rand} \sbintile{=} {\rand}\,\expsort\,{\rand}
      \sbintile{in} {\rand}\,\expsort \\
      & & 1^{\succ} & \mapsto &
      \expsort\,{\rand}\,(\sbintile{+} \ror \sbintile{-})\, {\rand}\,\expsort \\
      & & 2^{\succ} & \mapsto &
      \expsort\,{\rand}\,(\sbintile{*} \ror \sbintile{/})\, {\rand}\,\expsort \\
      & & 3 & \mapsto &
      \vartile{num} \ror \vartile{var} \\
      & & & & \ror
      \spretile{(} {\rand}\,\expsort\,{\rand}\, ( \sbintile{,} {\rand}\,\expsort
      )^* {\rand} \sposttile{)} \\[.5em]
      \patsort & \mapsto & 0 & \mapsto & \patsort\,{\rand}\sbintile[\patcolor]{:}{\rand}\,\typsort \\
      & & 1 & \mapsto & \vartile[\patcolor]{var} \ror \spretile[\patcolor]{(} {\rand}\,\patsort\,{\rand}\, ( \sbintile[\patcolor]{,} {\rand}\,\patsort )^* {\rand} \sposttile[\patcolor]{)} \\[.5em]
      \typsort & \mapsto & 0 & \mapsto & \soptile[\typcolor]{Num} \ror \spretile[\typcolor]{(} {\rand}\,\typsort\,{\rand}\, ( \sbintile[\typcolor]{,} {\rand}\,\typsort )^* {\rand} \sposttile[\typcolor]{)}
    \end{array}
  \right\}$
  \vskip 1pt
  \caption{A PBG $\ghazel$ for a small expression-oriented language.
  Sorts consist of expressions ($\expsort$) in grey, patterns
  ($\patsort$) in blue, and types ($\typsort$) in purple.
  Tiles are distinguished by text, shape, and color-coded sort.}
  \label{fig:pbg-example}
  \end{minipage}
  \hfill
  \begin{minipage}[b]{0.44\textwidth}
    \vskip -2pt
    \tikzset{nt/.style={
      inner sep=2,
      outer sep=0,
    }}
    \begin{tikzpicture}[
      >={Classical TikZ Rightarrow[]},
    ]
      \node (botbot) [nt]
      at (0, 0) {
        \phnt
        $\botsort{\expsort}$
      };
      \node (botbot') [right=1.1 of botbot] {
        \phnt
        $\bddsort{\bot}{\expsort}{1}\sbintile{+}\bddsort{2}{\expsort}{\bot} ~\gvert~
        \bddsort{\bot}{\expsort}{2}\sbintile{*}\bddsort{3}{\expsort}{\bot}$};

      \node (botone) [nt, below=.6 of botbot] {
        \phnt
        $\bddsort{\bot}{\expsort}{1}$
      };
      \node (botone') [right=1.1 of botone] {
        \phnt
        $\bddsort{\bot}{\expsort}{1}\sbintile{+}\bddsort{2}{\expsort}{1}
        ~\gvert~
        \bddsort{\bot}{\expsort}{2}\sbintile{*}\bddsort{3}{\expsort}{1}$
      };

      \node (bottwo) [nt, below=.6 of botone] {
        \phnt
        $\bddsort{\bot}{\expsort}{2}$
      };
      \node (bottwo') [right=1.1 of bottwo] {
        \phnt
        $\bddsort{\bot}{\expsort}{2}\sbintile{*}\bddsort{3}{\expsort}{2}$
      };

      \node (bottop) [nt, below=.6 of bottwo] {
        \phnt
        $\bddsort{\bot}{\expsort}{\top}$
      };

      \node (twobot) [nt, below right=0 and .35 of botbot]
      {\vphantom{\mathstrut}$\bddsort{2}{\expsort}{\bot}$};
      \node (twobot') [right=.85 of twobot] {
         \phnt
         $\bddsort{2}{\expsort}{2}\sbintile{*}\bddsort{3}{\expsort}{\bot}$
      };

      \node (threebot) [nt, below=.6 of twobot] {
        \phnt
        $\bddsort{3}{\expsort}{\bot}$
      };

      \node (topbot) [nt, below=.6 of threebot] {
        \phnt
        $\bddsort{\top}{\expsort}{\bot}$
      };
      \node (topbot') [right=.85 of topbot] {
        \phnt
        $\spretile{let}\bddsort{\bot}{\patsort}{\bot}\sbintile{=}\bddsort{\bot}{\expsort}{\bot}\sbintile{in}\bddsort{1}{\expsort}{\bot}$
      };

      \node (toptop) [nt, below=.6 of topbot] {
        \phnt
        $\topsort{\expsort}$
      };
      \node (toptop') [right=.85 of toptop] {};
      \node (toptop'') [above=.15 of toptop', anchor=west] {
        \phnt
        $\vartile{num}
        ~\gvert~
        \spretile{(}\botsort{\expsort}\sposttile{)}$
      };
      \node (toptop''') [below=.1 of toptop', anchor=west] {
        \phnt
        $\spretile{(}\botsort{\expsort}\sbintile{,}\botsort{\expsort}\sposttile{)}
        ~\gvert~ \dots$
      };

      \node (pbotbot) [nt, below=0.85 of bottop] {
        \phnt
        $\bddsort{\bot}{\patsort}{\top}$
      };
      \node (pbotbot') [right=1.1 of pbotbot] {
        \phnt
        $\bddsort{\bot}{\patsort}{1}\,\sbintile[\patcolor]{:}\,\botsort{\typsort}$};

      \node (ptoptop) [nt, below right=0 and .35 of pbotbot]
      {\vphantom{\mathstrut}$\topsort{\patsort}$};
      \node (ptoptop') [right=.85 of ptoptop] {
          \phnt
          $\vartile[\patcolor]{var}
          ~\gvert~
          \spretile[\patcolor]{(} \botsort{\patsort} \sposttile[\patcolor]{)}
          ~\gvert~ \dots$
      };


      \node (ttoptop) [nt, below=.15 of ptoptop]
      {\vphantom{\mathstrut}$\topsort{\typsort}$};
      \node (ttoptop') [right=.85 of ttoptop] {
          \phnt
          $\soptile[\typcolor]{Num}
          ~\gvert~
          \spretile[\typcolor]{(} \botsort{\typsort} \sposttile[\typcolor]{)}
          ~\gvert~ \dots$
      };

      \draw[-implies, double equal sign distance, \subsumecolor, line width=0.7pt] (botbot) -- (botbot');
      \draw[-implies, double equal sign distance, \subsumecolor, line width=0.7pt] (botone) -- (botone');
      \draw[-implies, double equal sign distance, \subsumecolor, line width=0.7pt] (bottwo) -- (bottwo');
      \draw[-implies, double equal sign distance, \subsumecolor, line width=0.7pt] (twobot) -- (twobot');
      \draw[-implies, double equal sign distance, \subsumecolor, line width=0.7pt] (topbot) -- (topbot');
      \draw[-implies, double equal sign distance, \subsumecolor, line
      width=0.7pt] (toptop) -- (toptop');
      \draw[-implies, double equal sign distance, \subsumecolor, line width=0.7pt] (pbotbot) -- (pbotbot');
      \draw[-implies, double equal sign distance, \subsumecolor, line width=0.7pt] (ptoptop) -- (ptoptop');
      \draw[-implies, double equal sign distance, \subsumecolor, line width=0.7pt] (ttoptop) -- (ttoptop');

      \draw[-implies, double equal sign distance, \tightencolor, line width=0.7pt] (botbot) -- (botone);
      \draw[-implies, double equal sign distance, \tightencolor, line width=0.7pt] (botone) -- (bottwo);
      \draw[-implies, double equal sign distance, \tightencolor, line width=0.7pt] (bottwo) -- (bottop);
      \draw[-implies, double equal sign distance, \tightencolor, line width=0.7pt] (bottop) -- (toptop.west);

      \draw[-implies, double equal sign distance, \tightencolor, line width=0.7pt] (botbot) -- (twobot.west);
      \draw[-implies, double equal sign distance, \tightencolor, line width=0.7pt] (twobot) -- (threebot);
      \draw[-implies, double equal sign distance, \tightencolor, line width=0.7pt] (threebot) -- (topbot);
      \draw[-implies, double equal sign distance, \tightencolor, line
      width=0.7pt] (topbot) -- (toptop);

      \draw[-implies, double equal sign distance, \tightencolor, line width=0.7pt] (pbotbot) -- (ptoptop.west);
    \end{tikzpicture}
    \vskip -7pt
    \caption{
          An excerpt of the CFG $\hhazel$
          elaborated (\autoref{fig:compile-precedence})
          from $\ghazel$ (\autoref{fig:pbg-example}).
          The production rules are arranged and color-coded
          by whether each is elaborated by \textcolor{\subsumecolor}{reduction}
          or by \textcolor{\tightencolor}{tightening}.
        \label{fig:cfg-example}}
  \end{minipage}
\end{figure}

The syntax of PBGs is given in \autoref{fig:pbg}.
A PBG $\gram$ is a partial function mapping a sort $\sort\in\sorts$ and precedence
$\precp\in\precs$ to a regex $\rgx$ over symbols $\sym$, each either a tile
$\labl\in\lbls$ or a sort $\sort\in\sorts$.
We assume that $\sorts$ includes a designated start sort $\startSym$.
We assume $\precs = \mathbb{N} \sqcup \{\bot,\top\}$ includes
the natural numbers $\mathbb{N}$ for precedence levels assigned in $\gram$
as well as minimum $\bot$ and maximum $\top$ precedence levels that are reserved
for internal use.
We further assume that $\precs$
is equipped with ordering relations $\plt_\sort, \pgt_\sort$
that abstract the details of associativity for each sort $\sort\in\sorts$.
For example, $5 \plt_\expsort 5$ would encode that infix operators at precedence
level $5$ of sort $\expsort$ are right-associative---otherwise, outside
of reflexive pairs, these relations coincide with the usual ordering
relations on natural numbers.
For all sorts $\sort\in\sorts$, we assume
$\precp \prec_\sort \top \succ_\sort \precp$ and
$\precp \succ_\sort \bot \prec_\sort \precp$
for all other $\precp\in\mathbb{N}$.


\autoref{fig:pbg-example} gives a concrete PBG $\ghazel$ encoding an excerpt of
Hazel expressions, patterns, and types, which we will use as a running example
throughout the rest of the paper.
Here, each precedence level $\precp\in\mathbb{N}$ of sort $\sort$ is optionally
tagged with the relation $\otimes\in\{\plt_\sort, \pgt_\sort\}$
that applies to the reflexive pair $\precp \otimes \precp$, if any---in this
case, levels $1^{\pgt}$ and $2^{\pgt}$ of sort $\expsort$ are marked as
left-associative, \ie $1 \pgt_\expsort 1$ and $2 \pgt_\expsort 2$.
Meanwhile, the tiles are distinguished by their shape and color (gray for expressions,
blue for patterns, purple for types) in addition to their text.

A regex $\rgx$ is either $\rone$, matching the empty string;
a symbol $\sym$;
a choice $\ralt{\rgx_L}{\rgx_R}$;
a concatenation $\rseq{\rgx_L}{\rgx_R}$; or
a Kleene star $\rstar{\rgx}$.
\newcommand{\rgxsem}[1]{\llbracket #1 \rrbracket}
Its language $\rgxsem{\rgx}$ of matching symbol strings
$\s{\sym}$ is defined as follows: \\[-1ex]
\begin{minipage}[t]{0.15\textwidth}
  \vskip 0pt
\begin{align*}
  \rgxsem{\rone} ~&=~ \{\snil\} \\
  \rgxsem{\sym} ~&=~ \{\sym\} \\
\end{align*}
\end{minipage}
\hfill
\begin{minipage}[t]{0.45\textwidth}
  \vskip 0pt
  \begin{align*}
    \rgxsem{\ralt{\rgx_L}{\rgx_R}} ~&=~ \rgxsem{\rgx_L} \cup \rgxsem{\rgx_R} \\
    \rgxsem{\rseq{\rgx_L}{\rgx_R}} ~&=~ \{\s{\sym}_L\s{\sym}_R \mid \s{\sym}_L \in \rgxsem{\rgx_L},\, \s{\sym}_R \in \rgxsem{\rgx_R}\} \\
\end{align*}
\end{minipage}
\hfill
\begin{minipage}[t]{0.35\textwidth}
  \vskip 0pt
  \begin{minipage}[t]{0.45\textwidth}
  \vskip 0pt
  \begin{align*}
    \rgxsem{\rstar{\rgx}} ~&=~ \bigcup_{k\in\mathbb{N}} \rgxsem{\rgx^k}
  \end{align*}
  \end{minipage}%
  \hfill%
  \begin{minipage}[t]{0.48\textwidth}
  \vskip 8.5pt
  \vphantom{$\mathstrut$}\text{where $\rgx^0 = \rone$} \\
  \text{and $\rgx^{k+1} = \rseq{\rgx}{\rgx^k$}}
  \end{minipage}
\end{minipage} \\
A \emph{$\gram$-form} is a symbol string $\s{\sym} \in
\rgxsem{\gram(\sort,\precp)}$ for any $\sort\in\sorts, \precp\in\precs$.
To make use of operator-precedence parsing techniques,
we assume that every $\gram$-form is in \emph{operator form}
\cite{floydSyntacticAnalysisOperator1963}:
\begin{assumption}[Operator Form] \label{assum:operator-form}
  There exist no sorts $\sort_L, \sort_R \in \sorts$
  and regex $\gram(\sort,\precp)$
  such that ${\dots\sort_L\sort_R\dots} \in \rgxsem{\gram(\sort,\precp)}$.
\end{assumption}
\noindent In other words, every $\gram$-form may be written in the form
$\o{\sort}_0 \iseq{\labl_i\o{\sort}_{i+1}}{0\leq i \leq k}$.
\citet{greibachNewNormalFormTheorem1965} showed that every CFG
can be normalized to a strongly equivalent one in operator form.

\subsubsection{Context-Free Grammars}
We assign meaning to the precedence-annotated grammar $\gram$
by elaborating it to an unannotated context-free grammar (CFG) $\cfg$,
whose syntax is given in \autoref{fig:cfg}.
An elaborated CFG $\cfg$ is a (possibly infinite) set
of production rules $\cfs \pruduce \s{\cfx}$,
each mapping a nonterminal $\cfs$
to a finite sequence $\s{\cfx}$ of symbols---we will refer
to $\cfs$ and $\s{\cfx}$ as the rule's \emph{producer} and \emph{product}.
Each symbol $\cfx$ is either a terminal $\cft$ or a nonterminal $\cfs$.
A terminal symbol $\cft$ is either a tile $\labl$ or a root delimiter, $\dlroot$ or $\drroot$,
marking the start or end of input.
Meanwhile, a nonterminal is a \emph{precedence-bounded sort} $\psq$,
where $\precp,\precq\in\precs$ will serve as constraints from the left and right
sides of the nonterminal in the overall production tree.
\autoref{fig:cfg-example} shows an excerpt of the CFG $\hhazel$
elaborated from $\ghazel$, the process of which we discuss in \autoref{sec:derive-prec}.

We have specialized the symbols here to serve as our elaboration outputs,
but their rewriting semantics are standard:
given a production rule ${\cfs}\pruduce{\s{\cfx}}$,
we say that the symbol string $\s{\cfx}_L \cfs \s{\cfx}_R$ \emph{produces}
the string $\s{\cfx}_L \s{\cfx}\,\s{\cfx}_R$,
written ${\s{\cfx}_L \cfs \s{\cfx}_R}\produce{\s{\cfx}_L \s{\cfx}\,\s{\cfx}_R}$
reusing the production rule syntax.
A production sequence $\s\cfx_0 \produce \s\cfx_1 \produce \dots$ from the
designated start string
$\s\cfx_0 = \dlroot \botsort{\startSym} \drroot$
is called a \emph{derivation};
the language of a CFG collects all of its derivable strings.

\begin{figure}[b]
  \begin{minipage}[t]{.34\textwidth}
  \vskip 2pt
$\text{\textcolor{gray}{comparison}} ~ \cmpr ~ \textcolor{gray}{::=} ~ \lessdot ~\gvert~ \doteq ~\gvert~ \gtrdot$\\[3ex]
\judgboxy[\hfill]{\cft_L \cmpr_{\o{\cfsr}} \cft_R}{$\cft_L$ compares with
$\cft_R$ \\ (over $\o{\cfsr}$)}
\begin{mathpar}
  \inferrule[Prec-LT]{
    \adj{\cft_L}{\cfs} \\
    \cfs \produce^* \o{\cfsr}\cft_R\dots
  }{\ltrel{\cft_L}{\dsty{\o{\cfsr}}}{\cft_R}} \\

  \inferrule[Prec-EQ]{
    \cfs \produce \dots\cft_L\o{\cfsr}\cft_R\dots
  }{\eqrel{\cft_L}{\dsty{\o{\cfsr}}}{\cft_R}} \\

  \inferrule[Prec-GT]{
    \cfs \produce^* \dots\cft_L\o{\cfsr} \\
    \adj{\cfs}{\cft_R}
  }{\gtrel{\cft_L}{\dsty{\o{\cfsr}}}{\cft_R}}
\end{mathpar}
\caption{Precedence comparisons}
\label{fig:prec-rel-2}
\end{minipage}
\hfill
\begin{minipage}[t]{.64\textwidth}
  \vskip 0pt
  \setlength{\tabcolsep}{1.5pt}
  \begin{tabular}{c|c@{\hskip3pt}cccccccccccc|}
  \cline{1-1}
  \multicolumn{1}{|c|}{\diagbox[leftsep=2.5pt,rightsep=2.4pt]{$\cft_L$}{$\cft_R$}} &
  ${}$\hspace{1.5pt}$\drroot$ & \spretile{let} & \sbintile{=} & \sbintile{in} & \sbintile{+} & \sbintile{*} &
  \spretile{(} & \sposttile{)} & \soptile{$num$} &\sbintile[\patcolor]{:} &  \spretile[\patcolor]{(} &
  \sposttile[\patcolor]{)} & \multicolumn{1}{c}{$\soptile[\typcolor]{Num}$}\\
  \hline
  $\dlroot$ & $\,\doteq$ & $\lessdot$ & & & $\lessdot$ & $\lessdot$ & $\lessdot$ &
  & $\lessdot$ & & & & \\
  \spretile{let} & & & $\doteq$ & & & & & & & $\lessdot$ & $\lessdot$ & & \\
  \sbintile{=} & & $\lessdot$ & & $\doteq$ & $\lessdot$ & $\lessdot$ &
  $\lessdot$ & & $\lessdot$ & & & & \\
  \sbintile{in} & $\gtrdot$ & $\lessdot$ & & $\gtrdot$ & $\lessdot$ & $\lessdot$ &
  $\lessdot$ & $\gtrdot$ & $\lessdot$ & & & & \\
  \sbintile{+} & $\,\gtrdot$ & $\lessdot$ & & $\gtrdot$ & $\gtrdot$ & $\lessdot$ &
  $\lessdot$ & $\,\gtrdot$ & $\lessdot$ & & & & \\
  \sbintile{*} & $\,\gtrdot$ & $\lessdot$ & & $\gtrdot$ & $\gtrdot$ & $\gtrdot$ &
  $\lessdot$ & $\gtrdot$ & $\lessdot$ & & & & \\
  \spretile{(} & & $\lessdot$ & & & $\lessdot$ & $\lessdot$ & $\lessdot$ &
  $\doteq$ & $\lessdot$ & & & & \\
  \sposttile{)} & $\,\gtrdot$ & & & $\gtrdot$ & $\gtrdot$ & $\gtrdot$ & &
  $\gtrdot$ & & & & & \\
  \soptile{$num$} & $\,\gtrdot$ & & & $\gtrdot$ & $\gtrdot$ & $\gtrdot$ & &
  $\gtrdot$ & & & & & \\
  \sbintile[\patcolor]{:} & & & $\gtrdot$ & & & & & & & $\gtrdot$ & & $\gtrdot$ & $\lessdot$ \\
  \spretile[\patcolor]{(} & & & & & & & & & & $\lessdot$ & $\lessdot$ & $\doteq$ &  \\
  \sposttile[\patcolor]{)} & & & $\gtrdot$ & & & & & & & $\gtrdot$ &  & $\gtrdot$ & \\
  \soptile[\typcolor]{Num} & & & $\gtrdot$ & &  &  &  &  &  &  &  & $\gtrdot$ & \\
  \cline{2-14}
  \end{tabular}
\caption{An excerpt of precedence comparisons $\cft_L \cmpr \cft_R$
for $\hhazel$ (\autoref{fig:cfg-example})}
\label{fig:hazel-prec}
\end{minipage}
\end{figure}






\subsubsection{Precedence Comparisons} \label{sec:precedence-comparisons}
Given a CFG, we may generate
a collection of \emph{precedence comparisons}
that classify derivation patterns between neighboring terminals.
Each comparison $\cft_L \odot_{\o{\cfsr}} \cft_R$ means there exists
a derivable string with the substring $\cft_L\o{\cfsr}\cft_R$,
which consists of neighbors $\cft_L, \cft_R$
that are either adjacent ($\o\cfsr = \none$)
or separated ($\o\cfsr = \some{\cfsr}$)
by a nonterminal $\cfsr$.
Floyd's original definition \cite{floydSyntacticAnalysisOperator1963}
does not surface the operator index $\o\cfsr$, whose omission we will later show
contributes to an unsound parsing method (\autoref{sec:invalid-reduction}),
and whose use in our resolution we describe in
\autoref{sec:op-parsing-error-handling}.
Until then, we will similarly
omit it from the notation $\cft_L \odot \cft_R$---%
note this is different from assuming $\o\cfsr = \none$,
which we will always notate explictly $\cft_L \odot_{\none} \cft_R$.

The comparison operator $\odot\in\{\lessdot, \doteq, \gtrdot\}$ indicates
in what relative order the neighbors $\cft_L, \cft_R$
were first produced in the derivation.
\autoref{fig:hazel-prec} shows an excerpt of the precedence comparisons for
$\hhazel$ (\autoref{fig:cfg-example}).
The derivation
$\botsort{\expsort}
\produce \bddsort{\bot}{\expsort}{1} \sbintile{+} \bddsort{2}{\expsort}{\bot}
\produce \bddsort{\bot}{\expsort}{1} \sbintile{+}
  \bddsort{2}{\expsort}{2}\sbintile{*}\bddsort{3}{\expsort}{\bot}$
tells us $\sbintile{+} \lessdot \sbintile{*}$
\big(``$\sbintile{+}$ binds less tightly than $\sbintile{*}$''\big)
since $\sbintile{+}$ is produced before its neighbor $\sbintile{*}$.
Meanwhile, the derivation
$\botsort{\expsort}
\produce \topsort{\expsort}
\produce \spretile{(} \botsort{\expsort} \sposttile{)}$
tells us $\spretile{(} \doteq \sposttile{)}$
\big(``$\spretile{(}$ matches $\sposttile{)}$''\big)
since neighbors $\spretile{(}$ and $\sposttile{)}$ are produced together.
Keep in mind that the written order of arguments $\cft_L, \cft_R$
in each comparison $\cft_L \odot \cft_R$
reflects their sequential order in the derived string,
so we should not generally expect $\cft_L \lessdot \cft_R$
to be equivalent to $\cft_R \gtrdot \cft_L$, nor for $\doteq$ to be symmetric.




\subsubsection{Precedence Elaboration}
\label{sec:derive-prec}
Elaboration turns an annotated PBG $\gram$
into an unannotated CFG $\cfg$
whose nonterminals internalize relevant bounding annotations.
Governing its design is the expectation that
precedence comparisons $\labl_L \odot \labl_R$ between tiles in $\cfg$
mirror, when relevant, numeric comparisons between the tiles' backing
annotations in $\gram$ (\autoref{theorem:prec-cohere}).

\begin{theorem}[Annotation-Comparison Coherence] \label{theorem:prec-cohere}
  For all sorts $\sort$, precedence levels $\precp_L,\precp_R$,
  and tiles $\labl_L, \labl_R$
  such that $\dots\labl_L s \in \rgxsem{\gram(\sort,\precp_L)}$
  and ${\sort\labl_R\dots}\in\rgxsem{\gram(\sort,\precp_R)}$,
  the following equivalences hold:
  \[
    \labl_L \lessdot \labl_R \iff \precp_L \plt_\sort \precp_R
    \hspace{3em}
    \labl_L \gtrdot \labl_R \iff \precp_L \pgt_\sort \precp_R
  \]
\end{theorem}

To motivate our design (\autoref{fig:compile-precedence}),
it is instructive first to consider the issues
with simpler alternatives---in particular, elaborating to nonterminals with
fewer than two bounds.
If we elaborated the PBG $\gram$ to a CFG $\cfg_0$
in which the nonterminals were plain unbounded sorts $\sort$,
the best we could do is a trivial elaboration that
simply ignores the precedence annotations in $\gram$: \\
${}$\hfill
$\cfg_0 \triangleq \{ \sort \pruduce \s{\sym} \mid \s{\sym} \in
\rgxsem{\gram(\sort,\precp)}, \precp \in \precs, \sort\in\sorts\}$
\hfill${}$ \\
$\cfg_0$ allows for problematic derivations like
$\expsort
\produce \glangle \mc{\expsort} \sbintile{*} \expsort \grangle
\produce \glangle \mlangle \expsort \sbintile{+} \expsort \mrangle \sbintile{*}
\expsort \grangle$, problematic because it witnesses the unwanted comparison
$\sbintile{+} \gtrdot \sbintile{*}$ \big(``$\sbintile{+}$ binds more tightly than $\sbintile{*}$''\big).

A better approach---similar in effect to
that of \citet{danielssonParsingMixfixOperators2011}
and of \citet{klintUsingFiltersDisambiguation}---would
use singly-bounded nonterminals $\sort^\precp$
and limit their productions
to $\gram$-forms of equal or stronger precedence:
\hfill$\vphantom{\mathstrut}
\cfg_1 \triangleq \{ \sort^\precp \pruduce \lceil\s{\sym}\rceil^\precp \mid \s{\sym} \in \rgxsem{\gram(\sort,\precq)}, \precp \preceq_\sort \precq, \sort\in\sorts\}
$\hfill{} \\
where $\lceil \s{\sym} \rceil^\precp$ lifts each sort symbol
$\sort\in\s{\sym}$ to some suitably bounded nonterminal.
\begin{wrapfigure}[5]{l}{0.525\textwidth}
(a) $\expsort^\bot
\produce \glangle \mc{\expsort^2} \sbintile{*} \expsort^3 \grangle
\produce \glangle \mlangle \expsort^1 \sbintile{+} \ec{\expsort^2} \mrangle
\sbintile{*} \expsort^3 \grangle$ \\
(b) $\expsort^\bot
\produce \glangle \mc{\expsort^2} \sbintile{*} \expsort^3 \grangle
\produce \glangle \mlangle
  \spretile{let} \patsort^\bot \sbintile{=} \expsort^\bot \sbintile{in} \ec{\expsort^0}
\mrangle \sbintile{*} \expsort^3 \grangle$ \\
(c) $\expsort^\bot
\produce \glangle \expsort^2 \sbintile{*} \mc{\expsort^3} \grangle
\produce \glangle \expsort^2 \sbintile{*} \mlangle
  \spretile{let} \patsort^\bot \sbintile{=} \expsort^\bot \sbintile{in} \expsort^0
\mrangle\grangle$
\end{wrapfigure}%
This approach properly rules out unwanted derivations on the left like
(a) \big(for witnessing $\sbintile{+} \gtrdot \sbintile{*}$\big) and
(b) \big($\sbintile{in} \gtrdot \sbintile{*}$\big),
since the left argument $\mc{\expsort^2}$ of $\sbintile{*}$
cannot produce the $\sbintile{+}$- and $\spretile{let}$-forms
of weaker precedence levels $1$ and $0$.
However, $\cfg_1$ is overly conservative:
it also rules out acceptable derivations like (c) \big($\sbintile{*} \lessdot \spretile{let}$\big).
Ultimately the purpose of precedence annotations is to resolve choices
between different possible reduction orders: given a reduced
child, which of the operators on either side of it should be reduced next as
part of its parent?
Derivations (a) and (b) represent disfavored choices of reducing the left parent
\big($\sbintile{+}$ and $\sbintile{in}$\big) before the right \big($\sbintile{*}$\big)
over the reduced children $\ec{\expsort^2}$ and $\ec{\expsort^0}$, respectively.
On the other hand, (c) has no viable alternative reduction order, since
$\spretile{let}$ cannot parent a child to its left.
In such cases, the precedence annotations need not be consulted.

\begin{figure}
\centering
\begin{minipage}[t]{0.4\textwidth}
  \judgboxy{\cfx\consistent\sym}{CFG symbol $\cfx$ is consistent\\ with PBG symbol $\sym$}
  \begin{mathpar}
    \inferrule{}{\labl\consistent\labl}

    \inferrule{}{\psq\consistent\sort}
  \end{mathpar}
\end{minipage}%
\hfill%
\begin{minipage}[t]{0.52\textwidth}
  \judgbox{{\cfs}\produce{\s{\cfx}}}{Nonterminal $\cfs$ produces
  symbols $\s{\cfx}$}
  \begin{mathpar}
    \inferrule[Produce-Subsume]{
      {\cfs}\reduce{\s{\cfx}}
    }{{\cfs}\produce{\s{\cfx}}}

    \inferrule[Produce-Tighten]{
      \precp \pleq_\sort \precm \\
      \precn \pgeq_\sort \precq
    }{{\psq}\produce{\msn}}
  \end{mathpar}
\end{minipage}

\vspace{1em}
\judgbox{\synelab{\gram}{\s{\cfx}}{\cfs}}{Symbol sequence $\s{\cfx}$ reduces to
nonterminal $\cfs$}
\begin{mathpar}
  \inferrule[PElab-Operand$\;\glp\textcolor{gray}{k\geq0}\glr$]{
    \iseq{\sym_i}{0\leq i\leq k}\in\rgxsem{\gram(\sort,\medsquare)} \\
    \iseq{\cfx_i\consistent\sym_i}{0\leq i\leq k} \\\\
    \sym_0 \neq \sort \neq \sym_k
  }{
    \synelab{\gram}{\iseq{\cfx_i}{0\leq i\leq k}}{\bddsort{\top}{\sort}{\top}}
  }

  \inferrule[PElab-Infix$\;\glp\textcolor{gray}{k\geq0}\glr$]{
    \iseq{\sym_i}{0\leq i\leq k}\in\rgxsem{\gram(\sort,\precm)} \\
    \iseq{\cfx_i\consistent\sym_i}{0\leq i\leq k} \\\\
    \cfx_0 = \bddsort{\precp}{\sort}{\precn_L} \\
    \bddsort{\precn_R}{\sort}{\precq} = \cfx_k \\\\
    \precn_L \pgt_\sort \precm \plt_\sort \precn_R
  }{
    \synelab{\gram}{\iseq{\cfx_i}{0\leq i\leq k}}{\bddsort{\min(\precp,\precm)}{\sort}{\min(\precm,\precq)}}
  }

  \inferrule[PElab-Prefix$\;\glp\textcolor{gray}{k\geq1}\glr$]{
    \iseq{\sym_i}{0\leq i\leq k}\in\rgxsem{\gram(\sort,\precm)} \\
    \iseq{\cfx_i\consistent\sym_i}{0\leq i\leq k} \\\\
    \sym_0 \neq \sort \\
    \bddsort{\precn_R}{\sort}{\precq} = \cfx_k \\\\
    \precm \plt_\sort \precn_R
  }{
    \synelab{\gram}{\iseq{\cfx_i}{0\leq i\leq k}}{\bddsort{\top}{\sort}{\min(\precm,\precq)}}
  }

  \inferrule[PElab-Postfix$\;\glp\textcolor{gray}{k\geq1}\glr$]{
    \iseq{\sym_i}{0\leq i\leq k}\in\rgxsem{\gram(\sort,\precm)} \\
    \iseq{\cfx_i\consistent\sym_i}{0\leq i\leq k} \\\\
    \cfx_0 = \bddsort{\precp}{\sort}{\precn_L} \\
    \sort \neq \sym_k \\\\
    \precn_L \pgt_\sort \precm
  }{
    \synelab{\gram}{\iseq{\cfx_i}{0\leq i\leq k}}{\bddsort{\min(\precp,\precm)}{\sort}{\top}}
  }
\end{mathpar}
\caption{Bidirectional elaboration of production $\cfs\produce\s{\cfx}$ and
reduction $\cfs\reduce\s{\cfx}$ rules for CFG $\cfg$ from PBG $\gram$} \label{fig:compile-precedence}
\end{figure}

To account properly for these left- and right-sided concerns,
our elaborated grammar $\cfg$ features nonterminals $\cfs = \psq$
with separate precedence bounds $\precp$ and $\precq$ on either side.
Uniquely to this work, we interpret these bounds in a \emph{bidirectional} fashion:
either $\precp$ and $\precq$ are bounds imposed by the surrounding
derivation tree producing $\cfs$, limiting the terms $\cfs$ \emph{produces};
or they are bound-requirements synthesized from the term that \emph{reduces}
to $\cfs$.
Our definition of elaboration in \autoref{fig:compile-precedence}
is organized accordingly.
A production rule $\cfs \produce \s{\cfx}$ is introduced either
by \emph{tightening} the bounds on $\cfs$ (\rulename{Produce-Tighten})
or by \emph{subsuming} the corresponding reduction $\cfs \reduce \s{\cfx}$
(\rulename{Produce-Subsume}),
as illustrated for $\hhazel$ in \autoref{fig:cfg-example}.

\begin{wrapfigure}[6]{l}{0.4\textwidth}
  \vspace{-1em}
  \begin{minipage}{0.4\textwidth}
  \begin{mathpar}
  \inferrule*{
    \spretile{let} \patsort \sbintile{=} \expsort \sbintile{in} \underline{\expsort} \in
    \rgxsem{\gram(\expsort,0)} \\\\
    \spretile{let} \neq \expsort \\
    0 \plt_\expsort 1 \\
    \bddsort{1}{\expsort}{\underline{\bot}} \consistent \underline{\expsort}
  }{
    \bddsort{\top}{\expsort}{\min(0,\underline{\bot})}
    \reduce
    \spretile{let}
      \botsort{\patsort}
    \sbintile{=}
      \botsort{\expsort}
    \sbintile{in}
      \bddsort{1}{\expsort}{\underline{\bot}}
  }
  \end{mathpar}
  \end{minipage}
\end{wrapfigure}
Meanwhile, a reduction rule $\cfs \reduce \s{\cfx}$
synthesizes the tightest possible bounds on $\cfs$
that can accommodate $\s{\cfx}$.
These correspond to Aasa's notion of \emph{precedence weights}
\cite{aasaPrecedencesSpecificationsImplementations1995}
that aggregate the precedence levels of operators exposed along
the left and right spines of a syntax tree.
Whether an operator contributes its annotated precedence level
to its left and right weights depends on its shape---either operand, prefix,
postfix, or infix.
For example, in the derivation on the left using rule \rulename{PElab-Prefix}
for $\hhazel$, the prefix-shaped $\spretile{let}$-form
synthesizes left weight $\top$
and right weight $\min(0,\underline{\bot})$,
where the latter folds in the annotated level $0$ into
the subweight $\underline{\bot}$ (underlined to distinguish it from
other $\bot$ values in the derivation) already computed for
the rightmost child $\bddsort{1}{\expsort}{\underline{\bot}}$.
Our bidirectional presentation reorganizes and generalizes Aasa's
to multi-sorted grammars of arbitrary mixfix forms,
which we discuss further in \autoref{sec:related-work}.

\subsection{OP Parsing Errors} \label{sec:op-parsing-errors}

In this section, we review Floyd's original method
for operator-precedence (OP) parsing
\cite{floydSyntacticAnalysisOperator1963}.
To motivate our error-handling generalization in
\autoref{sec:op-parsing-error-handling},
we consider in particular the different ways an OP parser can fail.

Parsing is the task of organizing token sequences into
grammatically well-formed terms.
\autoref{fig:terms-syntax} gives the syntax of terms:
a term $\term = \snode{\s{\nod}}$
demarcates a sequence $\s{\nod}$ of child nodes,
each either a token $\cft$ or a subterm.
We consider $\term$ to be well-formed if it is reducible to or producible
from a nonterminal $\cfs$, \ie $\cfs\Reduce \term$ or $\cfs\Produce\term$
as defined in \autoref{fig:terms-wf}.

\begin{figure}
  \hfill
  \begin{minipage}[b]{0.22\textwidth}
    \[\arraycolsep=3pt\begin{array}{rlrl}
      \text{\textcolor{gray}{node}} & \nod & \textcolor{gray}{::=} & \cft ~\gvert~ \term \\
      \text{\textcolor{gray}{term}} & \termr,\term & \textcolor{gray}{::=} & \snode{\s{\nod}}
    \end{array}\]
    \caption{Syntax of terms}\label{fig:terms-syntax}
  \end{minipage}\hfill%
  \begin{minipage}[b]{0.65\textwidth}
    \begin{minipage}[b]{0.55\textwidth}
      \centering
      \[\arraycolsep=3pt\begin{array}{rlrl}
        \text{\textcolor{gray}{leq}} & \cmprleq & \textcolor{gray}{::=} & \lessdot ~\gvert~ \doteq\ \ \textcolor{gray}{<:}\ \ \cmpr \\[0.5ex]
        \text{\textcolor{gray}{stack}} & \stack & \textcolor{gray}{::=} & \dlroot ~\ \gvert~\  \stack~\cmprleq_{\o{\term}}\cft
      \end{array}\]
    \end{minipage}\hfill%
    \begin{minipage}[b]{0.4\textwidth}
      $\arraycolsep=3pt\begin{array}{lcl}
          \head{\dlroot} & = & \dlroot \\
          \head{\stack\cmprleq_{\o{\term}}\cft} & = & \cft
      \end{array}$
    \end{minipage}
    \caption{Syntax of stacks}\label{fig:stack-syntax}
  \end{minipage}
\end{figure}


\begin{figure}
  \begin{minipage}[b]{0.72\textwidth}
    \begin{minipage}[b]{0.47\textwidth}
      \judgboxy{\cfx\Reduce\nod}{Node $\nod$ reduces\\ to symbol $\cfx$}
      \begin{mathpar}
        \inferrule*[left=\texttt{Reduce-Token}]{
        }{
          \cft\Reduce\cft
        } \\

        \inferrule[Reduce-Term$\;\glp\textcolor{gray}{k\geq0}\glr$]{
          \cfs\reduce\iseq{\cfx_i}{0\leq i\leq k} \\\\
          \iseq{\cfx_i\Reduce\nod_i}{0\leq i\leq k}
        }{
          \cfs\Reduce\snode{\iseq{\nod_i}{0\leq i\leq k}}
        }
      \end{mathpar}
    \end{minipage}\hfill%
    \begin{minipage}[b]{0.52\textwidth}
      \judgboxy[\ ]{\cfx\Produce\nod}{Symbol $\cfx$ produces\\ node $\nod$}
      \begin{mathpar}
        \inferrule*[left=\texttt{Produce-Token}]{
        }{
          \cft\Produce\cft
        } \\

        \inferrule[Produce-Term$\;\glp\textcolor{gray}{k\geq0}\glr$]{
          \cfs \produce \iseq{\cfx_i}{0\leq i\leq k} \\\\
          \iseq{\cfx_i\Produce\nod_i}{0\leq i\leq k}
        }{
          \cfs\Produce\snode{\iseq{\nod_i}{0\leq i\leq k}}
        }
      \end{mathpar}
    \end{minipage}
    \caption{A node is well-formed if it is reducible to or producible
    from a symbol.} \label{fig:terms-wf}
  \end{minipage}\hfill%
  \begin{minipage}[b]{0.25\textwidth}
    \judgboxy[\ \ \ \ ]{\wf{\stack}}{Stack $\stack$ is\\ well-formed}
    \begin{mathpar}
      \inferrule*[left=\texttt{WFStack-Nil}]{
      }{
        \wf{\dlroot}
      }

      \inferrule[WFStack-Cons]{
        \wf{\stack}\\
        \o{\cfs} \Produce \o{\term} \\\\
        \head{\stack}\cmprleq_{\o{\cfs}}\cft
      }{
        \wf{\stack~\cmprleq_{\o{\term}}\tnode{\cft}}
      }
    \end{mathpar}
    \caption{Well-formed stacks}
    \label{fig:stack-wf}
  \end{minipage}
  \vspace{-1em}
\end{figure}

OP parsing is a simple form of shift-reduce parsing:
input tokens are ingested one at a time, left-to-right,
and kept organized in a maximally reduced stack $\stack$
whose contents form prefixes of terms under construction.
\autoref{fig:stack-syntax} gives the syntax of OP parsing stacks:
a stack $\stack$ is either empty, the start delimiter $\dlroot$ affixed at its base;
or it is nonempty $\stack \cmprleq_{\o{\term}} \cft$, linking
a token $\cft$ to the rest of the stack $\stack$ with two pieces of information:
a comparison operator $\cmprleq$ recording how the head of $\stack$
precedence-relates to $\cft$, and an optional term $\o\term$
recording what was first reduced between them.
A stack $\stack$ is considered well-formed,
as specified in \autoref{fig:stack-wf},
when each of its links $\cft_L \cmprleq_{\o{\term}} \cft_R$
reflects a valid precedence relation $\cft_L \cmprleq_{\o{\cfs}} \cft_R$
such that $\o\cfs \Produce \o\term$.
For brevity, we will call optional nonterminals \emph{slots}
and optional terms \emph{cells}.

The \emph{height} of a stack is the number of $\lessdot$-operators it contains.
We can decompose any stack of height $h$
into a sequence of $h$ height-$1$ stacks, each of the form
${}$\hfill
$\cft \lessdot_{\o{\term}_0} \cft_0 \iseq{\doteq_{\o{\term}_i} \cft_i}{0<i\leq
k}$
$\glp\textcolor{gray}{k\geq 0}\glr$
\hfill${}$\\
We may interpret each such stack as a term under construction
${}$\hfill
$\lessdot_{\o{\term}_0} \cft_0 \iseq{\doteq_{\o{\term}_i} \cft_i}{0<i\leq k}$
\hfill${}$ \\
delimited on its left by $\cft$, which is
either the start of input $\dlroot$
or the head of the preceding stack in the decomposition.

\newcommand{\slabel}{\textcolor{\slabelcolor}{\ttfamily\bfseries\footnotesize Shift}\@\xspace}
\newcommand{\rlabel}{\textcolor{\rlabelcolor}{\ttfamily\bfseries\footnotesize Reduce}\@\xspace}

\autoref{fig:simple-push} shows Floyd's original algorithm,
presented here as a push operation $\pushl{\stack}{\o\termr}{\cft} = {\stack'}$
that pushes the next input token $\cft$ onto stack $\stack$
over the current reduction-in-progress $\o\termr$ to yield a new stack
$\stack'$.
\autoref{fig:simple-incomplete} and \autoref{fig:simple-unsound}
illustrate concrete OP parsing traces for $\hhazel$
using the precedence table in \autoref{fig:hazel-prec}---%
each colored box applies one of the rules in \autoref{fig:simple-push},
enumerating within it the satisfied premises, and sends
the push-inputs above it to the output stack below it.
Every push begins by consulting how the stack head $\head{\stack}$
precedence-compares with the pushed token $\cft$
to decide whether to \slabel or \rlabel.
If $\head{\stack} \cmprleq \cft$,
then the parser shifts $\cft$ onto $\stack$
and ``finalizes'' the reduction $\o\termr$ between them.
Else, if $\head{\stack} \gtrdot \cft$,
and $\stack$
has height $h\geq 1$,
then the parser has identified its next \emph{handle} (\ie reduction target)
of the form
${}$\hfill$\head{\stack}
\lessdot_{\dsty{\o{\termr}_0}}\! \cft_0
\iseq{\doteq_{\dsty{\o{\termr}_i}} \cft_i}{0<i\leq k}
\gtrdot_{\dsty{\o{\termr}_{k+1}}} \cft$
\hfill${}$ \\
where $\head{\stack}$ and $\cft$ delimit the handle
$\snode{\o{\termr}_0\iseq{\cft_i\o{\termr}_{i+1}}{0\leq i\leq k}}$
to be reduced and propagated up the stack.

\begin{figure}[h]
  \raggedright
  \begin{minipage}[t]{0.7\textwidth}
  \judgboxy{\push{\stack}{\dsty{\o\termr}}{\cft}{\stack'}}{
    Pushing token $\cft$ onto stack $\stack$
    over reduction $\o\termr$\\ \indent returns stack $\stack'$
  }
  \begin{mathpar}
  \mprset{vskip=1ex}
  \inferrule[OP-Shift]{
      \head{\stack} \cmprleq \cft
  }{
      \push{\stack}{\dsty{\o\termr}}{\cft}{\stack\cmprleq_{\dsty{\o\termr}}\cft}
  }
  \hspace{-1em}

  \inferrule[OP-Reduce$\;\glp\textcolor{gray}{k\geq0}\glr$]{
    \gtrel{\cft_k}{}{\cft} \\
    \push{\stack_0}{\dsty{\snode{\o{\termr}_0\iseq{\cft_i\o{\termr}_{i+1}}{0\leq i\leq k}}}}{\cft}{\stack'}
  }{
    \push{
      \stack_0
      \lessdot_{\dsty{\o{\termr}_0}}\cft_0
      ~\iseq{\doteq_{\dsty{\o{\termr}_i}}\cft_i}{0< i\leq k}
    }{\dsty{\o{\termr}_{k+1}}}{\cft}{\stack'}
  }
  \end{mathpar}
  \caption{OP parsing}\label{fig:simple-push}
  \end{minipage}
  \hfill
  \begin{minipage}[t]{0.26\textwidth}
    \vspace{-1em}
    \begin{tikzpicture}[
      box/.style={rectangle, inner sep=0pt, text width=2.75cm},
    ]
    \node (pushTwo) [box] {
      $\pushj{\dlroot}{\none}{\soptile{2}}$
    };
    \node (shiftTwo) [box, xshift=\pindent, below=of pushTwo] {
      $\dlroot \ltop{} \soptile{2}$
    };
    \node (pushLet) [box, xshift=-\pindent, below=of shiftTwo] {
      $\pushj[\ \ ]{\dlroot \ltop{\none}\soptile{2}}{\none}{\spretile{let}}$
    };
    \node (stuckLet) [box, xshift=\pindent, below=of pushLet] {
      $\soptile{2} ~?~ \spretile{let}$
    };
    \begin{scope}[on background layer]
      \shiftbox{shift2}{(shiftTwo)}
      \node (stucklet) [STUCK=(stuckLet)];
      \node [STLABEL] at (stucklet.north east) {Stuck};
    \end{scope}
    \end{tikzpicture}
    \caption{An OP parsing trace for $\hhazel$
    (\autoref{fig:hazel-prec}) that gets stuck trying to compare
    neighbors $\soptile{2}$ and $\spretile{let}$}\label{fig:simple-incomplete}
  \end{minipage}
\end{figure}

\begin{wrapfigure}[19]{R}{0.28\textwidth}
  \vspace{-1.5em}
  ${}$
  \hfill
  \begin{minipage}{0.26\textwidth}
  \newcommand{\termtwo}{\dsty{\snode{\soptile{2}}}}
  \newcommand{\termmult}{\dsty{\snode{\termtwo \sbintile{*}}}}
  \newcommand{\twidth}{2.5cm}
  \begin{tikzpicture}[
    box/.style={rectangle, inner sep=0pt, text width=\twidth},
  ]
  \node (pushTwo) [box] {
    $\pushj{\dlroot}{\none}{\soptile{2}}$
  };
  \node (shiftTwo) [box, xshift=\pindent, below=of pushTwo] {
    $\dlroot \ltop{} \soptile{2}$ \quad
  };
  \node (pushMult) [box, xshift=-\pindent, below=of shiftTwo] {
    $\pushj{\dlroot\ltop{\none} \soptile{2}}{\none}{\sbintile{*}}$
  };
  \node (reduceTwo) [box, xshift=\pindent, below=of pushMult] {
    $\soptile{2} \gtop{} \sbintile{*}$
  };
  \node (pushMult2) [box, below=of reduceTwo,
  ] {
    $\pushj{\dlroot}{\termtwo}{\sbintile{*}}$
  };
  \node (shiftMult) [
    box,
    text width={\twidth-\pindent},
    below right=4pt and \pindent of pushMult2.south west,
    anchor=north west
  ] {
    $\dlroot \ltop{} \sbintile{*}$
  };
  \node (pushEnd) [
    box,
    text width={\twidth+\pindent},
    below left=4pt and 2*\pindent of shiftMult.south west,
    anchor=north west
  ] {
    $\pushj{\dlroot \ltop{\termtwo} \sbintile{*}}{\none}{\drroot}$
  };
  \node (reduceMult) [
    box,
    below right=4pt and \pindent of pushEnd.south west,
    anchor=north west] {
    $\sbintile{*} \gtop{} \drroot$
  };
  \node (pushEnd2) [box, below=of reduceMult,
  ] {
    $\pushj[\hfill]{\dlroot}{\termmult}{\drroot}$
  };
  \node (shiftEnd) [box,
    text width={\twidth-\pindent},
    below right=4pt and \pindent of pushEnd2.south west,
    anchor=north west
  ] {
    $\dlroot \eqop{} \drroot$
  };
  \node (done) [box,
    text width={\twidth+\pindent},
    below left=4pt and 2*\pindent of shiftEnd.south west,
    anchor=north west
  ] {
    $\dlroot \eqop{\termmult} \drroot$
  };
  \begin{scope}[on background layer]
    \shiftbox{shift2}{(shiftTwo)}
    \reducebox{reduce2}{(reduceTwo) (pushMult2) (shiftMult)}
    \reducebox{reducemult}{(reduceMult) (pushEnd2) (shiftEnd)}
    \shiftbox{shiftmult}{(shiftMult)}
    \shiftbox{shiftend}{(shiftEnd)}
  \end{scope}
  \end{tikzpicture}
  \caption{A valid OP parsing trace for $\hhazel$
    that returns the invalid term $\termmult$}\label{fig:simple-unsound}
  \end{minipage}
\end{wrapfigure}
Let us consider the ways this algorithm can fail.

\subsubsection{Stuck}
Like with most (non-error-handling) methods,
a typical OP parser will easily get stuck.
This occurs when the stack head
and pushed token share no precedence relation,
like $\soptile{2}$ and $\soptile{let}$ in \autoref{fig:simple-incomplete}.


\subsubsection{Invalid Reduction} \label{sec:invalid-reduction}
OP parsing is unsound, meaning it can produce grammatically invalid reductions.
Recall from \autoref{sec:precedence-comparisons} that each precedence comparison
$\cft_L \odot_{\o\cfsr} \cft_R$ means there exists a derivable string of the
form $\dots\cft_L\o{\cfsr}\cft_R\dots$.
Floyd's original definition of precedence comparisons
omits the index $\o\cfsr$.
This ``nonterminal blindness'' means that an OP parser,
given a reduction $\termr?$ between delimiters $\cft_L$ and $\cft_R$,
can determine which parent delimiter(s) to reduce next,
but not whether $\termr?$ is a valid child of the chosen parent.
In the final \rlabel step in \autoref{fig:simple-unsound},
the parser identifies the handle pattern
$\dlroot \lessdot \sbintile{*} \gtrdot \drroot$
and proceeds blindly to reduce $\snode{\snode{\soptile{2}}\sbintile{*}}$
without checking that the initial reduction $\none$ is a valid right-argument to
$\sbintile{*}$.

\subsubsection{Invalid Prefix} \label{sec:invalid-prefix}
A core tenet of shift-reduce parsers is the \emph{valid prefix property},
which maintains that the parse stack
forms the prefix of some grammar-derivable symbol string.
This property ensures that the parser is sound, \ie every parsed term is
well-formed.

\begin{wrapfigure}[3]{l}{0.18\textwidth}
  \vspace{-1em}
  \begin{tabular}{lcl}
    $\startsort{\textsc{t}}$ & $\rightarrow$ & \texttt{\$}\,\textsc{x}\,\texttt{\$} $\mid$
    \texttt{t} \\
    \textsc{x} & $\rightarrow$ & \texttt{\{}\,\textsc{t}\,\texttt{\}} $\mid$ \texttt{x}
  \end{tabular}
\end{wrapfigure}
Ideally prefix-validity would be implied by stack well-formedness (\autoref{fig:stack-wf}),
but this is not always the case for an OP parser depending on the grammar.
Consider the small grammar on the left,
which produces strings like
\texttt{t}, \texttt{\$x\$},
\texttt{\$\{t\}\$}, \texttt{\$\{\$x\$\}\$}, etc.
Generated from this grammar are the
precedence comparisons $\dlroot \lessdot \texttt{\$}$
and $\texttt{\$} \doteq \texttt{\$}$,
so Floyd's parser would happily ingest the tokens $\texttt{\$x\$x\$}$
and organize them into the stack
$\dlroot \lessdot_{\none} \texttt{\$}
\doteq_{\some{\snode{\texttt{x}}}} \texttt{\$}
\doteq_{\some{\snode{\texttt{x}}}} \texttt{\$}$,
which is well-formed but prefix-invalid.
The issue here is the reuse of $\texttt{\$}$
as both opener and closer in the rule
$\startsort{\textsc{t}} \rightarrow \texttt{\$}\,\textsc{x}\,\texttt{\$}$.
Distinguishing between these two uses would require stack-level analyses
out of scope of the local pairwise precedence comparisons.





\subsection{OP Parsing with Error Handling} \label{sec:op-parsing-error-handling}

We now define our error-handling extension of OP parsing
that avoids or recovers from the various failure modes seen in the last section.
Some of our changes involve requirements (\autoref{sec:molding-tiles})
and transformations (\autoref{sec:injecting-grout}) of the language grammar;
others involve generalizing Floyd's algorithm (\autoref{sec:tylrcore-parsing})
to incorporate completion-based repairs
and to restore soundness by making use
of our nonterminal-enriched precedence comparisons.

\subsubsection{Molding Tiles} \label{sec:molding-tiles}
To secure the valid prefix property (\autoref{sec:invalid-prefix}),
we take the blunt approach of requiring
every tile $\labl\in\lbls$ to appear uniquely in the PBG $\gram$:
\noindent\begin{assumption}[Unique Tiles] \label{assum:unique-tiles}
  A tile $\labl\in\lbls$ is called \emph{unique}
  if \emph{\glp}${\dots \labl\dots} \in \rgxsem{\gram(\sort,\precp)}$
  and ${\dots \labl\dots} \in \rgxsem {\gram(\sortr,\precq)}$\emph{\glr}
  imply \emph{\glp}$\sort = \sortr$ and $\precp = \precq$
  and $\labl$ appears uniquely in $\gram(\sort,\precp)$\emph{\glr}.
  All tiles $\labl\in\lbls$ are unique.
\end{assumption}

\noindent
With \autoref{assum:unique-tiles}, we can guarantee for any height-$1$ precedence chain
of the form $\cft
\lessdot_{\dsty{\o{\cfsr}_0}}
\labl_0 \iseq{\doteq_{\dsty{\o{\cfsr}_i}} \labl_i}{0<i\leq k}$
that the string $\iseq{\o{\cfsr}_i\labl_i}{0\leq i\leq k}$
forms a prefix of the yield of some nonterminal $\cfs$ adjacent to $\cft$:

\begin{lemma}[Valid Prefixes] \label{lemma:valid-prefix}
  For all terminals $\cft$,
  tiles $\iseq{\labl_i}{0\leq i\leq k}$,
  and slots $\iseq{\o{\cfsr}_i}{0\leq i\leq k+1}$
  $\glp\textcolor{gray}{k\geq 0}\glr$
  such that
  \ \ $\cft \lessdot_{\dsty{\o{\cfsr}_0}} \labl_0
  \iseq{\doteq_{\dsty{\o{\cfsr}_i}} \labl_i}{0<i\leq k}
  \adjop \o{\cfsr}_{k+1}$\ \
  there exist nonterminals $\cfs_\cft,\cfs$,
  tiles $\iseq{\labl_i}{k < i \leq \ell}$,
  and slots $\iseq{\o{\cfsr}_i}{k + 1 < i \leq \ell + 1}$
  $\glp\textcolor{gray}{\ell\geq k}\glr$
  such that
  \ \ $\cft \adjop \cfs_\cft
  \lmsymop^* \cfs
  \produce \o{\cfsr}_0 \iseq{\labl_i\,\o{\cfsr}_{i+1}}{0\leq i\leq \ell}
  $\ \ .
\end{lemma}

\autoref{assum:unique-tiles} would be severely constraining
if $\gram$ were a grammar of purely textual tokens---%
for example, we would not be able to reuse
parentheses $\code{(}$ $\code{)}$ across different sorts.
In this work, we take $\gram$ to be a grammar of tiles,
which we conceptualized in our Hazel grammar $\ghazel$
(\autoref{fig:pbg-example})
as being textual tokens paired with ``molds'',
visually distinguished using color and shape.
Rather than requiring that the grammar author
manually design and distinguish their terminal symbols,
however, we can generically convert
any ordinary textual grammar $\mathcal{F}$
into a grammar $\gram$ of unique tiles,
simply by augmenting each terminal symbol in $\mathcal{F}$
with its \emph{one-hole context}, \ie its mold.
We continue this discussion in \autoref{sec:implementation},
where we describe how \tylr chooses between
multiple possible molds for a textual token.

\newcommand{\fbound}[2]{#1\,!\,#2}
\newcommand{\restrictb}[2]{\langle #1 \rangle^{#2}_{\sort}}
\newcommand{\restrictz}[1]{\restrictb{#1}{0}\,}
\begin{figure}[b]
\begin{minipage}[t]{0.4\textwidth}
  \vskip0pt
\judgbox{\syninj{\gram}{\s{\cfx}}{\cfs}}{$\s{\cfx}$ reduces to $\cfs$ in
grout-injected $\gram$}
\end{minipage}
\hfill
\begin{minipage}[t]{0.45\textwidth}
  \vskip0pt
  \[\arraycolsep=3pt\begin{array}{rlcl}
    \text{\textcolor{gray}{grout}} & \groot & \textcolor{gray}{::=} & \ophole ~\gvert~
    \prehole ~\gvert~ \poshole ~\gvert~ \binhole \\
    \text{\textcolor{gray}{terminal}} & \cft & \textcolor{gray}{::=} & \textcolor{gray}{\dots}
    ~\gvert~ \groot^{\sort}
  \end{array}\]
\end{minipage}
\vskip 4pt

\begin{mathpar}
  \ogdots

  \inferrule[GInj-Hole]{
  }{\syninj{\gram}{\ophole^\sort}{\topsort{\sort}}}

  \fbound{\cfs}{\sort} =
  \left\{\begin{array}{cl}
    \top & \text{if } \cfs \not\sim \sort \\
    \bot & \text{if } \cfs \sim \sort
  \end{array}\right\}

  \restrictb{\prq}{\precm} =
  \left\{\begin{array}{cl}
    \bddsort{\max(\precm, \precp)}{\sortr}{\max(\precq, \precm)} & \text{if } \sortr = \sort \\
    \prq & \text{if } \sortr \neq \sort
  \end{array}\right\} \\


  \inferrule[GInj-Operand$\;\glp\textcolor{gray}{k\geq0}\glr$]{
    \iseq{ \yieldsStar{\botsort{\sort}}{\cfsr_i}}{0 \leq i \leq k}
  }{\syninj{\gram}{\prehole^\sort \restrictz{\cfsr_0} \iseq{\binhole^\sort\restrictz{\cfsr_i}}{0 < i\leq k}\,\poshole^\sort}{\topsort{\sort}}
  }

  \inferrule[GInj-Infix$\;\glp\textcolor{gray}{k\geq0}\glr$]{
     \lmsymStar{\botsort{\sort}}{\cfsr_0} \\
    \iseq{ \yieldsStar{\botsort{\sort}}{\cfsr_i}}{0 < i < k} \\
     \rmsymStar{\cfsr_k}{\botsort{\sort}}
  }{
    \syninj{\gram}{
      \restrictz{\cfsr_0} \iseq{\binhole^\sort\restrictz{\cfsr_i}}{0 < i\leq k}
    }{\bddsort{\fbound{\cfsr_0}{\sort}}{\sort}{\fbound{\cfsr_k}{\sort}}}
  }

  \inferrule[GInj-Prefix$\;\glp\textcolor{gray}{k\geq0}\glr$]{
    \iseq{ \yieldsStar{\botsort{\sort}}{\cfsr_i}}{0 \leq i < k} \\
     \rmsymStar{\cfsr_k}{\botsort{\sort}}
  }{\syninj{\gram}{
    \prehole^\sort \restrictz{\cfsr_0} \iseq{\binhole^\sort\restrictz{\cfsr_i}}{0 < i\leq k}
  }{\bddsort{\top}{\sort}{\fbound{\cfsr_k}{\sort}}}}

  \inferrule[GInj-Postfix$\;\glp\textcolor{gray}{k\geq0}\glr$]{
     \lmsymStar{\botsort{\sort}}{\cfsr_k} \\
    \iseq{ \yieldsStar{\botsort{\sort}}{\cfsr_i}}{k > i \geq 0}
  }{\syninj{\gram}{
    \iseq{\restrictz{\cfsr_i}\,\binhole^\sort}{k\geq i > 0}\,\restrictz{\cfsr_0}\,\poshole^\sort
  }{\bddsort{\fbound{\cfsr_k}{\sort}}{\sort}{\top}}}
\end{mathpar}
\caption{Grout injection extending
the definition of terminals $\cft$ (\autoref{fig:cfg})
and reduction $\synelab{\gram}{\s{\cfx}}{\cfs}$
(\autoref{fig:compile-precedence})}
\label{fig:add-grout}
\end{figure}

\subsubsection{Injecting Grout} \label{sec:injecting-grout}

When a shift-reduce parser ``goes wrong'',
it is because of an unresolved mismatch between
the bottom-up reductions accumulated so far and
the remaining top-down expectations of the grammar.
Many of these mismatches are inconsistencies of \emph{multiplicity}:
in \autoref{fig:simple-unsound}, the reduction
$\snode{\snode{\soptile{2}}\sbintile{*}}$ in the last \rlabel step
is ill-formed because there is no term where one is expected
as the right multiplicand;
in \autoref{fig:simple-incomplete},
the parser gets stuck on neighbors $\soptile{2}$ and $\spretile{let}$
because it does not know how to combine these parts of two unrelated terms into
one as required.
When multiplicities align, there remains further the possibility of \emph{sort}
inconsistencies, such as the one in \autoref{fig:sort-transition-obligations}
between the $\spretile{let}$-delimiter expecting a pattern
and the $\sbintile[\typcolor]{->}$-term providing a type.

\begin{figure}
  \centering
  \tikzset{nt/.style={
    inner sep=2,
    outer sep=0,
  }}
  \begin{tikzpicture}[
    >={Classical TikZ Rightarrow[]},
  ]
    \node (botbot) [nt]
    at (0, 0) {
      \phnt
      $\botsort{\expsort}$
    };
    \node (botbot') [right=.5 of botbot] {
      \phnt
      $\zerosort{\expsort} \,\sbinhole\, \zerosort{\expsort} ~\gvert~
      \zerosort{\expsort} \,\sbinhole\, \zerosort{\expsort}
      \,\sbinhole\, \zerosort{\expsort}
      ~\gvert~ \dots$};

    \node (bottop) [nt, below=.6 of botbot] {
      \phnt
      $\bddsort{\bot}{\expsort}{\top}$
    };
    \node (bottop') [right=.9 of bottop] {
      \phnt
      $\zerosort{\expsort}\, \sposhole ~\gvert~
      \zerosort{\expsort}\,\sbinhole\,\botsort{\typsort}\,\sposhole
      ~\gvert~ \dots$
    };

    \node (topbot) [nt, below right=0 and .15 of botbot] {
      \phnt
      $\bddsort{\top}{\expsort}{\bot}$
    };
    \node (topbot') [right=.5 of topbot] {
      \phnt
      $\sprehole\,\zerosort{\expsort} ~\gvert~
      \sprehole\,\botsort{\patsort}
      \,\sbinhole\,\zerosort{\expsort}
        ~\gvert~ \dots$
    };

    \node (toptop) [nt, below=.6 of topbot] {
      \phnt
      $\topsort{\expsort}$
    };
    \node (toptop') [right=.5 of toptop] {
      \phnt
      $\sophole ~\gvert~
      \sprehole\,\zerosort{\expsort}\,\sposhole
      ~\gvert~ \sprehole\,\botsort{\patsort}\,\sposhole
      ~\gvert~ \dots$
    };

    \node (pbotbot) [nt, right=0.3 of botbot'] {
      \phnt
      $\botsort{\patsort}$
    };
    \node (pbotbot') [right=.5 of pbotbot] {
      \phnt
      $\zerosort{\patsort} \,\sbinhole[\patcolor]\, \zerosort{\patsort} ~\gvert~
      \zerosort{\patsort} \,\sbinhole[\patcolor]\, \botsort{\typsort}
      ~\gvert~ \dots$};

    \node (pbottop) [nt, below=.6 of pbotbot] {
      \phnt
      $\bddsort{\bot}{\patsort}{\top}$
    };
    \node (pbottop') [right=.9 of pbottop] {
      \phnt
      $\zerosort{\patsort}\, \sposhole[\patcolor] ~\gvert~
      \zerosort{\patsort}\,\sbinhole[\patcolor]\,\botsort{\typsort}\,\sposhole[\patcolor]
      ~\gvert~ \dots$
    };

    \node (ptopbot) [nt, below right=0 and .15 of pbotbot] {
      \phnt
      $\bddsort{\top}{\patsort}{\bot}$
    };
    \node (ptopbot') [right=.5 of ptopbot] {
      \phnt
      $\sprehole[\patcolor]\,\zerosort{\patsort} ~\gvert~
      \sprehole[\patcolor]\,\botsort{\typsort}
      ~\gvert~ \dots$
    };

    \node (ptoptop) [nt, below=.6 of ptopbot] {
      \phnt
      $\topsort{\patsort}$
    };
    \node (ptoptop') [right=.5 of ptoptop] {
      \phnt
      $\sophole[\patcolor] ~\gvert~
      \sprehole[\patcolor]\,\zerosort{\typsort}\,\sposhole[\patcolor] ~\gvert~
      \dots$
    };

    \node (tbotbot) [nt, right=0.3 of pbotbot'] {
      \phnt
      $\botsort{\typsort}$
    };
    \node (tbotbot') [right=.5 of tbotbot] {
      \phnt
      $\zerosort{\typsort} \,\sbinhole[\typcolor]\, \zerosort{\typsort}
      ~\gvert~ \dots$};

    \node (tbottop) [nt, below=.6 of tbotbot] {
      \phnt
      $\bddsort{\bot}{\typsort}{\top}$
    };
    \node (tbottop') [right=.9 of tbottop] {
      \phnt
      $\zerosort{\typsort}\, \sposhole[\typcolor]
      ~\gvert~ \dots$
    };

    \node (ttopbot) [nt, below right=0 and .15 of tbotbot] {
      \phnt
      $\bddsort{\top}{\typsort}{\bot}$
    };
    \node (ttopbot') [right=.5 of ttopbot] {
      \phnt
      $\sprehole[\typcolor]\,\zerosort{\typsort}
      ~\gvert~ \dots$
    };

    \node (ttoptop) [nt, below=.6 of ttopbot] {
      \phnt
      $\topsort{\typsort}$
    };
    \node (ttoptop') [right=.5 of ttoptop] {
        \phnt
        $\sophole[\typcolor] ~\gvert~
        \dots$
        \vphantom{$\sprehole[\typcolor]\,\zerosort{\typsort}\,\sposhole[\typcolor]$}
      };

    \draw[-implies, double equal sign distance, \subsumecolor, line width=0.7pt] (botbot) -- (botbot');
    \draw[-implies, double equal sign distance, \subsumecolor, line width=0.7pt] (bottop) -- (bottop');
    \draw[-implies, double equal sign distance, \subsumecolor, line width=0.7pt] (topbot) -- (topbot');
    \draw[-implies, double equal sign distance, \subsumecolor, line
    width=0.7pt] (toptop) -- (toptop');

    \draw[-implies, double equal sign distance, \subsumecolor, line
    width=0.7pt] (pbotbot) -- (pbotbot');
    \draw[-implies, double equal sign distance, \subsumecolor, line
    width=0.7pt] (pbottop) -- (pbottop');
    \draw[-implies, double equal sign distance, \subsumecolor, line width=0.7pt] (ptopbot) -- (ptopbot');
    \draw[-implies, double equal sign distance, \subsumecolor, line width=0.7pt] (ptoptop) -- (ptoptop');

    \draw[-implies, double equal sign distance, \subsumecolor, line width=0.7pt] (tbotbot) -- (tbotbot');
    \draw[-implies, double equal sign distance, \subsumecolor, line width=0.7pt] (tbottop) -- (tbottop');
    \draw[-implies, double equal sign distance, \subsumecolor, line width=0.7pt] (ttopbot) -- (ttopbot');
    \draw[-implies, double equal sign distance, \subsumecolor, line width=0.7pt] (ttoptop) -- (ttoptop');

    \draw[-implies, double equal sign distance, \tightencolor, line width=0.7pt] (botbot) -- (bottop);
    \draw[-implies, double equal sign distance, \tightencolor, line width=0.7pt] (bottop) -- (toptop.west);
    \draw[-implies, double equal sign distance, \tightencolor, line width=0.7pt] (botbot) -- (topbot.west);
    \draw[-implies, double equal sign distance, \tightencolor, line width=0.7pt] (topbot) -- (toptop);
    \draw[-implies, double equal sign distance, \tightencolor, line
    width=0.7pt] (topbot) -- (toptop);

    \draw[-implies, double equal sign distance, \tightencolor, line width=0.7pt] (pbotbot) -- (pbottop);
    \draw[-implies, double equal sign distance, \tightencolor, line width=0.7pt] (pbottop) -- (ptoptop.west);
    \draw[-implies, double equal sign distance, \tightencolor, line width=0.7pt] (pbotbot) -- (ptopbot.west);
    \draw[-implies, double equal sign distance, \tightencolor, line width=0.7pt] (ptopbot) -- (ptoptop);
    \draw[-implies, double equal sign distance, \tightencolor, line
    width=0.7pt] (ptopbot) -- (ptoptop);

    \draw[-implies, double equal sign distance, \tightencolor, line width=0.7pt] (tbotbot) -- (tbottop);
    \draw[-implies, double equal sign distance, \tightencolor, line width=0.7pt] (tbottop) -- (ttoptop.west);
    \draw[-implies, double equal sign distance, \tightencolor, line width=0.7pt] (tbotbot) -- (ttopbot.west);
    \draw[-implies, double equal sign distance, \tightencolor, line width=0.7pt] (ttopbot) -- (ttoptop);
    \draw[-implies, double equal sign distance, \tightencolor, line
    width=0.7pt] (ttopbot) -- (ttoptop);

  \end{tikzpicture}
  \vskip -7pt
  \caption{Excerpt of the grout rules injected
  (\autoref{fig:add-grout})
  into $\hhazel$
  (\autoref{fig:cfg-example}).
  The production rules are arranged and color-coded
  by whether they emerge from subsuming \textcolor{\subsumecolor}{reduction}
  or by \textcolor{\tightencolor}{tightening}
  (\autoref{fig:compile-precedence}).
}
  \label{fig:ginj-example}
\end{figure}

\autoref{fig:add-grout} shows how we materialize these inconsistencies
as \emph{grout} forms injected into the elaborated grammar,
extending our definition in \autoref{fig:compile-precedence},
while \autoref{fig:ginj-example} shows an excerpt of
the grout forms injected into $\hhazel$.
Every sort acquires the form $\topsort{\sort} \produce \ophole^\sort$,
injected via rule \rulename{GInj-Hole},
consisting of a single convex grout terminal $\ophole^\sort$
that stands in for missing terms of sort $\sort$.

Grout terminals also come in prefix $\prehole$, postfix $\poshole$,
and infix $\binhole$ shapes that are used to wrap
sort-inconsistent and extraneous terms,
injected via the rules \rulename{GInj-Operand}, \rulename{GInj-Infix},
\rulename{GInj-Prefix}, and \rulename{GInj-Postfix}.
There are four of these rules to enumerate over
whether the left and right ends of the form
are bookended with a prefix $\prehole^\sort$
or postfix $\poshole^\sort$ grout, respectively.
The left (right) bookend is optional when the exposed nonterminal
is a leftmost (rightmost) descendant of the unbounded sort $\botsort{\sort}$.
For example, \autoref{fig:ginj-example} includes the grout production
$\bddsort{\top}{\patsort}{\bot}
\produce \sprehole[\patcolor] \botsort{\typsort}$
because of the $\patsort$-sorted form
$\bddsort{\bot}{\patsort}{1} \sbintile[\patcolor]{:} \botsort{\typsort}$
in $\hhazel$ (\autoref{fig:cfg-example}), where $\botsort{\typsort}$ is the rightmost descendant.
On the other hand, there is no $\expsort$-sorted form with $\botsort{\patsort}$
as its leftmost or rightmost descendant, so $\botsort{\patsort}$ can only appear
in the $\expsort$-sorted grout forms that buffer it on both sides (\eg
$\sprehole[\expcolor] \botsort{\patsort} \sposhole[\expcolor]$).

Grout terminals behave like associative operators of loosest precedence within
each sort, where their left and right tip decorations follow the pattern of tiles.
More precisely, $\groot_L^\sort \doteq \groot_R^\sort$
if $\groot_L$ is right-concave and $\groot_R$ is left-concave,
and $\groot^\sort \lessdot \labl$ for any tile $\labl$ of sort $\sort$
if $\groot$ is right-concave.
The nonterminal descendants $\cfsr_i$ are precedence-bounded in
their injected forms $\restrictz{\cfsr_i}$, depending on their sort,
to prevent conflicting precedence comparisons between grout terminals of the same sort.

\subsubsection{Parsing with \tylrcore} \label{sec:tylrcore-parsing}


\begin{figure}
\begin{minipage}[t]{0.52\textwidth}
  \input{fig/fill}

  \vskip1em
  \begin{center}
  \begin{tabular}{rcl}
  $\parsel{\stack}{\snil}$ & $=$ & $\stack$ \\
  $\parsel{\stack}{\cft\,\s{\cft}}$ & $=$
  & $\parsel{\pushl{\stack}{\snil}{\cft}}{\s{\cft}}$
  \end{tabular}
  \end{center}

  \caption{Parsing with \tylrcore}\label{fig:tylrcore}
\end{minipage}
\hfill
\begin{minipage}[t]{0.4\textwidth}
\judgboxy{\push{\stack}{\dsty{\s{\termr}}}{\cft}{\stack'}}{
  Pushing token $\cft$ \\ onto stack $\stack$ \\
  over reductions $\s{\termr}$ \\returns stack $\stack'$%
}
\begin{mathpar}
\mprset{vskip=1ex}
\inferrule[Shift$\;\glp\textcolor{gray}{k \geq 0}\glr$]{
    \head{\stack}\iseq{\cmprleq_{\dsty{\o{\cfsr}_i}}\cft_i}{0\leq i\leq k} = \cft \\\\
    \pfill{\s{\termr}}{\iseq{\o{\cfsr}_i}{0\leq i\leq k}} = {\iseq{\o{\term}_i}{0\leq i\leq k}}
}{
    \push{\stack}{\dsty{\s{\termr}}}{\cft}{\stack\iseq{\cmprleq_{\dsty{\o{\term}_i}}\cft_i}{0\leq i\leq k}}
} \\

\inferrule[Reduce$\;\glp\textcolor{gray}{0 \leq k \leq \ell}\glr$]{
  \labl_k \iseq{\doteq_{\dsty{\o{\cfsr}}_i} \labl_i}{k \leq i \leq \ell}
  \gtrdot_{\dsty{\o{\cfsr}_{\ell+1}}~ {\medsquare}} \\
  \pfill{\s{\termr}}{\iseq{\o{\cfsr}_i}{k < i \leq \ell+1}} = {\iseq{\o{\termr}_i}{k < i \leq \ell+1}} \\\\
  \push{\stack_0}{\displaystyle{
    \snode{\o{\termr}_0\iseq{\labl_i\,\o{\termr}_{i+1}}{0\leq i\leq
    \ell}}
  }}{\cft}{\stack}
}{
  \push{
    \stack_0
    \lessdot_{\dsty{\o{\termr}_0}}\labl_0\,
    \iseq{\doteq_{\dsty{\o{\termr}_i}}\labl_i}{0\leq i\leq k}
  }{\dsty{\s{\termr}}}{\cft}{\stack}
} \\

\inferrule[Degrout$\;\glp\textcolor{gray}{k \geq 0}\glr$]{
  \push{\stack_0}{\dsty{
    \iseq{\o{\termr}_i}{0\leq i\leq k}\,
    \s{\termr}
  }}{\cft}{\stack}
}{
  \push{
    \stack_0
    \lessdot_{\dsty{\o{\termr}_0}}\groot^\sort_0\,
    \iseq{\doteq_{\dsty{\o{\termr}_i}}\groot^\sort_i}{0< i\leq k}
  }{\dsty{\s{\termr}}}{\cft}{\stack}
}
\end{mathpar}
\caption{Pushing with \tylrcore}\label{fig:tylrcore-push}
\end{minipage}








\end{figure}

\autoref{fig:tylrcore-push} gives the rules for parsing with \tylrcore,
whose notable features we will illustrate through several examples.

\begin{figure}
  \begin{minipage}[b]{0.7\textwidth}
    \centering
  \begin{minipage}[b]{0.43\textwidth}
    \newcommand{\twidth}{3.1cm}
    \begin{tikzpicture}[
      box/.style={rectangle, inner sep=0pt, text width=\twidth},
    ]
    \node (pushMult) [box, xshift=-\pindent, text width=3.5cm] {
      $\pushj{\dlroot \lessdot_{\none} \spretile{let} \lessdot_{\none} \soptile[\patcolor]{x}}{\none}{\sbintile[\patcolor]{:}}$
    };
    \node (reduceTwo) [box, xshift=\pindent, below=of pushMult.south west,
      anchor=north west] {
      $\soptile[\patcolor]{x} \gtrdot \sbintile[\patcolor]{:}$
      $\vphantom{\soptile[\patcolor]{x} \gtrdot_{\none_{\mathsf{slot}}}}$
    };
    \node (reduceTwo2) [box, below=of reduceTwo] {
      $\vphantom{\pick{\snil}{\none}{\none}}$
    };
    \node (pushMult2) [box, below=of reduceTwo2] {
      $\pushj{\dlroot \lessdot_{\none} \spretile{let}}{\snode{\soptile[\patcolor]{x}}}{\sbintile[\patcolor]{:}}$
    };
    \node (shiftMult) [
      box,
      text width={\twidth-\pindent},
      below right=4pt and \pindent of pushMult2.south west,
      anchor=north west
    ] {
      $\spretile{let} \lessdot \sbintile[\patcolor]{:}$
      $\vphantom{\spretile{let} \lessdot_{\dsty{\bddsort{\bot}{\patsort}{1}}} \sbintile[\patcolor]{:}}$
    };
    \node (shiftMult2) [
      box,
      text width={\twidth-\pindent},
      below=2pt of shiftMult,
    ] {
      $\vphantom{\pfill{\snode{\soptile[\patcolor]{x}}\!}{\!\bddsort{\bot}{\patsort}{1}} = {\snode{\soptile[\patcolor]{x}}}}$
    };
    \node (pushEnd) [
      box,
      text width={\twidth+\pindent},
      below left=4pt and 2*\pindent of shiftMult2.south west,
      anchor=north west
    ] {
      $\dlroot \lessdot_{\none} \spretile{let} \lessdot_{\dsty{\snode{\soptile[\patcolor]{x}}}} \sbintile[\patcolor]{:}$
    };
    \begin{scope}[on background layer]
      \reducebox{reduce2}{(reduceTwo) (pushMult2) (shiftMult2)}
      \shiftbox{shiftmult}{(shiftMult) (shiftMult2)}
    \end{scope}
    \end{tikzpicture}
  \end{minipage}
  \hfill
  \begin{minipage}[b]{0.56\textwidth}
    \centering
    \newcommand{\twidth}{4cm}
    \begin{tikzpicture}[
      box/.style={rectangle, inner sep=0pt, text width=\twidth},
    ]
      \node (pushMult) [box, xshift=-\pindent, text width=4.5cm] {
        $\pushj{\dlroot \lessdot_{\none} \spretile{let} \lessdot_{\none} \soptile[\patcolor]{x}}{\snil}{\sbintile[\patcolor]{:}}$
      };
      \node (reduceTwo) [box, xshift=\pindent, below=of pushMult.south west,
        anchor=north west] {
        \raisebox{2pt}{$\soptile[\patcolor]{x} \gtrdot_{\none_{\mathsf{slot}}} \sbintile[\patcolor]{:}$}
        $\vphantom{\topsort{\patsort}\soptile[\patcolor]{x} \gtrdot \sbintile[\patcolor]{:}}$
      };
      \node (reduceTwo2) [box, below=of reduceTwo] {
        $\pick{\snil}{\none_{\mathsf{slot}}}{\none_{\mathsf{cell}}}$
      };
      \node (pushMult2) [box, below=of reduceTwo2
      ] {
        $\pushj{\dlroot \lessdot_{\none} \spretile{let}}{\dsty{\snode{\none\soptile[\patcolor]{x}\none_{\mathsf{cell}}}}}{\sbintile[\patcolor]{:}}$
      };
      \node (shiftMult) [
        box,
        text width={\twidth-\pindent},
        below right=4pt and \pindent of pushMult2.south west,
        anchor=north west
      ] {
        $\spretile{let} \lessdot_{\dsty{\bddsort{\bot}{\patsort}{1}}} \sbintile[\patcolor]{:}$
      };
      \node (shiftMult2) [
        box,
        text width={\twidth-\pindent},
        below=2pt of shiftMult,
      ] {
        $\pfill{\snode{\soptile[\patcolor]{x}}\!}{\!\bddsort{\bot}{\patsort}{1}} = {\snode{\none\soptile[\patcolor]{x}\none_{\mathsf{cell}}}}$
      };
      \node (pushEnd) [
        box,
        text width={\twidth+\pindent},
        below left=4pt and 2*\pindent of shiftMult2.south west,
        anchor=north west
      ] {
        $\dlroot \lessdot_{\none} \spretile{let} \lessdot_{\dsty{\snode{\none\soptile[\patcolor]{x}\none_{\mathsf{cell}}}}} \sbintile[\patcolor]{:}$
      };
      \begin{scope}[on background layer]
        \reducebox{reduce2}{(reduceTwo) (pushMult2) (shiftMult2)}
        \shiftbox{shiftmult}{(shiftMult) (shiftMult2)}
      \end{scope}
      \end{tikzpicture}
    \end{minipage}
    \caption{Corresponding traces of OP parsing (left) and \tylrcore (right)
    on the same inputs to highlight their differences}\label{fig:meldr-subsume}
    \end{minipage}
    \hfill
    \begin{minipage}[b]{0.26\textwidth}
      \centering
      \newcommand{\twidth}{3.1cm}
      \begin{tikzpicture}[
        box/.style={rectangle, inner sep=0pt, text width=\twidth},
      ]
      \node (pushMult2) [box] {
        $\pushj{\dlroot \lessdot_{\none}
        \spretile{let}}{\snil}{\sbintile[\patcolor]{:}}$
        $\vphantom{\pushj{\dlroot \lessdot_{\none} \spretile{let}}{\dsty{\snode{\none\soptile[\patcolor]{x}\none_{\mathsf{cell}}}}}{\sbintile[\patcolor]{:}}}$
      };
      \node (shiftMult) [
        box,
        text width={\twidth-\pindent},
        below right=4pt and \pindent of pushMult2.south west,
        anchor=north west
      ] {
        $\spretile{let} \lessdot_{\dsty{\bddsort{\bot}{\patsort}{1}}} \sbintile[\patcolor]{:}$
        $\vphantom{\spretile{let} \lessdot_{\dsty{\bddsort{\bot}{\patsort}{1}}} \sbintile[\patcolor]{:}}$
      };
      \node (shiftMult2) [
        box,
        text width={\twidth-\pindent},
        below=2pt of shiftMult,
      ] {
        $\pfill{\snil\!}{\!\bddsort{\bot}{\patsort}{1}} = {\snode{\ophole^\patsort}}$
      };
      \node (pushEnd) [
        box,
        text width={\twidth+\pindent},
        below=of shiftMult2,
      ] {
        $\dlroot \lessdot_{\none} \spretile{let}
        \lessdot_{\dsty{\snode{\ophole^\patsort}}} \sbintile[\patcolor]{:}$
        $\vphantom{\dlroot \lessdot_{\none} \spretile{let} \lessdot_{\dsty{\snode{\none\soptile[\patcolor]{x}\none_{\mathsf{cell}}}}} \sbintile[\patcolor]{:}}$
      };
      \begin{scope}[on background layer]
        \shiftbox{shiftmult}{(shiftMult) (shiftMult2)}
      \end{scope}
      \end{tikzpicture}
    \caption{\tylrcore filling the slot ${\bddsort{\bot}{\patsort}{1}}$
    with the grout form $\snode{\ophole^\patsort}$} \label{fig:meldr-fill-grout}
    \end{minipage}
\end{figure}

\autoref{fig:meldr-subsume} illustrates how \tylrcore directly generalizes the
standard non-error-handling algorithm (\autoref{fig:simple-push}).
The main difference is the new \emph{fill} operation,
defined in \autoref{fig:fill},
invoked in \autoref{fig:meldr-subsume} as
\ $\pick{\snil}{\none_{\mathsf{slot}}}{\none_{\mathsf{cell}}}$\
in \rlabel
and
\ $\pfill{\snode{\soptile[\patcolor]{x}}\!}
{\!\bddsort{\bot}{\patsort}{1}} = {\snode{\soptile[\patcolor]{x}}}$\
in \slabel.
Filling is responsible for
assigning accumulated reductions to grammatically appropriate slots,
now exposed as operator indices in the precedence comparisons.
In \rlabel, nothing $\snil$ is assigned
to the unfillable slot $\none_{\mathsf{slot}}$;
in \slabel, the reduction $\snode{\soptile[\patcolor]{x}}$ is assigned
to the fillable slot $\some{\bddsort{\bot}{\patsort}{1}}$.
In these cases, the input reduction is returned as is
because it is consistent with its assigned slot.
In other cases, filling may additionally repair the given reduction
with additional grout to bridge any multiplicity or sort inconsistencies.
\autoref{fig:meldr-fill-grout} shows how pushing $\sbintile[\patcolor]{:}$
against the stack $\dlroot \lessdot_{\none} \spretile{let}$
leads to the slot $\some{\bddsort{\bot}{\patsort}{1}}$
getting filled instead with convex grout $\snode{\ophole^\patsort}$.

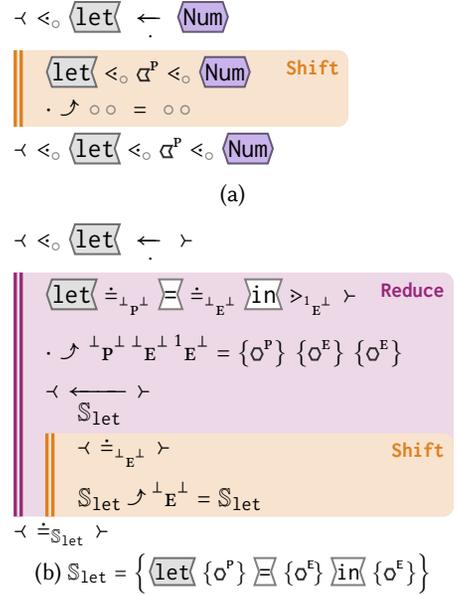
\begin{wrapfigure}[22]{R}{0.44\textwidth}
  ${}$
  \hfill
  \begin{minipage}{0.42\textwidth}
  \newcommand{\dlrt}{\vphantom{\mathstrut}\dlroot}
    \newcommand{\twidth}{3.6cm}
    \begin{subfigure}{\linewidth}
    \begin{tikzpicture}[
  box/.style={rectangle, inner sep=0pt, text width=\twidth},
]
  \node (pushTwo) [box] {
    $\pushj{
      \dlroot \lessdot_{\none} \spretile{let}
    }{\snil}{\soptile[\typcolor]{Num}}$
  };
  \node (shiftTwo) [box, xshift=\pindent, below=of pushTwo] {
    $\spretile{let} \lessdot_{\none} \prehole^\patsort
    \lessdot_{\none} \soptile[\typcolor]{Num}$
  };
  \node (shiftTwo2) [box, below=of shiftTwo] {
    $\pfill{\snil}{\none\,\none}\ =\ \none\,\none$
  };
  \node (pushMult) [box, xshift=-\pindent, below=of shiftTwo2.south west, anchor=north west, text width=\twidth] {
    $\dlroot \lessdot_{\none} \spretile{let}
    \lessdot_{\none}
    \prehole^\patsort
    \lessdot_{\none} \soptile[\typcolor]{Num}$
  };
  \begin{scope}[on background layer]
    \shiftbox{shift2}{(shiftTwo) (shiftTwo2)}
  \end{scope}
\end{tikzpicture}
    \caption{}
    \end{subfigure}
    \vspace{-8pt}

    \renewcommand{\twidth}{5cm}
    \begin{subfigure}{\linewidth}
    \begin{tikzpicture}[
  box/.style={rectangle, inner sep=0pt, text width=\twidth},
]
\node (pushMult) [box, xshift=-\pindent] {
  $\pushj{\dlroot \lessdot_{\none} \spretile{let}}{\snil}{\drroot}$
};
\node (reduceTwo) [box, xshift=\pindent, below=of pushMult.south west,
  anchor=north west] {
  $\spretile{let}
  \doteq_{\botsort{\patsort}} \sbintile[white]{=}
  \doteq_{\botsort{\expsort}} \sbintile[white]{in}
  \gtrdot_{\bddsort{1}{\expsort}{\bot}} \drroot$
};
\node (reduceTwo2) [box, below=of reduceTwo] {
  $\snil
  \rcurvearrowne
    \botsort{\patsort}\,\botsort{\expsort}\,\bddsort{1}{\expsort}{\bot}
  =
    \snode{\ophole^\patsort}\,\snode{\ophole^{\expsort}}\,\snode{\ophole^{\expsort}}$
};
\node (pushMult2) [box, below=of reduceTwo2
] {
  $\pushl{\dlrt}{\dsty{\term_{\texttt{let}}}}{\drroot}$
};
\node (shiftMult) [
  box,
  text width={\twidth-\pindent},
  below right=4pt and \pindent of pushMult2.south west,
  anchor=north west
] {
  $\dlroot \doteq_{\botsort{\expsort}} \drroot$
};
\node (shiftMult2) [
  box,
  text width={\twidth-\pindent},
  below=2pt of shiftMult,
] {
  $\pfill{\term_{\texttt{let}}\!}{\!\botsort{\expsort}} = \term_{\texttt{let}}$
};
\node (pushEnd) [
  box,
  text width={\twidth+\pindent},
  below left=4pt and 2*\pindent of shiftMult2.south west,
  anchor=north west
] {
  $\dlroot \doteq_{\term_{\texttt{let}}} \drroot$
};
\begin{scope}[on background layer]
  \reducebox{reduce2}{(reduceTwo) (pushMult2) (shiftMult2)}
  \shiftbox{shiftmult}{(shiftMult) (shiftMult2)}
\end{scope}
\end{tikzpicture}
    \vspace{-1.5em}
    \caption{$\term_{\texttt{let}} = \snode{
      \spretile{let}\snode{\ophole^\patsort}\sbintile[white]{=}\snode{\ophole^{\expsort}\!}\sbintile[white]{in}\snode{\ophole^{\expsort}}
    }$}
    \end{subfigure}
  \vspace{-1.5em}
  \caption{\tylrcore traces with multi-step precedence walks}
  \label{fig:meldr-walk}
  \end{minipage}
\end{wrapfigure}

\autoref{fig:meldr-walk} shows how \tylrcore generalizes the single-step
precedence comparisons of the original algorithm to multi-step
precedence \emph{walks}.
Where the original method would get stuck trying to push
$\soptile[\typcolor]{Num}$ or $\drroot$
against the stack $\dlroot \lessdot_{\none} \spretile{let}$
because $\spretile{let}$ is precedence-comparable with either,
\tylrcore can proceed because it finds extended walks like
$\spretile{let} \lessdot_{\none} \prehole^\patsort \lessdot_{\none}
\soptile[\typcolor]{Num}$
and $\spretile{let}
      \doteq_{\botsort{\patsort}} \sbintile[white]{=}
      \doteq_{\botsort{\expsort}} \sbintile[white]{in}
      \gtrdot_{\bddsort{1}{\expsort}{\bot}} \drroot$.
In both of these cases, the fill operation has multiple ways
of assigning the initial reduction $\snode{\soptile[\patcolor]{x}}$
to the traversed slots,
as determined by the rule \rulename{Fill-Partition}.
Whichever walk is chosen and however its slots are filled
(\autoref{sec:min-oblig}),
the intermediate terminals and filled slots
traversed between the comparands
form the completion \tylrcore uses to repair the input.


A subtle but consequential difference between \tylrcore and OP parsing
lies in our definition of \rulename{Reduce}:
\tylrcore does not require that the comparison walk
conclude with the pushed terminal $\cft$---%
any concluding terminal (notated $\medsquare$) is sufficient.
This relaxation allows \tylrcore to fall back to \rulename{Reduce}
when the stack head and pushed terminal are not monotonically
precedence-walkable,
completing and reducing the head stack level
and deferring the comparison to something further up the stack.
An example of this is shown in the first \rlabel step
in \autoref{fig:fix-incomplete} that handles pushing $\spretile{let}$
against the stack $\dlroot \lessdot_{\none} \soptile{2}$.
As shown in the example and in our metatheory,
this recursive deferral is guaranteed to conclude eventually with
the base rule \rulename{Shift},
thanks to the various grout forms that can accommodate
both the accumulated reduction and the pushed terminal.
This fallback to completion and reduction is a sort of
opposite of ``panic mode'', which is forced instead to drop
parts of the stack without the multiplicity-handling
guarantees of grout.

\subsubsection{Sound and Total} \label{sec:sound-total}
Altogether, molded tiles, injected grout, and our filling and walking extensions
of OP parsing guarantee that \tylrcore can complete and reduce any sequence of
input tiles into a well-formed term.


\begin{lemma}[Pushing is Sound and Total] \label{lemma:push-total-0}
  For all well-formed stacks $\wf{\stack}$ and tiles $\labl$,
  there exists well-formed stack $\wf{\stack'}$
  such that $\push{\stack}{\snil}{\labl}{\stack'}$.
\end{lemma}

\begin{theorem}[Parsing is Sound and Total]
\label{theorem:totality}
  For all well-formed stacks $\wf{\stack}$ and tile sequences $\s{\labl}$,
  there exists well-formed stack $\wf{\dlroot\doteq_{\term}\drroot}$
  such that $\parse{\stack}{\s{\labl}\,\drroot}{\dlroot\doteq_{\term}\drroot}$.
\end{theorem}

\noindent
The traces in \autoref{fig:fix-unsound-incomplete} illustrate
this guarantee for the failed examples in \autoref{fig:simple-incomplete}
and \autoref{fig:simple-unsound}.

\input{fig/parsing-example}

\section{From \tylrcore to \tylr} \label{sec:implementation}

\tylrcore formalizes a nondeterministic parser of tile sequences, possibly completing them with some choice of grout and additional tiles, and we showed that the resulting term is grammatical and guaranteed to exist.
To turn this into a deterministic parser of textual input, we must answer the following questions. (A) How does the parser ``mold'' raw text into the tiles to be parsed, in particular when numerous grammatically unique tiles share a common textual form? (B) How does the parser rank and choose among different possible completions?

Moreover, \tylrcore assumes a batch processing context, where the entire input is parsed left-to-right from scratch.
Further questions arise when incorporating \tylrcore into an interactive editor like \tylr. (C) How might existing structures and completions guide or constrain subsequent molding and completion choices? (D) How does the user interact with the system-chosen completions, in particular when it differs from their intent?

This section describes how we addressed these questions in our implementation of \tylr.
Subsequently, Section~\ref{sec:user-study} presents the user study we conducted to evaluate these decisions.

\subsection{Minimizing Obligations} \label{sec:min-oblig}
Guiding \tylr's various decisions is a simple principle: \emph{minimize obligations}.
Obligations serve not only to scaffold and complete partial structures, but also as a useful metric for resolving ambiguities.
Because \tylrcore is total, we may adopt the simple strategy of trying every choice at each juncture---setting aside efficiency concerns for the moment---and taking the one that inserts the fewest (and removes the most) obligations.

Each type of obligation is weighted differently.
Recall from \autoref{sec:design-overview} that the various forms of obligations can be viewed as indicators of \emph{multiplicity} and \emph{sort} inconsistencies between the top-down expectations of the grammar and the bottom-up reductions of the input:
\begin{itemize}
  \itemsep 0pt
  \item Operand grout $\ophole$ indicate there is no term where one is expected ($0 = \bullet < 1$).
  \item Ghosts indicate there is a partial term where one is expected ($0 < \bullet < 1$).
  \item Prefix $\prehole$ and postfix $\poshole$ grout indicate there is a term
  as expected ($\bullet = 1$), but of the wrong sort.
  \item Infix grout $\binhole$ indicate there are multiple terms where one is expected ($1 < \bullet$).
\end{itemize}
The obligations are listed above in order of increasing weight class.
Given two sets of changes in obligations, we compare them lexicographically from highest to lowest weight class.
The principle underlying this ordering is \emph{context preservation}: lower-weighted obligations like operand grout and ghosts are introduced to complete a form independent of its context, whereas higher-weighted obligations like pre-, post-, and infix grout appear when the current context cannot accommodate a form and must change.

\paragraph*{Molding Tiles}
\begin{wrapfigure}[5]{R}{0.52\textwidth}
\setlength{\tabcolsep}{3pt}
\vspace{-1em}
\begin{tabular}{rcl}
  $\pushl{\dlroot \ltop{\none} \spretile{let}}{\none}{\spretile{(}}$ & $=$ &
  $\dlroot \ltop{\none} \spretile{let} \eqop{\nnnode{\ophole^{\patsort}}} \sbintile[white]{=} \ltop{\none} \spretile{(}$ \\
  $\pushl{\dlroot \ltop{\none} \spretile{let}}{\none}{\spretile[\patcolor]{(}}$ & $=$ &
  $\dlroot \ltop{\none} \spretile{let} \ltop{\none} \spretile[\patcolor]{(}$ \\
  $\pushl{\dlroot \ltop{\none} \spretile{let}}{\none}{\spretile[\typcolor]{(}}$ & $=$ &
  $\dlroot \ltop{\none} \spretile{let} \ltop{\none} \prehole^{\patsort} \ltop{\none}\spretile[\typcolor]{(}$
\end{tabular}
\end{wrapfigure}
When a token is inserted, \tylr looks up which tiles in the grammar share the same textual form (typically only a few) and considers the consequences of parsing each one.
For example, when editing Hazel (\autoref{fig:pbg-example}),
suppose the token \code{(} is inserted against the stack
$\dlroot \ltop{\none} \spretile{let}$.
There are three distinct tiles with the same textual label,
each of a different sort.
Pushing each tile against the stack leads to the following minimal outcomes, where ghosts are indicated with a white background:
The first option introduces an operand hole and a ghost, while the third introduces a prefix grout.
The clear winner is the second option, an opening parenthesis of pattern sort, which introduces no new obligations.

\paragraph*{Choosing Completions}
The parsing rules allow for arbitrary walks through the precedence relation graph, with each step from the head of the stack inserting one or more new obligations.
For example, the following are all valid precedence walks when applying the \texttt{Shift} rule to derive $\pushl{\dlroot \ltop{\none} \spretile{(}}{\none}{\sbintile{,}}$:
\[\begin{array}{rlrl}
  (A) & \spretile{(} \eqop{\nnnode{\ophole^\expsort}} \sbintile{,}
  & (D) & \spretile{(} \ltop{\none} \spretile[white]{(} \eqop{\nnnode{\ophole^\expsort}} \sbintile{,} \\
  (B) & \spretile{(} \eqop{\nnnode{\ophole^\expsort}} \sbintile[white]{,} \eqop{\nnnode{\ophole^\expsort}} \sbintile{,}
  & (E) & \spretile{(} \ltop{\none} \spretile[white]{let} \eqop{\nnnode{\ophole^\patsort}} \sbintile[white]{=} \ltop{\none} \spretile[white]{(} \eqop{\nnnode{\ophole^\expsort}} \sbintile{,} \\
  (C) & \spretile{(} \eqop{\nnnode{\ophole^\expsort}} \sbintile[white]{,}  \eqop{\nnnode{\ophole^\expsort}} \sbintile[white]{,} \eqop{\nnnode{\ophole^\expsort}} \sbintile{,}
  & (F) & \spretile{(} \ltop{\none} \prehole^{\expsort} \ltop{\none} \spretile[white]{(} \eqop{\nnnode{\ophole^\expsort}} \sbintile{,}
\end{array}\]
\tylr limits the walks considered to those of shortest length found via breadth-first search, ruling out options like (B) and (C).
\tylr also filters out walks with strictly $\ltop{}$-intermediate tile levels, such as $\spretile[white]{let} \eqop{\nnnode{\ophole^\patsort}} \sbintile[white]{=}$ in (E),
preferring instead to abstract such possibilities with grout like in (F).
The remaining walks are subsequently sorted by height and length to break ties in obligation deltas when filling in any accumulated terms.


\subsection{Maintaining Obligations}
Total error-correcting parsing lends itself to a continuously structured editing experience. Indeed, our obligation design is inspired directly by numerous structure editor designs~\cite{moonTylrTinyTilebased2022,moonGradualStructureEditing2023}.
In this setting, questions arise as to how to maintain and remove existing obligations to produce a smooth editing experience, and how to insert new obligations around existing structures.

The main concern regards inserting, maintining, and removing ghosts, as the minimal requisite grout needed to complete an edit state is fully determined if all requisite tiles are in place.
Ghost maintenance concerns roughly divide into three areas.
The first concerns inserting ghost replacements after deleting requisite tiles---in this case, to maximize continuity, \tylr replaces deleted requisite tiles with ghosts in the same position.
The second concerns inserting fresh ghosts around existing structures on insertion.
As mentioned in \autoref{sec:design-overview}, \tylr uses a simple policy of inserting any pending ghosts at the first newline following the insertion---all other positioning concerns are deferred to obligation minimization and completion choices.

The third area concerns removing existing ghosts when they are no longer needed.
\tylr models its edit state as a pair of prefix and suffix stacks, where the suffix is reparsed after each change.
There are two cases to consider.
The first is when a ghost in the suffix becomes redundant---for example, when
inserting \sposttile{)} between the stacks
\hfill
$\dlroot \ltop{\none} \spretile{(} \ltop{\none} \soptile{2}
\hskip10pt \boldsymbol{\vert} \hskip10pt \sbintile{+}
\gtop{\snode{\soptile{3}}} \sposttile[white]{)} \gtrdot_{\none} \drroot$ \\
When a ghost is encountered in the suffix, \tylr pushes it onto the prefix stack as if it were a solid tile and removes it if it cannot find an $\eqop{}$-match.
The second case is when a ghost in the prefix becomes redundant---for example,
when inserting $\sbintile{in}$ between the stacks

\vspace{-1em}
\[\arraycolsep=10pt\begin{array}{rcl}
\dlroot \ltop{\none} \spretile{let} \eqop{\nnnode{\ophole^\patsort}} \sbintile{=} \eqop{\nnnode{\ophole^\expsort}} \sbintile[white]{in} \ltop{\none} \soptile{4}
& \boldsymbol{\vert} & \drroot
\end{array}\]
\vspace{-1em}

\noindent
When \tylr pushes $\sbintile{in}$ onto the prefix and encounters the ghost $\sbintile[white]{in}$, \tylr tries removing it and commits to the removal if the pushed $\sbintile{in}$ finds an $\eqop{}$-match.
Our current design is limited in that it provides no way to removing ghosts directly, instead requiring the user to insert a solid tile replacement elsewhere, the consequences of which we discuss in more detail in \autoref{sec:user-study}.

\subsection{Performance}

The focus of this paper is on the conceptual, theoretical, and interaction
design of \talltylr.
We did little to optimize its performance beyond what was needed for
responsiveness on relatively small programs (less than 100 lines)
and make no strong claims, though we report basic performance numbers in the
supplemental appendix for the sake of completeness.
There are high-level reasons to believe that this approach would scale
performantly.
Standard OP parsing scales linearly with the input and moreover enjoys the property of
local parsability which greatly simplifies incrementalization and
parallelization
\cite{barenghiParallelParsingOperator2013,barenghiParallelParsingMade2015}.
Meanwhile, prior work on enumerating local repairs
\cite{considineSyntaxRepairIdempotent,diekmannDonPanicBetter2020}
suggests this can be done efficiently.
We leave detailed optimizations along these lines to future work.

\section{User Study} \label{sec:user-study}

Prior work on error-handling parsing does not explicitly consider user
interfaces for representing and interacting with parse errors.
In \tylr, we explore a novel UI that materializes obligation-based repairs
as inline completions.
This requires making choices about where to insert obligations in situations
underdetermined in our formal model.
Providing a good user experience thus requires choosing heuristics which
adequately anticipate user intent across real-world coding tasks,
as well as providing affordances to correct obligation placement in cases where
these heuristics fail.

We took a maximally structured approach, inserting or removing  obligations on every code edit so that the edit state remains structured at all times. While this strategy is desirable in that it allows the possibility of continuous language server feedback, it is relatively aggressive, raising questions about the impact of frequent insertion and removal of elements within the text flow.

We considered the following questions:

\begin{itemize}[noitemsep,nolistsep]
    \item[\textbf{Q1}] Do users generally find \tylr usable and useful across a range of naturalistic code insertion and modification tasks?
    \item[\textbf{Q2}] During which kinds of editing operations do users find specific \tylr mechanisms useful, confusing, or cumbersome?
\end{itemize}


\subsection{Study Design}

We ran a remote user study, recording participants' screens as they performed nine code transcription and modification tasks. Each sixty minute session began with a series of pre-recorded videos outlining the motivation for \tylr and its essential editor mechanisms. To reduce jargon, we referred to syntactic obligations as `placeholders' in the study materials.

After the introduction, users performed a practice task to familiarize themselves with the study setup. For each task, they were asked to read and internalize their goal, ask any clarifying questions, and then proceed, pausing after each task to relay any reflections, possibly replaying their actions. At the end of the study, participants were sent a link to an exit survey.

We piloted a shorter version of this study with an earlier prototype;
we have included quotes from one previous participant (labeled
\textbf{P0}) in \autoref{sec:user-study-results}.

\subsubsection{Participants}

We recruited participants with self-reported experience in expression-based
languages by posting on Bluesky, Mastodon, and X offering
compensation of \$25 USD for a 1-hour session. Our study had 9 participants (8
male, 1 non-binary); ages 19-38 ($\mu = 28$); 5-25 years of programming experience
($\mu = 13$), and 1-15 years of functional programming
experience ($\mu = 6$).

\subsubsection{Tasks}
We chose nine code editing tasks (\autoref{tab:tasks}) intended to reflect
real-world use patterns, six of which are adapted from a
previous study \cite{moonGradualStructureEditing2023}. As well as simple entry
and spot-editing tasks, we included more complex goals most economically
accomplished by multiple
edits which temporarily break term structure; an example is shown in
\autoref{fig:task-example}. Since the language syntax is new to study
participants, we asked them to carefully read the desired end state, and to ask
the study administrator any questions about the semantics of the requested
transformation.

\begin{table}
  \centering
  \caption{Study tasks including line count change between initial and target
  states}  \label{tab:tasks}
  \vspace{-0.2em}
  \begin{tabular}{|p{0.05\textwidth}|p{0.15\textwidth}|p{0.55\textwidth}|p{0.06\textwidth}|}
      \hline
      \textbf{Task} & \textbf{Type} & \textbf{Description} & \textbf{Lines} \\
      \hline
       1 & Transcription & Linear entry of data pipeline & +5 \\
      \hline
       2 & Modification & Rearrange the elements of a data pipeline & +2 -2 \\
      \hline
       3 & Transcription & Linear entry of geometry processing function & +5 \\
      \hline
       4 & Modification & Extract helper function & +6 -3 \\
      \hline
       5 & Transcription & Linear entry of graphics function definition & +7 \\
      \hline
       6 & Modification & Refactor a function to remove redundancy & +7 -7\\
      \hline
       7 & Modification & Add a sum type and add branching to linear code & +9 -3 \\
      \hline
       8 & Modification & Uncurry function definition and type annotation & +2 -2\\
      \hline
       9 & Modification & Fuse a series of transformations & +4 -4 \\
      \hline
  \end{tabular}
\end{table}

\begin{figure}
  \centering
  \includegraphics[width=0.8\linewidth]{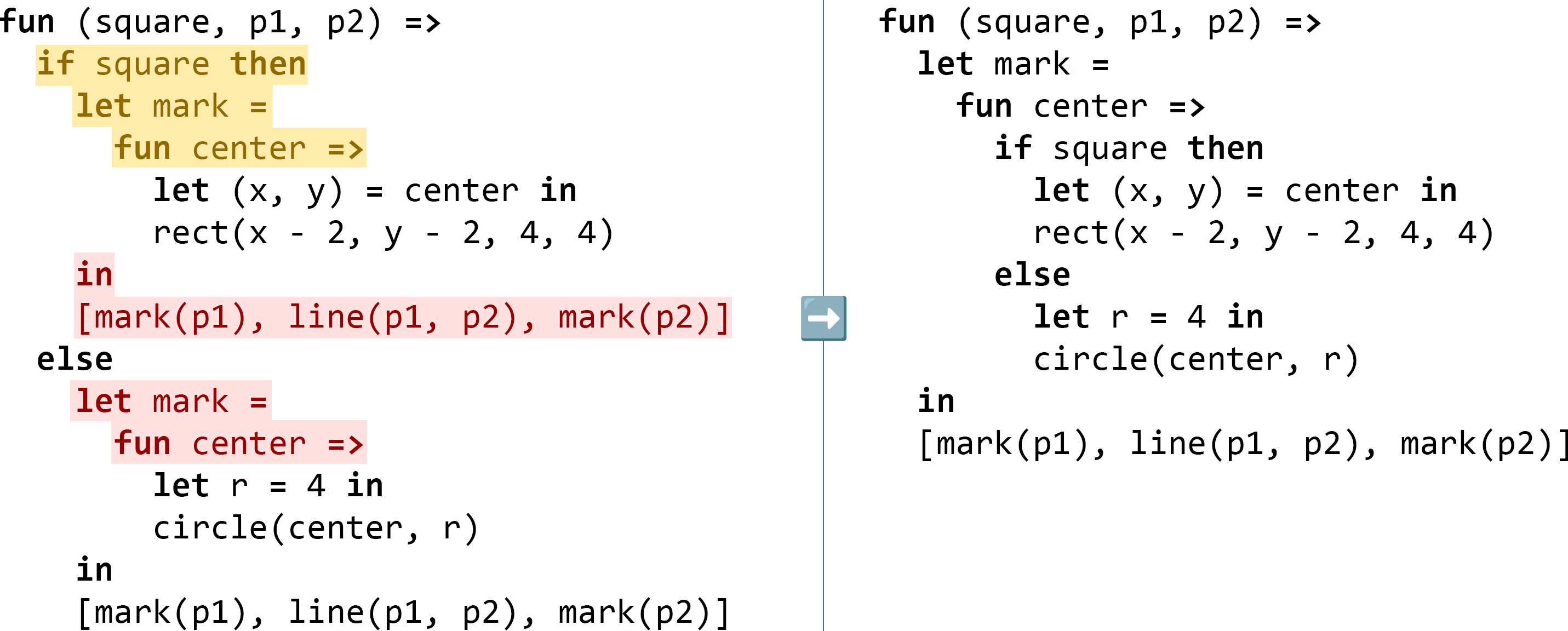}
  \caption{Task 6 asked participants to refactor a function from the start state
  on the left to the target state on the right. Highlighting is added here for readability and was not present in the study.}
  \label{fig:task-example}
\end{figure}
After each task, participants were asked to reflect on any unexpected or interesting behaviors they encountered. Since we knew participants would approach tasks via different editing strategies, for some tasks we provided a follow-up reflection slide illustrating a specific edit and ensuing placeholder insertion in order to more directly solicit opinions on particular heuristics.

\subsection{Results} \label{sec:user-study-results}

\begin{figure}
    \centering
    \includegraphics[width=1.0\linewidth]{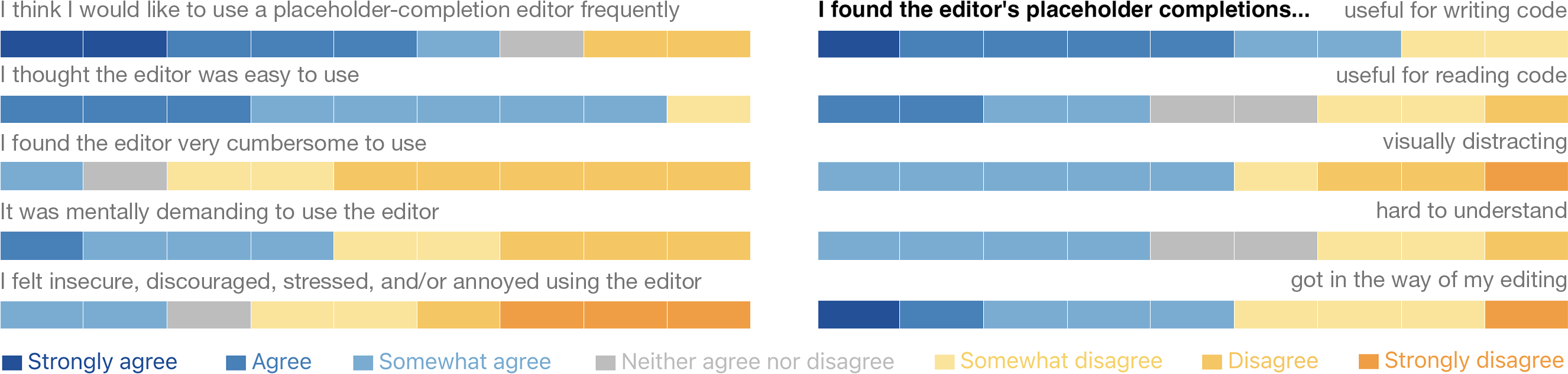}
    \caption{Participant opinions on \tylr's general usability (left) and reactions to placeholders (right)}
    \label{fig:likert}
\end{figure}

Participant assessments of overall usability are summarized in \autoref{fig:likert} , with eight of nine participants at least somewhat agreeing that \tylr was easy to use. However four participants found the editor at least somewhat mentally demanding, with two experiencing stress or annoyance. This may have been impacted by bugs in the prototype. Promisingly, six participants reported desire to frequently use an editor supporting placeholder completions.

Of those who found \tylr easy to use, \textbf{P8} said ``the typing experience felt premium, bespoke, closer to video game than text editor''. With respect to obligations, \textbf{P0} remarked ``I don’t think I expect an editor to exactly pinpoint what fix I need to make. How would it know what I intended? But I like that this allows me to instantly see what it is that you’re assuming I meant.''

Seven participants at least somewhat agreed obligations helped while writing code. Participants found placeholders particularly helpful during left-to-right entry, with \textbf{P8} saying ``I always would like placeholder completions until I have a complete expression!''. Some attributed this to lowered mental load - \textbf{P0} liked that they could ``turn my brain off a bit while typing them out''. \textbf{P2} appreciated that they ``just have to remember how to write the first token in a term'' due to ghost insertion.

However, five participants felt obligations at least somewhat got in the way while modifying code, and half of participants found placeholders at least somewhat visually distracting and hard to understand. \textbf{P9} said that ``When I've created an invalid state [during refactoring], the placeholders often didn't feel helpful''. \textbf{P7} agreed, saying ``it seemed to just break and also just jumble up the screen which is when I probably would've preferred a normal editor with red text''.

We identified a number of specific scenarios where placeholders proved problematic. Three are included below, and others (along with more participant reactions) are located in the appendix.

\paragraph*{Failed attempts to bust ghosts directly}

Although our intended workflow to address ghosts in unwanted positions is for the user to insert the delimiter where they wanted, leaving \tylr to clean up the misplaced ghost, many participants found themselves wanting to interact with ghosts more directly. At least five participants attempted to directly delete ghosts in one or more tasks, despite a caution against this in the introductory video. \textbf{P4} felt ``the placeholder completions felt like they were there to help the computer, not me'', saying that ``where I was unable to delete the ghosts and grout, it took a while to figure out how to get rid of them.'' This was particularly felt when the obligations were inserted in the middle of a complex edit, with \textbf{P5} saying ``the editor sometimes added a lot of holes while I was in the middle of editing an expression, which I instinctively tried to delete.'' This issue was exacerbated by a bug in the \tylr prototype that sometimes prevented ghost cleanup in the presence of nested ghost delimiters.


\vspace{-0.2pt}
\paragraph*{Ghosts are sometimes too tied to the place where they were deleted}

\begin{wrapfigure}[14]{r}{0.6\textwidth}
  ${}$\hfill
  \begin{minipage}{0.58\textwidth}
  \centering
  \vspace{-1em}
  \includegraphics[width=\linewidth]{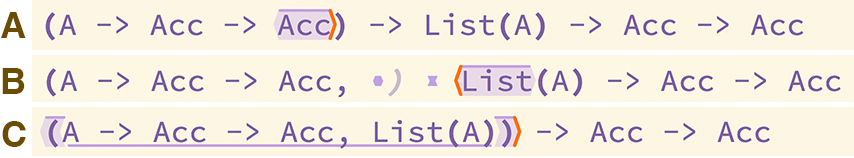}
  \vspace{-1em}
  \caption{During Task 8, participants must modify type annotation (A) to uncurried form. If this is approached in a left-to-right fashion, the user will insert a comma (creating an operand obligation), delete the parenthesis (leaving a ghost), and delete the type arrow (creating a infix obligation) as shown in (B). If the ghost parenthesis did not retain its location, the grout could be combined and cleaned up. This cleanup only occurs when the user re-inserts the closing parenthesis (C).}
  \label{fig:scenario-curry}
  \end{minipage}
\end{wrapfigure}
Although participants generally liked that ghost delimiters remembered their original positions, this did lead to some confusing in-between (``tween'') states. \autoref{fig:scenario-curry} shows a scenario encountered by three participants during Task 8. Here users found the tween state distracting, sometimes attempting unsuccessfully to delete the obligations directly, although all eventually moved on to complete the task successfully.






\paragraph*{Uncertainty around triggering token remolding}
In \tylr users must press space after entering a leading delimiter like \code{let} before the associated grout and trailing delimiter ghosts are inserted. This special treatment of space is primarily to permit entry of tokens beginning with \code{let}, and secondarily to mitigate jarring changes by limiting them to occur only when certain 'action keys' are pressed. In our study this behavior was unproblematic when writing code, but for editing it caused issues, particularly during typo correction. For example, during Task 1 participant \textbf{P3} mistyped an operator requiring space and continued on to the end of the line. They later went back to correct it, but since there was already a space afterwards, they didn't bother to press space after the correction, resulting in remaining infix obligation and the operator  left unmolded. Similar issues confusion for at least 3 participants, including  \textbf{P7} who noted that "space has a learning curve".



\subsection{Threats to Validity}


Our participants are few and drawn from social media networks already self-selected for affinity towards novel programming tools and concepts.

Our study is a synthetic representation of real coding tasks in the sense that
participants' attention is artificially divided. Where programmers might
otherwise be focused on writing, they now must go back and forth between the
slides and the editor. This might make \tylr look both worse and better, in that
participants cannot devote their full attention to editor mechanics, but also
may avoid being distracted by confusing obligations during awkward tween states.

For the purpose of reducing jargon, we referred to syntactic obligations as `placeholders' in the study materials, a choice which may have backfired as some participants seemed to expect that these placeholders should be less insistent and easier to dismiss.

There are also a number of factors which complicate clearly ascribing participant difficulties to \tylr mechanics, including: (1) unfamiliar syntax leading to task confusion and higher rates of typos; (2) bugs in the editor interfering with participants' accurately internalizing editor mechanics; (3) learning curve and difficulty internalizing novel concepts within 60 minutes.

\section{Related Work} \label{sec:related-work}




\paragraph*{OP Parsing}
Operator-precedence (OP) parsing \cite{floydSyntacticAnalysisOperator1963} is an
early form of bottom-up parsing, which proceeds by iteratively reducing the
input sequence of terminal symbols to a single start nonterminal.
The goal at each iteration is to find the next \emph{handle}, a subsequence of symbols matching the righthand side of a grammar production rule, and replace it with the lefthand side nonterminal.
Unlike other parsing methods, OP parsing elides the reduced nonterminals.
Handles are identified as chains of precedence-related terminals of the form
$\labl_L \ltop{} \labl_0 \eqop{} \cdots \eqop{} \labl_k \gtop{} \labl_R$, where
$\labl_L$ and $\labl_R$ delimit the handle's terminals $\labl_0\dots\labl_k$.
Where typically the handle would be replaced by its reduced nonterminal,
OP parsing instead replaces it with a precedence comparison $\odot\in\{\lessdot,
\doteq, \gtrdot\}$ describing the relation between the delimiters $\labl_L$ and $\labl_R$.

\citet{levyCompleteOperatorPrecedence1975} observed that OP parsing is consequently unsound,\footnote{
  Levy took the goal of parsing to be detecting invalid sentences rather than valid ones, and hence called ``complete'' (detect all invalid sentences) what we call ``sound'' (detect only valid sentences) in this work.
} as we illustrated in \autoref{fig:simple-unsound}.
In a resolution similar to ours, \citet{hendersonExtendedOperatorPrecedence1976} split each precedence relation $\odot$ into two relations $\odot_1$ and $\odot_2$, the difference being whether a reduced nonterminal is expected between the related terminals.
Our approach generalizes this idea by indexing each relation by the optional
nonterminal itself---the slot-filling operation in \tylrcore
(\autoref{fig:fill}) uses this extra information to validate
that the bottom-up accumulated reduction meets the top-down slot's requirements.

\paragraph*{Precedence Annotations}
In OP parsing, the precedence comparisons are derived from
the derivation patterns of an unannotated CFG.
In other words, these methods expect operator precedence conventions to be encoded in the CFG's nonterminal dependency structure.
This is tedious to do by hand and leads to a profileration of nonterminals, one
for each precedence level, that obscure the language's natural organization
into semantically meaningful sorts.
The rule \texttt{Produce-Tighten} in our PBG-to-CFG elaboration
(\autoref{fig:compile-precedence}) automatically extracts
these dependency structures from a sort-organized PBG.
Not only does this organization benefit grammar authoring and documentation,
it also helps our system repair errors using a concentration of grout forms
that are semantically meaningful and thereby more easily user-communicable.


Predominant interpretations of precedence annotations are specific to the parsing method---in LR parsers generators, for example, the annotations are used
to resolve shift/reduce conflicts in the generated action table \cite{ahoDeterministicParsingAmbiguous1975}.
Less common are parser-independent semantics, such as our PBG-to-CFG elaboration
(\autoref{fig:compile-precedence}).
\autoref{sec:derive-prec} described how prior semantics
by \citet{desouzaamorimMultipurposeSyntaxDefinition2020} (for the language
workbench Spoofax \cite{katsSpoofaxLanguageWorkbench2010})
and \citet{danielssonParsingMixfixOperators2011} (for mixfix operators in Agda)
are unnecessarily restrictive in how they handle prefix and postfix operators.
Making similar observations,
\citet{aasaPrecedencesSpecificationsImplementations1995}
defined the precedence weights that we recapitulate
in our elaborated reduction rules (\autoref{fig:compile-precedence}).
Aasa used these measure to define when a derivation tree
of the underlying unannotated grammar is \emph{precedence-correct}
according to the annotations.
Separately, Aasa also defined a translation from annotated to unannotated
grammars, but this translation follows a different, more complicated design, an opinion we share with
\citet{danielssonParsingMixfixOperators2011}.
Our elaboration re-centers Aasa's precedence weights via
a novel bidirectional organization.

\paragraph*{Error Handling}
Modern parsers are expected to be able to \emph{recover} from errors (\ie
unexpected tokens) and continue parsing around the error site.
Most recovery methods attempt to \emph{repair} the input text around the error, differing in what repairs they consider and how they choose among them.
The simple ``panic mode'' method \cite{ahoCompilersPrinciplesTechniques2007,gruneParsingTechniquesPractical2008} limits itself to repair by deletion, dropping tokens around the error until parsing can resume from some prior state.
While simple to implement, this method is liable to skip large regions of code,
as illustrated in \autoref{fig:syntax-error}B, leaving the programmer without
downstream semantic analysis.
\tylrcore takes an opposite approach, where the tokens dropped by a
panicking parser are instead completed, reduced, and propagated up the stack
with the assurance that, eventually, grout can be used to join together
extraneous terms.

To minimize skipped input, more sophisticated methods
\cite{grahamPracticalSyntacticError1975,considineSyntaxRepairIdempotent,fischerLocallyLeastCostLRError1979}
consider the full range of possible repairs around an error,
including insertions as well as deletions,
and pick one of least cost according to a language-specific cost
vector of token modifications or else textual edit distance.
Most similar to our work is the FMQ method \cite{fischerEfficientLL1Error1980}, which performs repairs using only insertions.
This work defines the \emph{insert-correctable} class of grammars against which
any input text can be repaired by insertions to grammatical form---in our work,
grout injection (\autoref{fig:add-grout}) systematically relaxes any grammar to
be insert-correctable.

Across their variations, these prior repair-based recovery methods limit themselves to purely textual repairs.
This can lead to combinatorial explosion in the space of possible insertion repairs, as illustrated in \autoref{fig:syntax-error}C.
While these repairs can be enumerated efficiently in practice \cite{diekmannDonPanicBetter2020,considineSyntaxRepairIdempotent}, prior work does not consider the question of how to effectively surface these repairs to the programmer.
Our approach is novel in its use of abstract syntactic obligations to compress, communicate, and rank the space of possible repairs.

\paragraph*{Structure Editing}
\talltylr continues a series of design experiments in increasingly
flexible and text-like structure editing.
Its precedessors, \tinytylr \cite{moonTylrTinyTilebased2022} and \teentylr
\cite{moonGradualStructureEditing2023},
proposed and refined a model of structure editing
in which the primary units of the edit state
are nested spans of matching tokens,
there called \emph{tiles} of matching \emph{shards}.
A design consequence of this particular physical metaphor
was that shards remained matched for life in those editors, which
\citet{moonGradualStructureEditing2023} observed compromised
overall usability.
Relaxing this restriction, and repurposing the term \emph{tile}
for molded tokens, led to the present design.

Error-handling parsing and structure editing are kin in their goals
of maximizing structure, but emphasize quite different aspects of
the design/technical problem space.
Our prior emphasis on structure editor design led here to a unique
approach to error-handling parsing, one that builds on the underexplored
symbol-based perspective of OP parsing (driven by symbol-to-symbol precedence
comparisons),
as opposed to the predominant item-based perspective
of methods like LL/LR (where items refer to points \emph{in between} the
symbols of a production rule, used to define the states of handle-finding automata).
The token-based perspective comes with design advantages, simply because tokens
provide more visual surface area to decorate than zero-width items---%
it is, in our opinion, much easier to display and describe molds to the programmer
than it is to display and describe items and automaton states.

\section{Conclusion} \label{sec:future-work}


This paper presented \tylr,
a tile-based parser and editor generator
that handles errors by completing its input with syntactic obligations.
We developed these ideas precisely in our parsing calculus
\tylrcore, which extends OP parsing with error handling and guarantees a
well-formed result on all inputs---along the way, it offers
a new unified account of operator precedence.
Key components of \tylrcore's assured totality
include relaxing grammaticality with grout,
used to buffer inconsistencies of multiplicity and sort,
and generalizing the single-step comparisons of OP parsing
to multi-step walks that serve as completion-repairs.
We proposed the principle of minimizing obligations
that governs how \talltylr discharges the various choices
required for parsing and handling errors.
Our user study suggested that syntactic obligations generated this way have both
demand and promise, but more design work is needed to give the programmer
more control over their placement and removal, especially when modifying
existing code.
Altogether, this work opens up a significant new design space and we look forward to future design experiments driven by the core ideas introduced in this paper.



\newpage
\section*{Acknowledgements} \label{sec:acknowledgements}
This work was partially funded by the National Science Foundation under Grant No. 2238744. Any opinions, findings, and conclusions or recommendations expressed in this material are those of the author and do not necessarily reflect the views of the National Science Foundation.

\section*{Data Availability Statement} \label{sec:data-availability}

This paper includes an artifact \cite{moonArtifactSyntacticCompletions2025} consisting of the study materials given to
participants in the user study.
These include a slideshow, introductory videos, an exit survey,
and a copy of the \tylr editor prototype,
which can be run in a modern web browser.
A detailed table of contents and instructions are provided
in \texttt{README.md}.

\bibliography{references}

\renewcommand{\chpath}{appendix}

\appendix
\section{Proofs for \autoref{sec:tylrcore}} \label{appendix:proofs}



In this appendix, we will use the following shorthand notations:
\begin{align*}
\lmsym{\cfs}{\cfx} & \ \ \triangleq\ \  \cfs\produce\cfx\dots \\
\adj{\cfx_L}{}{\cfx_R} & \ \ \triangleq\ \ \exists \cfs.~\cfs\produce\dots\cfx_L\cfx_R\dots \\
\rmsym{\cfx}{}{\cfs} & \ \ \triangleq\ \ \cfs\produce\dots\cfx \\
\yields{\cfs}{\cfx} & \ \ \triangleq\ \ \cfs\produce\dots\cfx\dots
\end{align*}

\subsection{Precedence Comparisons}

Lemmas \ref{lemma:homogeneity}-\ref{lemma:no-trespassing} follow by inspection of the elaboration (\autoref{fig:compile-precedence}) and injection (\autoref{fig:add-grout}) rules.

\begin{lemma}[Homogeneity] \label{lemma:homogeneity}
If $\cft_L \doteq_{\medsquare} \cft_R$ then either
\begin{itemize}
\item $\cft_L = \dlroot$ and $\cft_R = \drroot$;
\item $\cft_L = \labl_L$ and $\cft_R = \labl_R$ for some tiles $\labl_L,\labl_R$; or
\item $\cft_L = \groot_L^{\sort}$ and $\cft_R = \groot_R^{\sort}$ for some grout
$\groot_L,\groot_R$ and sort $\sort$.
\end{itemize}
\end{lemma}

\begin{lemma}[Grout Precedence] \label{lemma:grout-prec}
  The following statements hold:
  \begin{itemize}
  \item
  If $\groot$ is left-convex, then $\cft \cmpr_{\o{\cfs}} \groot^{\sort}$ if and only if $\o{\cfs} = \none$ and $\cmpr = {\lessdot}$. \\
  If $\groot$ is right-convex, then $\groot^{\sort} \cmpr_{\o{\cfs}} \cft$ if and only if $\o{\cfs} = \none$ and $\cmpr = {\gtrdot}$.
  \item
  If $\groot_R$ is left-concave,
  then $\cft \doteq_{\o{\cfs}} \groot_R^{\sort}$ if and only if
  $\o{\cfs} = \some{\zerosort{\sort}}$ and $\cft = \groot_L^{\sort}$ is right-concave. \\
  If $\groot_L$ is right-concave,
  then $\groot_L^{\sort} \doteq_{\o{\cfs}} \cft$ if and only if
  $\o{\cfs} = \some{\zerosort{\sort}}$ and $\cft = \groot_R^{\sort}$ is left-concave.
  \end{itemize}
\end{lemma}

\begin{lemma}[Start Matches End] \label{lemma:droot-eq}
  $\dlroot \doteq_{\o{\cfs}} \cft$ if and only if $\o{\cfs} = \some{\botsort{\startSym}}$ and $\cft = \drroot$.
\end{lemma}
\begin{lemma}[No Escaping] \label{lemma:no-escaping}
  There exists no $\cft$ such that $\cft \lessdot \drroot$ or $\dlroot \gtrdot \cft$.
\end{lemma}
\begin{lemma}[No Trespassing] \label{lemma:no-trespassing}
  There exists no $\cft$ and $\cmpr$ such that $\cft\cmpr\dlroot$ or $\drroot\cmpr\cft$.
\end{lemma}

\begin{lemma}[Root-Grout Delimits Tiles] \label{lemma:root-grout-prec}
  If $\groot$ is right-concave,
  then for all tiles $\labl$,
  there exist slots $\iseq{\o{\cfs}_i}{0\leq i \leq k}$
  and tiles $\iseq{\labl_i}{0\leq i \leq k}$
  such that
  $\groot^{\startSym} \lessdot_{\o{\cfs}_0} \labl_0
  \iseq{\doteq_{\o{\cfs}_i}\labl_i}{1\leq i \leq k} = \labl$.
\end{lemma}
\begin{proof}
By \Assumptionref{assum:unique-tiles},
there exists sort $\sort$ and precedence $\precp$
such that ${\dots\labl\dots} \in
\rgxsem{\gram(\sort,\precp)}$.
By \Assumptionref{assum:operator-form},
there exist optional sorts $\iseq{\o{\sort}_i}{0\leq i \leq n+1}$
and tiles $\iseq{\labl_i}{0\leq i \leq n}$ such that $\labl_k = \labl$ for some
$k\leq n$ and
\begin{align}
  \o{\sort}_0 \iseq{\labl_i\o{\sort}_{i+1}}{0\leq i \leq n} \in \rgxsem{\gram(\sort,\precp)} \label{eq:root-grout-prec-derive}
\end{align}
Depending on whether $\o{\sort}_0 = \sort$ and $\o{\sort}_{n+1} = \sort$,
apply one of the elaboration rules in \autoref{fig:compile-precedence} to \eqref{eq:root-grout-prec-derive} to
construct a production rule for $\zerosort{\sort}$, knowing that $\precz
\prec_\sort \precp$ and $\precp \succ_\sort \precz$ as needed.
Across all cases, we can show
\begin{align}
\zerosort{\sort} \produce {\dots\labl_0\iseq{\promotesym{\o{\sort}_i}\labl_i}{1\leq i \leq n}\dots} \label{eq:prec-eq}
\end{align}
Define $\iseq{\o{\cfs}_i = \promotesym{\o{\sort}_i}}{0\leq i \leq n}$.
 \texttt{Prec-EQ} applied to \eqref{eq:prec-eq} gives us
\begin{align}
\iseq{\labl_i \doteq_{\o{\cfs}_i} \labl_{i+1}}{0\leq i < n}
\end{align}

It remains to show $\groot^{\startSym} \lessdot_{\o{\cfs}_0} \labl_0$.
It can be shown using one of the injection rules in \autoref{fig:add-grout} that
$\groot^{\startSym} \adjop \zerosort{\startSym}$.
\end{proof}

\begin{lemma}[Reachability] \label{lemma:root-connected}
  For every tile $\labl$,
  there exist slots $\iseq{\o{\cfs}_i}{1\leq i \leq k}$
  and tiles $\iseq{\labl_i}{1\leq i \leq k}$
  such that
  $\dlroot \iseq{\cmprleq_{\o{\cfs}_i}\labl_i}{1\leq i \leq k} = \labl$.
\end{lemma}

\begin{proof}
Follows from $\dlroot \lessdot_{\none} \prehole^{\startSym}$
(since $\dlroot \adjop \botsort{\startSym} \lmsymop \prehole^{\startSym}$)
and \Lemmaref{lemma:root-grout-prec}.
\end{proof}

\subsection{Precedence Bounds}

Lemmas \ref{lemma:dispute-inherit}-\ref{lemma:left-bounded-produce}
follow by inspection of the elaboration (\autoref{fig:compile-precedence})
and injection (\autoref{fig:add-grout}) rules.

\begin{lemma} \label{lemma:dispute-inherit}
  If $\bddsort{\precp}{\sort}{\medsquare} \produce
  \bddsort{\precq}{\sort}{\precm}\labl\dots$
  then $\precp \pleq_\sort \precq \plt_\sort \precm$.
\end{lemma}


\begin{lemma} \label{lemma:unbounded-sort-transition}
  If $\psq \lmsymop^* \mrn$
  and $\sort \neq \sortr$
  then $\psq \lmsymop^* \botsort{\sortr}$.
\end{lemma}


\begin{lemma} \label{lemma:adj-bot}
  If $\labl \adjop \psq$
  then $\labl \adjop \bddsort{\precp}{\sort}{\bot}$.
\end{lemma}

\begin{lemma} \label{lemma:left-bounded-produce}
  If $\cfs \lmsymop^* \cfsr$ and $\cfs \consistent \sort$
  then either $\cfsr \consistent \sort$
  or $\cfs \produce \cfsr\,\poshole^\sort$.
\end{lemma}

\begin{lemma}\label{lemma:derive-leftmost-yield}
  If $\cfsr \produce^* \o{\cfs}\cft\dots$ then there exists a nonterminal $\cfs$ such that
  \[ \lmsymStar{\cfsr}{\cfs \produce \o{\cfs}\cft\dots}  \]
\end{lemma}

\begin{proof}
We will show a slight generalization:
if $\cfsr \produce^+ \o{\cfs}\cft\dots$ then there exists a nonterminal $\cfs$ such that
$\lmsymStar{\cfsr}{\cfs \produce \o{\cfs}\cft\dots}$.
This is a proper generalization,
despite the use of $\produce^+$ instead of $\produce^*$,
because ${\cfsr}\neq{\o{\cfs}\cft\dots}$
and therefore $\cfsr \produce^* \o{\cfs}\cft\dots$
implies $\cfsr \produce^+ \o{\cfs}\cft\dots$,
which implies $\cfsr\dots \produce^+ \o{\cfs}\cft\dots$.
We will assume here that
every yield step $\s{\cfx}_0 \produce \s{\cfx}_1$ is a leftmost yield,
meaning it rewrites the leftmost nonterminal in $\s{\cfx}_0$.

Induct on the premise ${\cfsr\dots} \produce^+ \o{\cfs}\cft\dots$:
\begin{itemize}
\item Suppose ${\cfsr\dots} \produce \o{\cfs}\cft\dots$.
Our assumption of leftmost yields implies $\cfsr \produce \o{\cfs}\cft\dots$,
so returning $\cfs = \cfsr$ gives
$\lmsymStar{\cfsr}{{\o{\cfs}\cft\dots}}$ as desired.

\item Suppose ${\cfsr\dots} \produce \s{\cfx} \produce^+ \o{\cfs}\cft\dots$.

Further suppose $\s{\cfx} = \cfsr_1\dots$ for some nonterminal $\cfsr_1$.
Then our assumption of leftmost yields implies $\cfsr \lmsymop \cfsr_1$,
and our inductive hypothesis gives us $\cfs$ such that
$\lmsymStar{\cfsr_1}{\cfs \produce \o{\cfs}\cft\dots}$.
Putting it altogether, we have
$\lmsymStar{\cfsr}{\cfs \produce \o{\cfs}\cft\dots}$ as desired.

Otherwise, assume $\s{\cfx} = \cft_1\dots$ for some terminal $\cft_1$.
Given ${\cfsr\dots} \produce \cft_1\dots$,
our assumption of leftmost yields implies $\cfsr \produce \cft_1\dots$.
Given $\cft_1\dots \produce^+ \o{\cfs}\cft\dots$,
it must be that $\o{\cfs} = \none$ and $\cft_1 = \cft$.
It follows that setting $\cfs = \cfsr$ gives
$\lmsymStar{\cfsr}{\cfs \produce {\cft_1\dots = \o{\cfs}\cft\dots}}$
as desired.
\end{itemize}
\end{proof}

\begin{lemma} \label{lemma:lt-bot-witness}
  If $\cft_L \lessdot_{\o{\cfsr}} \cft_R$, then there exists a nonterminal $\cfs =
  \bddsort{\medsquare}{\sort}{\bot}$ such that $\cft_L \adjop \cfs$ and $\cfs \produce^* \o{\cfsr} \cft_R\dots$.
\end{lemma}

\begin{proof}
  Invert the premise $\cft_L \lessdot_{\o{\cfsr}} \cft_R$ to get
  nonterminal $\cfsr = \psq$ such that
  \begin{align}
  \adj{\cft_L}{\cfsr} \label{eq:adj} \\
  \cfsr \produce^* \o{\cfsr} \cft_R\dots \label{eq:derive'}
  \end{align}
  Let $\cfs = \bddsort{\precp}{\sort}{\bot}$.
  Apply Lemma \ref{lemma:adj-bot} to \eqref{eq:adj} to get
  $\adj{\cft_L}{\cfs}$.
  Since $\precq \pgeq_\sort \bot$ by definition,
  we have by rule \texttt{Produce-Tighten}
  that $\cfs \produce \cfsr$,
  which together with \eqref{eq:derive'} gives
  $\cfs \produce^* \o{\cfsr} \cft_R\dots$ as desired.
\end{proof}








\subsection{Proof of \autoref{theorem:prec-cohere}}
For all sorts $\sort$, precedence levels $\precp_L,\precp_R$,
and tiles $\labl_L, \labl_R$
such that $\dots\labl_L \sort \in \rgxsem{\gram(\sort,\precp_L)}$
and ${\sort\labl_R\dots}\in\rgxsem{\gram(\sort,\precp_R)}$,
the following equivalences hold:
\begin{align*}
    \labl_L \lessdot \labl_R &\iff \precp_L < \precp_R \\
    \labl_L \gtrdot \labl_R &\iff \precp_L > \precp_R
\end{align*}
\begin{proof}
   We assume the premises. Beginning with the left-to-right direction of the first equivalence, we also assume that $\labl_L \lessdot \labl_R$. By repeated rule inversion, we obtain $\bddsort{a}{\sort}{b} \reduce \dots\labl_L \bddsort{c}{\sort}{d}$ and $\bddsort{d^+}{\sort}{d^+} \reduce \bddsort{e}{\sort}{f}\labl_R\dots$, for some precedences $a, b, c, d, e, f$ such that $c \leq c^+$ and $d^+ \geq d$ (the result of \rulename{Produce-Tighten}). Applying rule inversion once more to each yields the inequalities $\precp_L < c$ and $c^+ = \min(e, \precp_R)$. Together we have $\precp_L < c \leq c^+ \leq \precp_R$, completing the implication.

    For the other direction, we assume that $\precp_L < \precp_R$. We will first derive $\bddsort{a}{\sort}{\precp_R} \reduce \dots\labl_L \bddsort{\precp_R}{\sort}{\precp_R}$ for some extension $\dots$ and some precedence $a$. In the case that the prefix of $\dots\labl_L \sort$ begins with $\sort$, \rulename{PElab-Infix} is used, with the bounds of each $\chi_i$ being $\top$ except for $\chi_k = \bddsort{\precp_R}{\sort}{\precp_R}$. The premises are satisfied with $\dots\labl_L \sort \in \rgxsem{\gram(\sort,\precp_L)}$, $\top > \precp_L < \precp_R$, and $\precp_L = \min(\precp_L, \precp_R)$. The case when the prefix does not begin with $\sort$ is analogous, with \rulename{PElab-Prefix} used instead.

    From this reduction judgment we obtain $\adj{\labl_L}{\bddsort{\precp_R}{\sort}{\precp_R}}$, the first premise to our desired conclusion. For the second premise, it suffices to show that $\bddsort{a}{\sort}{\precp_R} \reduce \bddsort{\top}{\sort}{\top} \labl_R \dots$, which we obtain from ${\sort\labl_R\dots}\in\rgxsem{\gram(\sort,\precp_R)}$ via \rulename{PElab-Infix} or \rulename{PElab-Postfix} as needed, with $\top > \precp_R < \top$ and $\precp_R = \min(\precp_R, \top)$.

    The second equivalence is proven symmetrically.
\end{proof}


\subsection{Proof of \Lemmaref{lemma:valid-prefix}}
\begin{lemma}[Splice]
  \label{lemma:splice}
  Suppose symbol $\sym$ appears uniquely in regex $\rgx$.
  If  ${\s{\sym}_L \sym \dots}\in\rgxsem{\rgx}$
  and ${\dots \sym \s{\sym}_R}\in\rgxsem{\rgx}$,
  then ${\s{\sym}_L \sym \s{\sym}_R}\in\rgxsem{\rgx}$.
\end{lemma}
\begin{proof}
  Proceed by induction on $\rgx$.
  Start each case by assuming the premises:
  \begin{align}
    {\s{\sym}_L \sym \s{\symy}_R}\in\rgxsem{\rgx} \label{eq:splice-L} \\
    {\s{\symy}_L \sym \s{\sym}_R}\in\rgxsem{\rgx} \label{eq:splice-R}
  \end{align}
  \begin{itemize}
  \item $\rgx = \rone$: Impossible, cannot derive a nonempty string.
  \item $\rgx = \sym_0$: Then $\sym_0 = \sym$ and $\s{\sym}_L =
  \s{\symy}_L = \s{\sym}_R = \s{\symy}_R = \none$ and we
  have the goal by assumption.

  \item $\rgx = \ralt{\rgx_L}{\rgx_R}$:
  $\sym$ appears in either $\rgx_L$ or $\rgx_R$---assume $\rgx_L$ without loss of generality. Then
  \begin{align}
    {\s{\sym}_L \sym \s{\symy}_R}\in\rgxsem{\rgx_L} \label{eq:splice-alt-L} \\
    {\s{\symy}_L \sym \s{\sym}_R}\in\rgxsem{\rgx_L} \label{eq:splice-alt-R}
  \end{align}
  By the inductive hypothesis, we have
  ${\s{\sym}_L \sym \s{\sym}_R}\in\rgxsem{\rgx_L}$
  and therefore
  ${\s{\sym}_L \sym \s{\sym}_R}\in\rgxsem{\rgx_L} \cup
  \rgxsem{\rgx_R} = \rgxsem{\rgx}$.

  \item $\rgx = \rseq{\rgx_L}{\rgx_R}$:
  $\sym$ appears in either $\rgx_L$ or $\rgx_R$---assume $\rgx_L$ without loss of generality.

  There exists partitions $\s{\symy}_R =
  \s{\symy}_{\ell}\s{\symy}_{r}$
  and $\s{\sym}_R = \s{\sym}_{\ell}\s{\sym}_{r}$ such that
  \begin{align}
    \s{\sym}_L\sym\s{\symy}_{\ell} \in \rgxsem{\rgx_L} \label{eq:splice-seq-y'L} \\
    \s{\symy}_{r} \in \rgxsem{\rgx_R} \label{eq:splice-seq-y'R} \\
    \s{\symy}_L\sym\s{\sym}_{\ell} \in \rgxsem{\rgx_L} \label{eq:splice-seq-x'L} \\
    \s{\sym}_{r} \in \rgxsem{\rgx_R} \label{eq:splice-seq-x'R}
  \end{align}
  Apply the inductive hypothesis
  to \eqref{eq:splice-seq-y'L}
  and \eqref{eq:splice-seq-x'L}
  to get
  $\s{\sym}_L\sym\s{\sym}_{\ell} \in \rgxsem{\rgx_L}$,
  which combined with \eqref{eq:splice-seq-x'R} gives
  $\s{\sym}_L\sym \s{\sym}_R = \s{\sym}_L\sym
  \s{\sym}_{\ell}\s{\sym}_{r}
  \in \rgxsem{\rseq{\rgx_L}{\rgx_R}}$.

  \item $\rgx = \rstar{\rgx_0}$:

  There exist $m,n\in\mathbb{N}$ such that
  \begin{align}
    {\s{\sym}_L \sym \s{\symy}_R}\in\rgxsem{\rgx_0^m} \label{eq:splice-star-L} \\
    {\s{\symy}_L \sym \s{\sym}_R}\in\rgxsem{\rgx_0^n} \label{eq:splice-star-R}
  \end{align}
  We know $m > 0$ and $n > 0$ because $\rgx^0_0 = \rone$ derives only the empty
  string, meaning
  \begin{align}
    {\s{\sym}_L \sym \s{\symy}_R}\in\rgxsem{\rseq{\rgx_0}{\rgx_0^{m-1}}} \label{eq:splice-star-L'} \\
    {\s{\symy}_L \sym \s{\sym}_R}\in\rgxsem{\rseq{\rgx_0}{\rgx_0^{n-1}}} \label{eq:splice-star-R'}
  \end{align}
  Case analysis on whether $\sym$ appears in $\rgx_0^{m-1}$ and in $\rgx_0^{n-1}$:
  \begin{itemize}
    \item Suppose $\sym$ appears in $\rgx_0^{m-1}$
    \ie there exists partition
    $\s{\sym}_L = \s{\sym}_\ell\s{\sym}_r$
    such that
    \begin{align}
      \s{\sym}_\ell \in \rgxsem{\rgx_0} \label{eq:splice-star-xL} \\
      \s{\sym}_r\sym\s{\symy}_R \in \rgxsem{\rgx_0^{m-1}} \label{eq:splice-star-xR}
    \end{align}
    Apply the inductive hypothesis to \eqref{eq:splice-star-xR} and
    \eqref{eq:splice-star-R} to get
    $\s{\sym}_r\sym\s{\sym}_R \in \rgxsem{\rstar{\rgx_0}}$,
    which combined with \eqref{eq:splice-star-xL} gives
    \begin{align}
      {\s{\sym}_L \sym \s{\sym}_R}
      = {\s{\sym}_\ell\cdot\s{\sym}_r\sym\s{\sym}_R}
      \in \rgxsem{\rseq{\rgx_0}{\rstar{\rgx_0}}}
      = \rgxsem{\rstar{\rgx_0}} \label{eq:splice-star-x}
    \end{align}

    \item Suppose $\sym$ appears in $\rgx_0^{n-1}$
    \ie there exists partition
    $\s{\symy}_L = \s{\symy}_\ell\s{\symy}_r$
    such that
    \begin{align}
      \s{\symy}_\ell \in \rgxsem{\rgx_0} \label{eq:splice-star-yL} \\
      \s{\symy}_r\sym\s{\sym}_R \in \rgxsem{\rgx_0^{n-1}} \label{eq:splice-star-yR}
    \end{align}
    Apply the inductive hypothesis to \eqref{eq:splice-star-L} and
    \eqref{eq:splice-star-yR} to get
    \begin{align}
      {\s{\sym}_L \sym \s{\sym}_R}
      \in \rgxsem{\rseq{\rgx_0^m}{\rgx_0^{n-1}}}
      \subseteq \rgxsem{\rstar{\rgx_0}}
    \end{align}

    \item Suppose $\sym$ appears in neither $\rgx_0^{m-1}$ nor $\rgx_0^{n-1}$.
    Then there exist partitions $\s{\symy}_R = \s{\symy}_\ell\s{\symy}_r$
    and $\s{\sym}_R = \s{\sym}_\ell\s{\sym}_r$ such that
    \begin{align}
      \s{\sym}_L\sym\s{\symy}_r \in \rgxsem{\rgx_0} \\
      \s{\symy}_L\sym\s{\sym}_r \in \rgxsem{\rgx_0} \\
      \s{\symy}_r \in \rgxsem{\rgx_0^{m-1}} \\
      \s{\symy}_r \in \rgxsem{\rgx_0^{n-1}}
    \end{align}
    Apply the same argument as in the case $\rgx = \rseq{\rgx_L}{\rgx_R}$ to
    reach the goal.
  \end{itemize}
  \end{itemize}
\end{proof}



\begin{proof}[Proof of \Lemmaref{lemma:valid-prefix}]
  Assume the premises
  \begin{align}
    \cft \lessdot_{\dsty{\o{\cfsr}_0}} \labl_0 \label{eq:valid-lt0}\\
    \iseq{\labl_{i-1} \doteq_{\dsty{\o{\cfsr}_{i}}} \labl_{i}}{0< i\leq k} \label{eq:valid-eqs} \\
    \labl_k \adjop \o{\cfsr}_{k+1} \label{eq:valid-adj}
  \end{align}
  Apply Lemma \ref{lemma:lt-bot-witness} to \eqref{eq:valid-lt0}
  to get nonterminal $\cfs_\cft = \bddsort{\precm}{\sort_\cft}{\bot}$
  such that
  \begin{align}
  \adj{\cft}{\cfs_\cft} \label{eq:cft-adj} \\
  \cfs_\cft \produce^* \o{\cfsr}_0\labl_0\dots \label{eq:valid-derive-star}
  \end{align}
  Apply Lemma \ref{lemma:derive-leftmost-yield} to \eqref{eq:valid-derive-star}
  to get nonterminal $\cfs_0$ such that
  \begin{align}
    \lmsymStar{\cfs_\cft}{\cfs_0} \label{eq:valid-lmsym0} \\
    \cfs_0 \produce \o{\cfsr}_0\labl_0\dots \label{eq:valid-yield0}
  \end{align}
  Invert \eqref{eq:valid-eqs} to get nonterminals $\iseq{\cfs_i}{0< i\leq k}$ such that
  \begin{align}
    \iseq{\cfs_{i} \produce \dots\labl_{i-1}\o{\cfsr}_{i}\labl_{i}\dots}{0< i\leq k} \label{eq:valid-yieldi}
  \end{align}
  Finally, unabbreviate \eqref{eq:valid-adj} to get nonterminal $\cfs_{k+1}$
  such that
  \begin{align}
    \cfs_{k+1} \produce \dots\labl_k\o{\cfsr}_{k+1}\dots \label{eq:valid-produce-k+1}
  \end{align}

  First, we will show that the productions in
  \eqref{eq:valid-yield0}-\eqref{eq:valid-produce-k+1}
  are backed by a shared derivation from some regex in $\gram$,
  given \Assumptionref{assum:unique-tiles}.
  Letting
    $\iseq{\cfs_i = \bddsort{\precp_i}{\sort_i}{\precq_i}}{0\leq i\leq k}$,
  invert \eqref{eq:valid-yield0}-\eqref{eq:valid-produce-k+1} by
  rule \texttt{Produce-Subsume} to get
  \begin{align}
    \bddsort{p_0}{\sort_0}{q_0} \reduce \o{\cfsr}_0\labl_0\dots \label{eq:valid-rule0'} \\
    \iseq{\bddsort{p_i}{\sort_i}{q_i} \reduce \dots\labl_{i-1}\o{\cfsr}_{i}\labl_{i}\dots}{0< i\leq k} \label{eq:valid-rulei'} \\
    \bddsort{p_{k+1}}{\sort_{k+1}}{q_{k+1}} \reduce \dots\labl_k\o{\cfsr}_{k+1}\dots \label{eq:valid-rule-k+1'}
  \end{align}
  Invert again to get precedence levels $\iseq{\precn_i}{0\leq
  i\leq k+1}$ and optional sorts $\iseq{\o{\sortr}_i}{0\leq i\leq k+1x}$ such that
  \begin{align}
    \iseq{\o{\cfsr}_i \consistent \o{\sortr}_i}{0\leq i\leq k+1} \label{eq:valid-consistent} \\
    {\o{\sortr}_0\labl_0\dots}\in\rgxsem{\gram(\sort_0,\precn_0)} \label{eq:valid-derive0'} \\
    \iseq{{\dots\labl_{i-1}\o{\sortr}_{i}\labl_{i}\dots}\in\rgxsem{\gram(\sort_i,\precn_i)}}{0< i\leq k} \label{eq:valid-derivei'} \\
    \dots\labl_k{\o{\sortr}_{k+1}\dots}\in\rgxsem{\gram(\sort_{k+1},\precn_{k+1})} \label{eq:valid-derive-k+1'}
  \end{align}
  \Assumptionref{assum:unique-tiles} applied to
  \eqref{eq:valid-derive0'}-\eqref{eq:valid-derive-k+1'}
  gives us sort $\sort$ and precedence $\precn$ such that
  \begin{align}
    \iseq{\sort = \sort_i}{0\leq i\leq k+1} \label{eq:valid-same-sort} \\
    \iseq{\precn = \precn_i}{0\leq i\leq k+1} \label{eq:valid-same-precn}
  \end{align}
  Given \eqref{eq:valid-same-sort} and \eqref{eq:valid-same-precn},
  we can now apply \Lemmaref{lemma:splice} to \eqref{eq:valid-derive0'}-\eqref{eq:valid-derive-k+1'}
  to get
  \begin{align}
    {\iseq{\o{\sortr}_i\labl_i}{0\leq i\leq k}\,\o{\sortr}_{k+1}\,\s{\sym}}\in\rgxsem{\gram(\sort,\precn)} \label{eq:valid-symbar-derive}
  \end{align}
  for some symbol sequence $\s{\sym}$, as desired.
  By \Assumptionref{assum:operator-form}, we can rewrite $\s{\sym}$
  with optional sorts $\iseq{\o{\sortr}_i}{k+1 < i \leq \ell + 1}$
  and tiles $\iseq{\labl_i}{k < i \leq \ell}$
  such that
  \begin{align}
    {\o{\sortr}_0\iseq{\labl_i\o{\sortr}_{i+1}}{0\leq i\leq \ell}}\in\rgxsem{\gram(\sort,\precn)} \label{eq:valid-derive'}
  \end{align}

  Now it remains to determine bounds $\precp,\precq$,
  tiles $\iseq{\labl_i}{k < i \leq \ell}$,
  and slots $\iseq{\o{\cfsr}_i}{k+1 < i \leq \ell+1}$
  such that
  \begin{align}
  \lmsymStar{\cfs_\cft}{\psq} \label{eq:lmsym-star-goal}\\
  {\psq}\reduce{\o{\cfsr}_0
  \iseq{\labl_i~\o{\cfsr}_{i+1}}{0\leq i\leq \ell}} \label{eq:derive-goal}
  \end{align}
  We will show \eqref{eq:derive-goal} using one of the elaboration
  rules in \autoref{fig:compile-precedence}.
  Determining which elaboration rule to use requires case analysis on whether
  $\o{\sortr}_0 = \some{\sort}$ and $\o{\sortr}_{\ell+1} = \some{\sort}$.
  Without loss of generality,
  assume $\o{\sortr}_0 = \some{\sort}$
  and $\o{\sortr}_{\ell+1} \neq \some{\sort}$.

  Proceed by case analysis on whether $\sort_\cft = \sort$:
  \begin{itemize}

  \item
  Suppose $\sort_\cft \neq \sort$.
  Rewrite \eqref{eq:valid-lmsym0} with \eqref{eq:valid-same-sort}
  to get $\cfs_\cft \lmsymop^* \bddsort{\medsquare}{\sort}{\medsquare}$,
  then apply Lemma \ref{lemma:unbounded-sort-transition} to get
  \begin{align}
  \cfs_\cft \lmsymop^* \bddsort{\bot}{\sort}{\bot}
  \end{align}
  Pick $\precp = \bot$ and $\precq = \bot$ to reach \eqref{eq:lmsym-star-goal}.

  It remains to construct slots $\iseq{\o{\cfsr}_i}{k+1 < i \leq \ell+1}$
  satisfying \eqref{eq:derive-goal}.
  Apply \texttt{PElab-Postfix} to \eqref{eq:valid-derive'}
  to get
  \begin{align}
  \psq \reduce \o\cfsr_0 \labl_0 \iseq{\promotesym{\o{\sortr}_{i}} \labl_i}{0< i \leq \ell} \label{eq:pe-postfix}
  \end{align}
  where we write
  \begin{align}
    \promotesym{\o\sortr_i} \triangleq \begin{cases}
      \o{\cfsr}_i & \text{if } 0 < i \leq k \\
      \some{\botsort{\sortr_i}} & \text{if } k < i \text{ and } \o{\sortr}_i = \some{\sortr_i} \\
      \none & \text{if } \o{\sortr}_i = \none
    \end{cases} \label{eq:promotesym}
  \end{align}
  Picking $\iseq{\o{\cfsr}_i = \promotesym{\o{\sortr}_i}}{k < i \leq \ell + 1}$
  gives us our goal \eqref{eq:derive-goal}.


  \item
  Suppose $\sort_\cft = \sort$.
  Pick $\precp = \precm$ and $\precq = \bot$, \ie
  $\cfs_\cft = \bddsort{\precm}{\sort}{\bot} = \psq$, which satisfies \eqref{eq:lmsym-star-goal}.

  It remains to construct slots $\iseq{\o{\cfsr}_i}{k < i \leq \ell+1}$ satisfying \eqref{eq:derive-goal}.
  We have $\bddsort{\precm}{\sort}{\bot} = \cfs_\cft \lmsymop^* \cfsr_0 =
  \bddsort{\precp_0}{\sort}{\precq_0}$.
  By Lemma \ref{lemma:dispute-inherit},
  we have $\precm \preccurlyeq_\sort \precp_0$.
  We have $\o{\sortr}_0 = \some{\sort}$
  which implies $\o{\cfsr}_0 = \some{\bddsort{\precn_L}{\sort}{\precn_R}}$
  for some $\precn_L,\precn_R$.
  We have $\bddsort{\precp_0}{\sort}{\precq_0} \produce
  \bddsort{\precn_L}{\sort}{\precn_R}\labl_0\dots$ by \eqref{eq:valid-yield0}%
  ---case analysis on the underlying reduction
  (either \texttt{PElab-Postfix} or \texttt{PElab-Infix})
  tells us that $\precn_R = \precn$.
  By Lemma \ref{lemma:dispute-inherit}, we have
  $\precp_0 \preccurlyeq_\sort \precn_L \prec_\sort \precn_R$.
  Inverting \eqref{eq:valid-yield0} gives us $\precn_R = \precn$.
  Therefore $\precm \preccurlyeq_\sort \precp_0 \preccurlyeq_\sort \precn_L \prec_\sort
  \precn_R = \precn$.

  We have $\precm \prec_\sort \precn$
  and that the rightmost symbol of the derived string
  in \eqref{eq:valid-derive'} is $\labl_\ell$,
  so apply \texttt{PElab-Postfix} to \eqref{eq:valid-derive'} to get
  \begin{align}
  \psq \reduce \bddsort{\precp}{\sort}{\precn} \labl_0 \iseq{\promotesym{\o{\sortr}_{i}} \labl_i}{0< i \leq \ell} \label{eq:pe-postfix'}
  \end{align}
  where $\promotesym{\o\sortr_i}$ is defined as in \eqref{eq:promotesym}.
  The same reasoning applied in the case $\sort_\cft \neq \sort$ from
  \eqref{eq:pe-postfix} onward gives us our goal \eqref{eq:derive-goal}.


  \end{itemize}

\end{proof}

\subsection{Proof of Lemma \ref{lemma:push-total-0}} \label{sec:push-total-proof}

\begin{lemma}[Fill Produces Well-Formed Terms] \label{lemma:fill-wf}
  If $\pfill{\o{\termr}}{\iseq{\o{\cfs}_i}{1\leq i\leq k}}=\iseq{\o{\term}_i}{1\leq i \leq
  k}$
  then $\iseq{\o{\cfs}_i\Produce\o{\term}_i}{1\leq i \leq k}$.
\end{lemma}

\begin{figure}
  \judgbox{\nat{\term}}{Term $\term$ is natural}
  \begin{mathpar}
    \inferrule[Natural]{
      0 \pleq_\sort \precp \\
      \psq \Reduce \term \\
      \precq \pgeq_\sort 0
    }{
      \nat{\term}
    }
  \end{mathpar}

  \judgbox{\wf{\pushl{\stack}{\s{\termr}}{\medsquare}}}{Stack configuration $\pushl{\stack}{\dsty{\s{\termr}}}{\medsquare}$ is well-formed}
  \begin{mathpar}
    \inferrule[WFConfig$\;\glp\textcolor{gray}{k\geq 0}\glr$]{
      \wf{\stack} \\
      \adj{\head{\stack}}{\o\cfs} \\
      \pfill{\iseq{\termr_i}{0\leq i\leq k}}{\o\cfs} = \s{\o\term} \\
      \iseq{\nat{\termr_i}}{0\leq i\leq k}
    }{
      \wf{\pushl{\stack}{\dsty{\iseq{\termr_i}{0\leq i\leq k}}}{\medsquare}}
    }
  \end{mathpar}


\caption{Push invariant}\label{fig:push-invariant}
\end{figure}





\begin{proof}[Proof of \Lemmaref{lemma:push-total-0}]
  Corollary of \Lemmaref{lemma:push-total},
  generalized to support proof by induction
  over recursive \texttt{Reduce} steps.
\end{proof}

\begin{lemma}[Generalized Push Totality]
\label{lemma:push-total}
For every stack $\stack$,
reduction sequence $\s{\termr}$,
and token $\cft$
such that $\wf{\pushl{\stack}{\s{\termr}}{\medsquare}}$,
there exists $\stack'$
such that $\push{\stack}{\s{\termr}}{\cft}{\stack'}$
and $\wf{\stack'}$.
\end{lemma}

\begin{proof}
Proceed by strong induction on the stack $\stack$.

\begin{itemize}
\item Suppose $\stack$ is empty, \ie $\stack = \dlroot$.

First, we will show that we can apply \texttt{Shift},
\ie there exist nonterminals $\iseq{\cfsr_i}{0\leq i\leq k}$,
tokens $\iseq{\cft_i}{0\leq i\leq k}$,
and cells $\iseq{\o\term_i}{0\leq i\leq k}$
such that
\begin{align}
  \dlroot\iseq{\cmprleq_{\dsty{\o{\cfsr}_i}}\cft_i}{0\leq i\leq k} = \cft \label{eq:shift-premise-chain} \\
  \pfill{\s{\termr}}{\iseq{\o{\cfsr}_i}{0\leq i\leq k}} = {\iseq{\o{\term}_i}{0\leq i\leq k}} \label{eq:shift-premise-fill}
\end{align}
Proceed by case analysis on $\cft$:
\begin{itemize}
\item First suppose $\cft = \drroot$.
Goal \eqref{eq:shift-premise-chain} is satisfied by the fact that $\dlroot \doteq_{\dsty{\some{\botsort{\startSym}}}} \drroot$.
It remains to show there exist cells $\s{\o\term}$ such that
$\pfill{\s{\termr}}{\some{\botsort{\startSym}}} = \s{\o\term}$.
Inverting the premise $\wf{\pushl{\dlroot}{\s{\termr}}{\medsquare}}$,
there exists slot $\o\cfsr$ and cells $\s{\o\term}$ such that
\begin{align}
\dlroot\adjop\o{\cfsr} \label{eq:adj-dlroot} \\
\pfill{\s{\termr}}{\o\cfsr} = \s{\o\term} \label{eq:reductions-fill-startsym}
\end{align}
Given \eqref{eq:adj-dlroot}, we further have
\begin{align}
\o\cfsr = \some{\botsort{\startSym}} \label{eq:cfsr-startsym}
\end{align}
with which rewriting \eqref{eq:reductions-fill-startsym} gives us \eqref{eq:shift-premise-fill}.

\item Otherwise, suppose $\cft = \labl$ for some tile $\labl$.
Let $\s{\termr} = \iseq{\termr_i}{0< i\leq k}$ for some $k\geq 0$.

Lemma \ref{lemma:root-grout-prec} gives us
tiles $\iseq{\labl_i}{0\leq i \leq \ell}$
and slots $\iseq{\o{\cfs}_i}{0\leq i \leq \ell}$
such that
\begin{align}
\binhole^{\startSym} \lessdot_{\dsty{\o{\cfs}_0}} \labl_0
  \iseq{\doteq_{\dsty{\o{\cfs}_i}}\labl_i}{1\leq i\leq \ell} = \labl
  \label{eq:binhole-tile-chain}
\end{align}
We can further show that
\begin{align}
\dlroot
  \lessdot_{\dsty{\none}} \prehole^{\startSym}
  \iseq{
    \doteq_{\dsty{\some{\zerosort{\startSym}}}}
    \binhole^{\startSym}
  }{0\leq i\leq k} \label{eq:dlroot-binhole-chain}
\end{align}
given that $\dlroot \adjop \botsort{\startSym}$
and $\botsort{\startSym} \reduce \prehole^{\startSym}\iseq{
  \zerosort{\startSym}
  \binhole^{\startSym}
}{0\leq i\leq k}$ by rule \texttt{GInj-Prefix}.
Composing \eqref{eq:dlroot-binhole-chain} with \eqref{eq:binhole-tile-chain}
gives us \eqref{eq:shift-premise-chain}.

To show \eqref{eq:shift-premise-fill},
it suffices to show
there exist cells $\iseq{\o{\term}_i}{0 < i \leq k}$ such that
\begin{align}
\iseq{\pfill{\termr_i}{\some{\zerosort{\startSym}}} = \o{\term}_i}{0< i\leq k}
\end{align}
which can be shown using reduction naturality.
\end{itemize}

Given \eqref{eq:shift-premise-chain} and \eqref{eq:shift-premise-fill},
rule \texttt{Shift} gives us
\begin{align}
\push{\dlroot}{\dsty{\s{\termr}}}{\cft}{\dlroot\iseq{\cmprleq_{\dsty{\o{\term}_i}}\cft_i}{0\leq i\leq k}}
\end{align}
It remains to show
\begin{align}
  \wf{\dlroot\iseq{\cmprleq_{\dsty{\o{\term}_i}}\cft_i}{0\leq i\leq k}} \label{eq:shifted-stack-wf}
\end{align}
\Lemmaref{lemma:fill-wf} applied to \eqref{eq:shift-premise-fill} gives us
\begin{align}
  \iseq{\o{\cfs}_i \Produce \o{\term}_i}{0\leq i\leq k} \label{eq:cells-wf}
\end{align}
Using \eqref{eq:shift-premise-chain} and \eqref{eq:cells-wf}
and $\wf{\dlroot}$ (given by \texttt{WFStack-Empty}),
apply rule \texttt{WFStack-Cons} $k+1$ times in succession to get \eqref{eq:shifted-stack-wf}.

\item Otherwise, assume $\stack$ is nonempty
and the inductive hypothesis that,
for any substack $\stack_0$ of $\stack$
and reductions $\s{\term}_0$
such that $\wf{\pushl{\stack_0}{\s{\termr}_0}{\medsquare}}$
there exists $\stack'$
such that $\push{\stack_0}{\s{\term}_0}{\cft}{\stack'}$
and $\wf{\stack'}$.

Inverting the premise $\wf{\pushl{\stack}{\s{\termr}}{\medsquare}}$,
there exists slot $\o\cfs$ and cells $\s{\o\term}$ such that
\begin{align}
\wf{\stack} \label{eq:wf-stack} \\
\head{\stack}\adjop\o{\cfsr} \label{eq:adj-stack} \\
\pfill{\s{\termr}}{\o\cfsr} = \s{\o\term} \label{eq:reductions-fill-neighbor}
\end{align}
Since $\stack$ is nonempty,
\Lemmaref{lemma:droot-eq} implies
there exist tokens $\iseq{\cft_i}{0 \leq i \leq k}$
and cells $\iseq{\o{\termr}_i}{0 \leq i \leq k}$
and stack $\stack_0$ such that
\[
\stack ~=~ \stack_0\lessdot_{\dsty{\o{\termr}_0}}\tnode{\cft_0}
~\iseq{\doteq_{\dsty{\o{\termr}_i}}\tnode{\cft_i}}{1\leq i\leq k}
\]
for some $k\geq 0$.
Starting with \eqref{eq:wf-stack}, invert rule \texttt{WFStack-Cons} $k+1$ times to get
\begin{align}
\wf{\stack_0} \label{eq:wf-stack-0} \\
\head{\stack_0}\lessdot_{\dsty{\o{\cfsr}_0}}\cft_0\iseq{\doteq_{\dsty{\o{\cfsr}_i}}\cft_i}{1\leq i\leq k} \label{eq:terr-prec} \\
\iseq{\o{\cfsr}_i \Produce \o{\termr}_i}{0\leq i\leq k} \label{eq:subterms-wf}
\end{align}

\Lemmaref{lemma:homogeneity} applied to \eqref{eq:terr-prec} implies
either
\begin{align}
\iseq{\cft_i = \labl_i}{0\leq i\leq k} \label{eq:terr-tile}
\end{align}
for some tiles $\iseq{\labl_i}{0\leq i\leq k}$
or
\begin{align}
\iseq{\cft_i = \groot_i^{\sort}}{0\leq i\leq k} \label{eq:terr-grout}
\end{align}
for some grout $\iseq{\groot_i}{0\leq i\leq k}$ and sort $\sort$.

\begin{itemize}
\item Suppose \eqref{eq:terr-tile} holds.
We will show that it is possible to apply rule \texttt{Reduce}.

By \Lemmaref{lemma:valid-prefix} applied to \eqref{eq:terr-prec} and \eqref{eq:adj-stack},
there exist nonterminals $\cfs_0,\cfs$,
slots $\iseq{\o{\cfsr}_i}{k < i \leq \ell+1}$,
and tiles $\iseq{\labl_i}{k < i \leq \ell}$
for some $\ell\geq k$
such that
\begin{align}
  \o{\cfsr}_{k+1} = \o{\cfsr} \\
  \adj{\head{\stack_0}}{\lmsymStar{\cfs_0}{\cfs}} \label{eq:stack-bound-tile} \\
  \cfs \reduce \o{\cfsr}_0 \iseq{\labl_i~\o{\cfsr}_{i+1}}{0\leq i\leq \ell} \label{eq:reduce-produce-premise}
\end{align}
To conclude with rule \texttt{Reduce},
given \eqref{eq:reduce-produce-premise} and \eqref{eq:subterms-wf},
it remains to show
there exist cells $\iseq{\o{\termr}_i}{k < i \leq \ell+1}$
and stack $\stack'$ such that
\begin{align}
  \pfill{\s{\termr}}{\iseq{\o{\cfsr}_i}{k < i \leq \ell+1}} = \iseq{\o{\termr}_i}{k < i \leq \ell+1} \label{eq:reduce-fill-premise} \\
  \push{\stack_0}{\displaystyle{
    \snode{\o{\termr}_0\iseq{\labl_i\,\o{\termr}_{i+1}}{0\leq i\leq
    \ell}}
  }}{\cft}{\stack'} \label{eq:reduce-subpush}
\end{align}
\begin{itemize}
\item To show \eqref{eq:reduce-fill-premise},
it suffices by rule \texttt{Fill-Partition} to show
there exists a partition $\s{\termr} = \iseq{\s{\termr}_i}{k< i\leq \ell+1}$
and cells $\iseq{\o{\termr}_i}{k < i \leq \ell+1}$
such that
\begin{align}
  \iseq{\pfill{\s{\termr}_i}{\o{\cfsr}_i} = \o{\termr}_i}{k < i \leq \ell+1} \label{eq:reduce-fill-partition}
\end{align}
Pick the partition $\s{\termr} = \s{\termr}\,\iseq{\snil}{k+1 <i \leq \ell+1}$
and use \eqref{eq:reductions-fill-neighbor} for index $k+1$
and either rule \texttt{Fill-None} or \texttt{Fill-Default} for the remaining
indices $k+1 < i \leq \ell+1$.

\item To show \eqref{eq:reduce-subpush},
it suffices by the inductive hypothesis to show
\begin{align}
  \wf{\pushl{\stack_0}{\displaystyle{
    \snode{\o{\termr}_0\iseq{\labl_i\,\o{\termr}_{i+1}}{0\leq i\leq
    \ell}}
  }}{\medsquare}}
\end{align}
Let $\term = \snode{\o{\termr}_0\iseq{\labl_i\,\o{\termr}_{i+1}}{0\leq i\leq
\ell}}$.
Given \eqref{eq:wf-stack-0} and \eqref{eq:stack-bound-tile},
it suffices to show there exists term $\term_0$ such that
\begin{align}
  \pfill{\term}{\some{\cfs_0}} = \some{\term_0} \label{eq:reduce-fill-0}
\end{align}
It suffices by rule \texttt{Fill-Postfix} to show
\begin{align}
  \cfs_0\Produce\snode{\term\,\poshole^\sort}
\end{align}
It suffices by rule \texttt{Produce-Term}
to show
\begin{align}
  \cfs_0\produce\cfs\,\poshole^\sort
\end{align}
It suffices by rule \texttt{Produce-Subsume} to show
\begin{align}
  \cfs_0\reduce\cfs\,\poshole^\sort
\end{align}
Let $\cfs_0 = \bddsort{\precp_0}{\sort_0}{\precq_0}$
and $\cfs = \psq$.
It suffices by rule \texttt{GInj-Postfix} to show
\begin{align}
  \zerosort{\sort_0} \lmsymop^* \cfs
\end{align}
By Lemma \ref{lemma:left-bounded-produce} applied to
\eqref{eq:stack-bound-tile},
either $\cfs \consistent \sort_0$ or $\cfs_0 \produce \cfs\,\poshole^{\sort_0}$.

\begin{itemize}
\item Case $\cfs \consistent \sort_0$, \ie $\sort = \sort_0$:
By reduction naturality applied to
\eqref{eq:reduce-produce-premise},
we know $\precz \pleq_{\sort} \precp$ and $\precq \pgeq_{\sort} \precz$,
so we are done by \texttt{Produce-Tighten}.

\item Case $\cfs_0 \produce \cfs\,\poshole^{\sort_0}$:
Invert rule \texttt{Produce-Subsume} to get
$\cfs_0 \reduce \cfs\,\poshole^{\sort_0}$.
Invert rule \texttt{GInj-Postfix} in turn to conclude.
\end{itemize}
\end{itemize}

\item Suppose \eqref{eq:terr-grout} holds.
It suffices by rule \texttt{Degrout} to show there exists $\stack'$ such that
\begin{align}
  \push{\stack_0}{\dsty{\iseq{\o{\termr}_i}{0\leq i\leq k}\,\s{\termr}}}{\cft}{\stack'}
\end{align}
It suffices by our inductive hypothesis to show
\begin{align}
  \wf{\pushl{\stack_0}{\dsty{\iseq{\o{\termr}_i}{0\leq i\leq k}\,\s{\termr}}}{\medsquare}}
\end{align}
That is, by \rulename{WFConfig}, to show
\begin{align}
  \wf{\stack_0}\label{eq:theq-c1}\\
  \head{\stack_0}\adjop\cfs_0?\label{eq:theq-c2}\\
  \pfill{\iseq{\o{\termr}_i}{0\leq i\leq k}\,\s{\termr}}{\cfs?_0} = \s{\o\term_0}\label{eq:theq-c3}\\
  \nat{\iseq{\o{\termr}_i}{0\leq i\leq k}\,\s{\termr}}\label{eq:theq-c4}
\end{align}
For some $\cfs?_0$ and $\s{\o\term_0}$. We already have \eqref{eq:theq-c1} from \eqref{eq:wf-stack-0}. We obtain \eqref{eq:theq-c2} for $\cfs?_0 = \some{\cfs_0}$ by rule inversion of \rulename{Prec-LT} on $\head{\stack_0}\lessdot_{\o{\cfsr}_0}\cft_0$ in \eqref{eq:terr-prec}. For \eqref{eq:theq-c3}, we apply either \rulename{Fill-Default} or \rulename{Fill-Operand}, depending on whether $\iseq{\o{\termr}_i}{0\leq i\leq k}\,\s{\termr}$ is empty. The former case is trivial. For the latter case, we use \rulename{Produce-Tighten}. This means we have to show that $\cfs_0$ produces a sequence of alternating grout and nonterminals, such that the nonterminals produce the nonempty entries in $\iseq{\o{\termr}_i}{0\leq i\leq k}\,\s{\termr}$. We already have the nonterminals for $\iseq{\o{\termr}_i}{0\leq i\leq k}$; they are $\iseq{\o{\cfsr}_i}{0\leq i\leq k}$ by \eqref{eq:subterms-wf}. The sorts for $\s{\termr}$ are the sorts they synthesize, with $0,0$ bounds, they analyzed against these because they are natural.

Now we need to show that $\cfs_0$ actually does produce this form. We obtain this by widening to $\top$-$\top$ bounds, applying subsumption, then applying \rulename{GInj-Operand}. Now we have to show that each of the nonterminals mentioned above is accessible from the sort of $\cfs_0$, and has non-$\bot$ bounds. We obtain non-$\bot$ bounds for the first set, $\iseq{\o{\cfsr}_i}{0\leq i\leq k}$, by rule inversion on the fact that they appear next to grout in \eqref{eq:terr-prec}. The second set was constructed to have 0 bounds, which are not $\bot$. We have accessibility from the sort of $\cfs_0$ for $\iseq{\o{\cfsr}_i}{0\leq i\leq k}$ by \textit{splicing} \eqref{eq:terr-prec}. For accessibility of the root sorts of $\s{\termr}$, we observe that each one is accessible from $\o\cfsr$ by \eqref{eq:reductions-fill-neighbor}, which is in turn accessible from the sort of $\cfs_0$ by the aforementioned splice.

This completes the slot filling obligation. The final obligation is to show naturality, \eqref{eq:theq-c4}. We already have $\nat{\s{\termr}}$ by assuming the original stack configuration is well-formed. For $\iseq{\o{\cfsr}_i}{0\leq i\leq k}$, we rely on the fact that the bounds for each $\iseq{\o{\cfsr}_{i}}{0\leq i\leq k}$ is non-$\bot$, since they appear next to grout.
\end{itemize}
\end{itemize}
\end{proof}

\section{Additional Data for \autoref{sec:user-study}} \label{appendix:tall-study}

\subsection{More Scenarios}

\begin{figure}[h]
    \centering
    \includegraphics[width=0.75\linewidth]{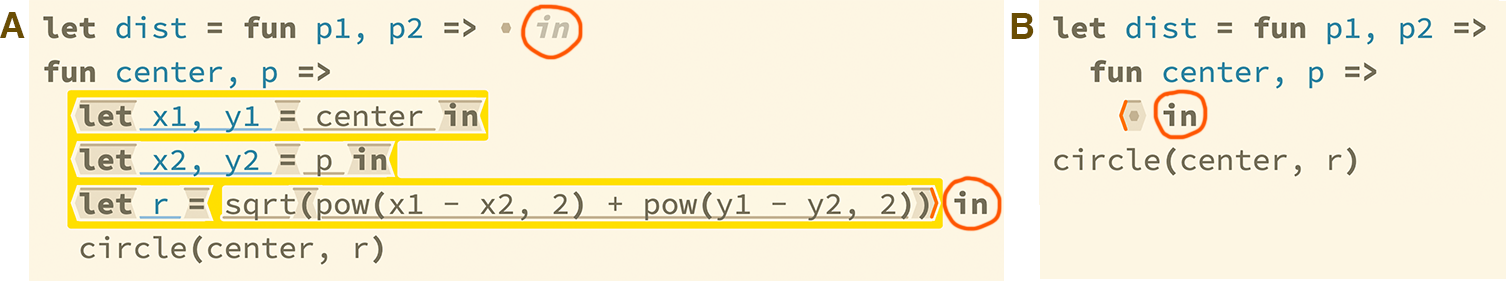}
    \caption{A participant has stubbed the header for a helper function, and is about to cut some relevant code to paste in the helper. However, they left the \code{in} delimiter belonging to the helper stub as a ghost, and incidentally omitted an \code{in} from their selection. On cut, that latter \code{in} becomes an orphan, which is then matched to the ghost \code{in}. This has the effect of shunting the existing function literal into the body of the helper.}
    \label{fig:scenario-reparenting}
\end{figure}

\paragraph*{Spooky action-at-a-distance due to unattended ghosts}

Ghost cleanup logic can trigger non-local effects which can be hard to anticipate when the ghost is not directly involved in the current edit. \autoref{fig:scenario-reparenting} illustrates an example from Task 4, where the participant was confused that their cut action seemingly sucked the function literal into the body of the definition, taking them further from their goal.

\begin{figure}[h]
    \centering
    \includegraphics[width=0.7\linewidth]{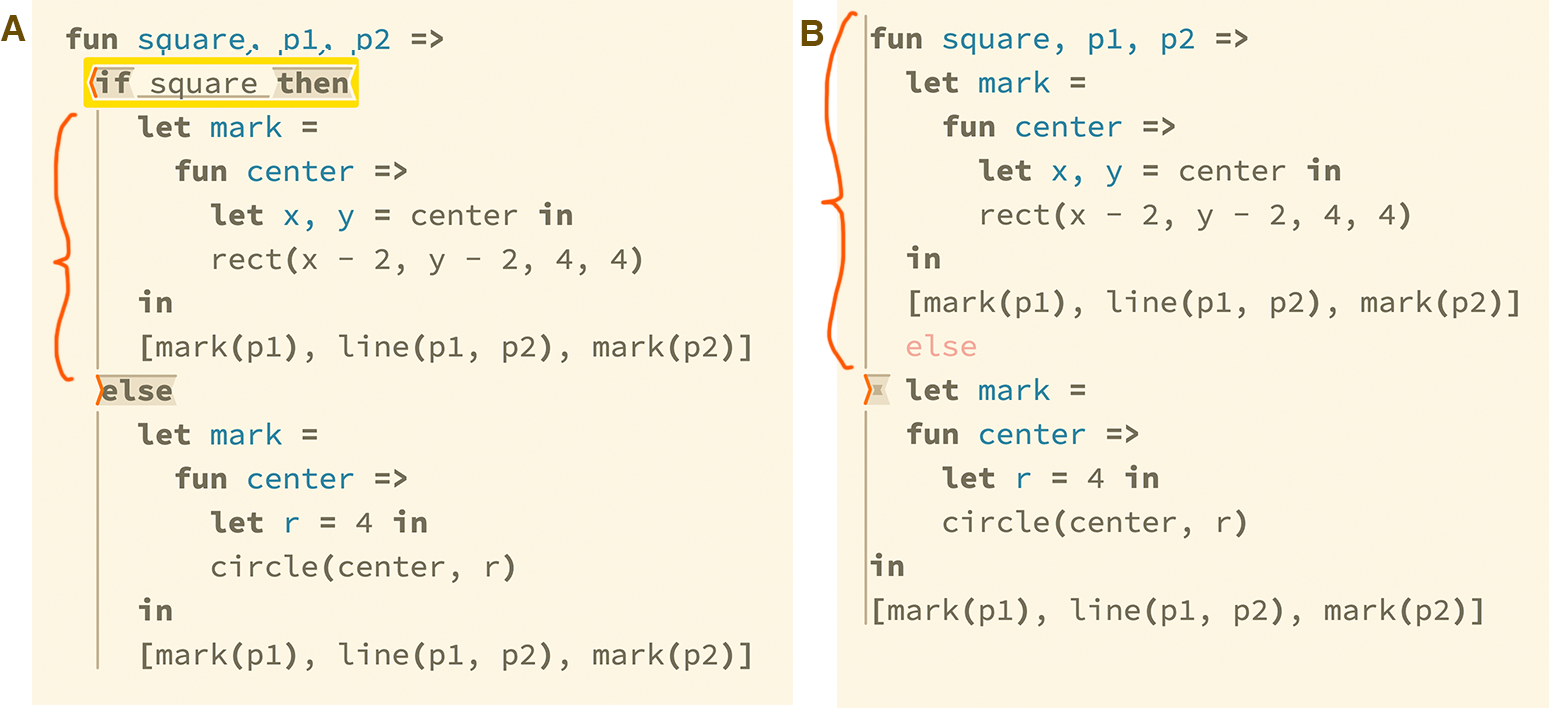}
    \caption{During Task 6, participants had to push an \code{if} expression deeper into a function, which they typically approached by cutting the segment highlighted in yellow. This cut leaves behind an unmolded red 'else' delimiter and an infix obligation. Since infix obligations are assigned the loosest precedence, the function literal taking square as an argument is now entirely on the left side of the grout, and the expression on the last line of the program is no longer inside that function literal, resulting in a subtle but substantial change to the program structure.}
    \label{fig:scenario-grout-delete}
\end{figure}

\paragraph*{Infix obligations and unexpected shifts in structure}

While participants seemed to accept the notion of infix obligations as missing infix operators, their sudden appearance during the deletion of compound syntactic forms like definitions and conditionals lead to subtle and surprising situations. One of these is illustrated in \autoref{fig:scenario-grout-delete}. Many participants did not notice the infix obligation in such situations, pushing on with their planned edits, but for those that did this was a source of confusion. Participant \textbf{P6} notes: "What I expected to be the case is that this piece of grout is a binary operator on this rect and this list in the body of this let. But that's not even true."

\subsection{More Participant Reactions}

\begin{displayquote}
    \textbf{P2:} "It's feeling pretty good. [...] I don't know if I'm just like not thinking of a lot of stuff right now, or if that was just so smooth that it doesn't give me a lot of thoughts"
\end{displayquote}

\begin{displayquote}
   \textbf{P3:} "It takes probably takes more time to develop an intuition behind when to use tylr powers [...] in particular, when use space and start to type something"
\end{displayquote}

\begin{displayquote}
    \textbf{P9:} "It was always clear visually if the editor was on the same page as me."
\end{displayquote}

\begin{displayquote}
    \textbf{P2:} "It's feeling pretty good. [...] I don't know if I'm just like not thinking of a lot of stuff right now, or if that was just so smooth that it doesn't give me a lot of thoughts"
\end{displayquote}

\begin{displayquote}
    P4: "The grout and ghosts, when they worked, felt pretty seamless"
\end{displayquote}

\begin{displayquote}
    \textbf{P8:} "I always would like placeholder completions until I have a complete expression! They allow me to scaffold until I have a complete expression of code."
\end{displayquote}

\begin{displayquote}
    \textbf{P2:} "I don't have to remember the syntax of the language as much, I just have to remember how to write the first token in a term."
\end{displayquote}


\begin{displayquote}
    \textbf{P5:} "The editor sometimes added a lot of holes [...] while I was in the middle of editing an expression, which I instinctively tried to delete."
\end{displayquote}

\begin{displayquote}
   \textbf{P7:} "I was actually surprised. I thought it was going to break. But then it worked. So, I'm pleasantly surprised, I'll say for that one."
\end{displayquote}

\begin{displayquote}
   \textbf{P4:} "So the the fact that these ghosts and grout are in my way and I can't get rid of them is super annoying."
\end{displayquote}

\begin{displayquote}
   \textbf{P3:} "It takes probably takes more time to develop an intuition behind when to use tylr powers [...] in particular, when use space and start to type something"
\end{displayquote}

\section{\tylr Performance} \label{appendix:perf}

We performed some simple benchmarks to test \tylr's parsing and insertion edit performance.

\subsection{Left-to-right parsing performance setup}

For the purposes of judging parsing performance with respect to program length,
we approximated a naturalistic code example by assembling a base program of 100
lines consisting of all examples from our user study (\autoref{sec:user-study}),
expressed as definitions with a trailing hole. For each trial, we concatenated
as many copies of that program as necessary to hit the line total, and then
truncated the result to the desired line length. This creates a program which is
syntactically correct and complete modulo trailing obligations. This program is
then parsed, by splitting into tokens and performing insertion actions for each
token, measuring the total time for all insertions.

(The base program consists of 100 lines, 2233 chars, and 1270 tokens, resulting in an average of 12.7 tokens and 22.3 characters per line).

\subsection{Left-to-right parsing performance results}

Parsing results are shown in \autoref{fig:parse-time}. Performance is currently quadratic in (realistic) program length owing to the fact that each let definition imposes an additional level of nesting. We are considering a specific optimization for top-level definition forms in the future to make this linear.


\begin{figure}[h]
    \centering
    \includegraphics[width=\textwidth]{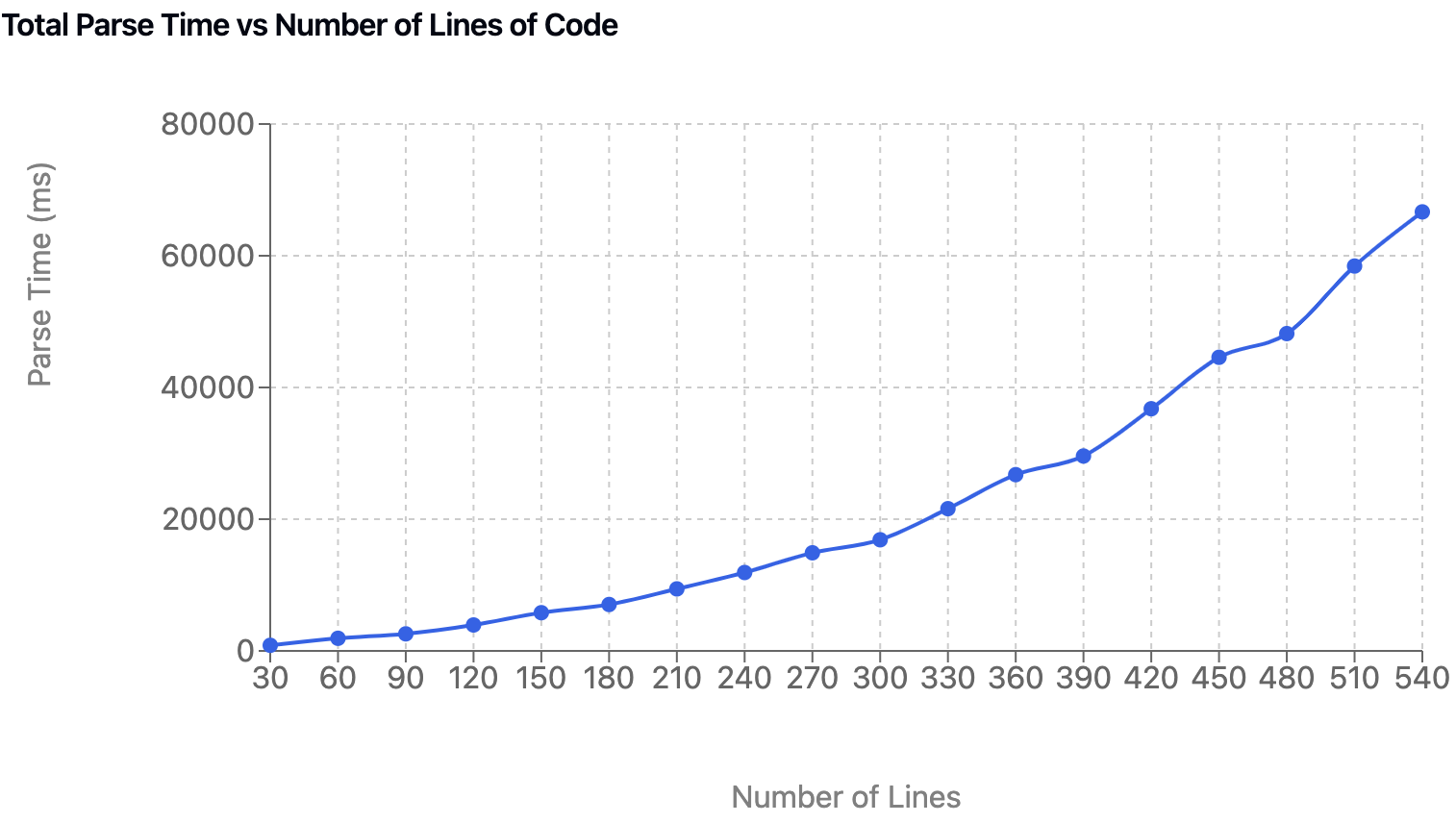}
    \caption{Time taken to parse syntactically correct and prefix-complete programs across a range of program lengths}
    \label{fig:parse-time}
\end{figure}


\subsection{Insertion action performance setup}

Taking the base 100-line program from the above task, we deleted 20 single-token terms distributed randomly over the program, leaving 20 operand obligations. For each resulting obligation, we measured the time taken to insert a single character token into the hole, repeating each such insertion 200 times to increase precision.

For each operand obligation, we derived the syntactic nesting depth (number of containing forms) and the total length of the operand sequence (prefix plus suffix) in which the hole is contained, and plotted these versus the time taken per insertion.

\subsection{Insertion action performance results}

\begin{figure}[!h]
  \centering
  \begin{minipage}[b]{0.45\textwidth}
    \includegraphics[width=\textwidth]{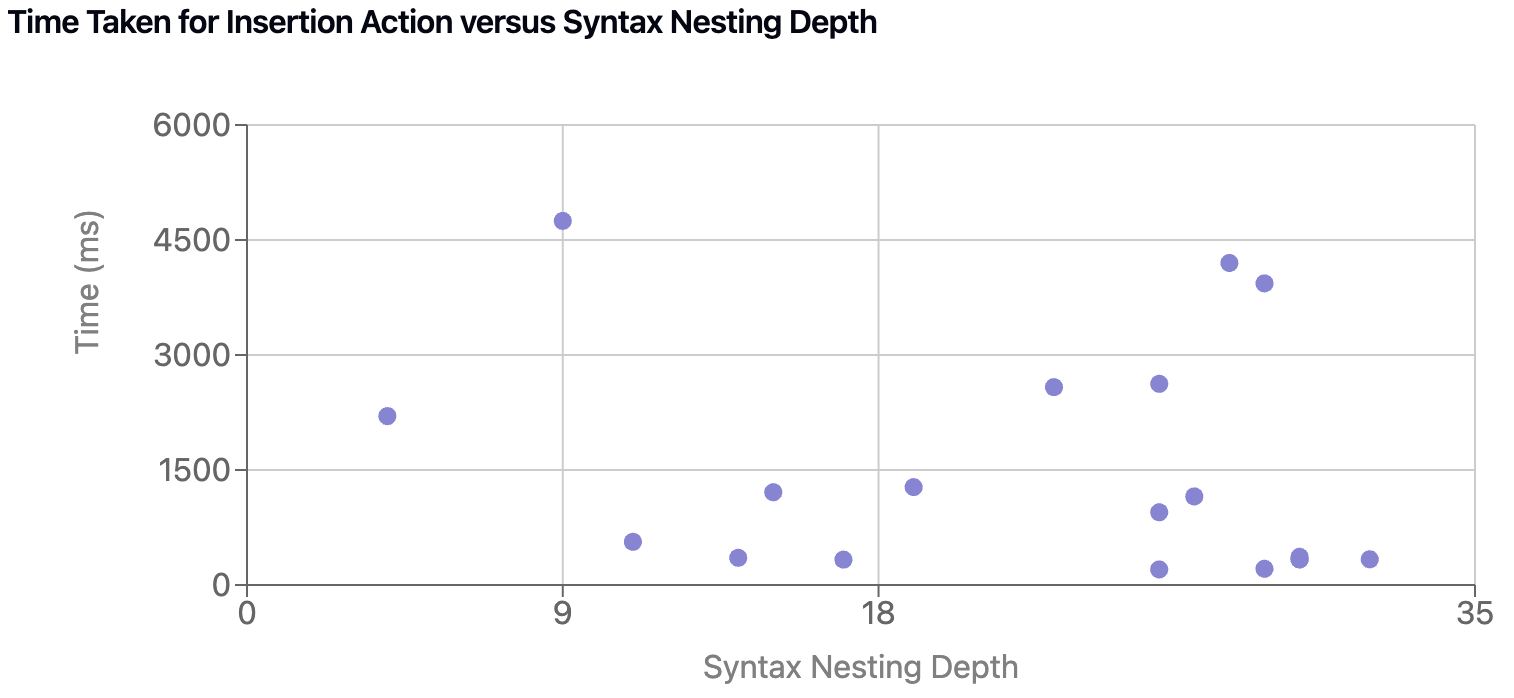}
    \caption{Time taken to perform 200 single-character token insertions in an operand hole of the specified depth}
    \label{fig:insertion-depth}
  \end{minipage}
  \hfill
  \begin{minipage}[b]{0.45\textwidth}
    \includegraphics[width=\textwidth]{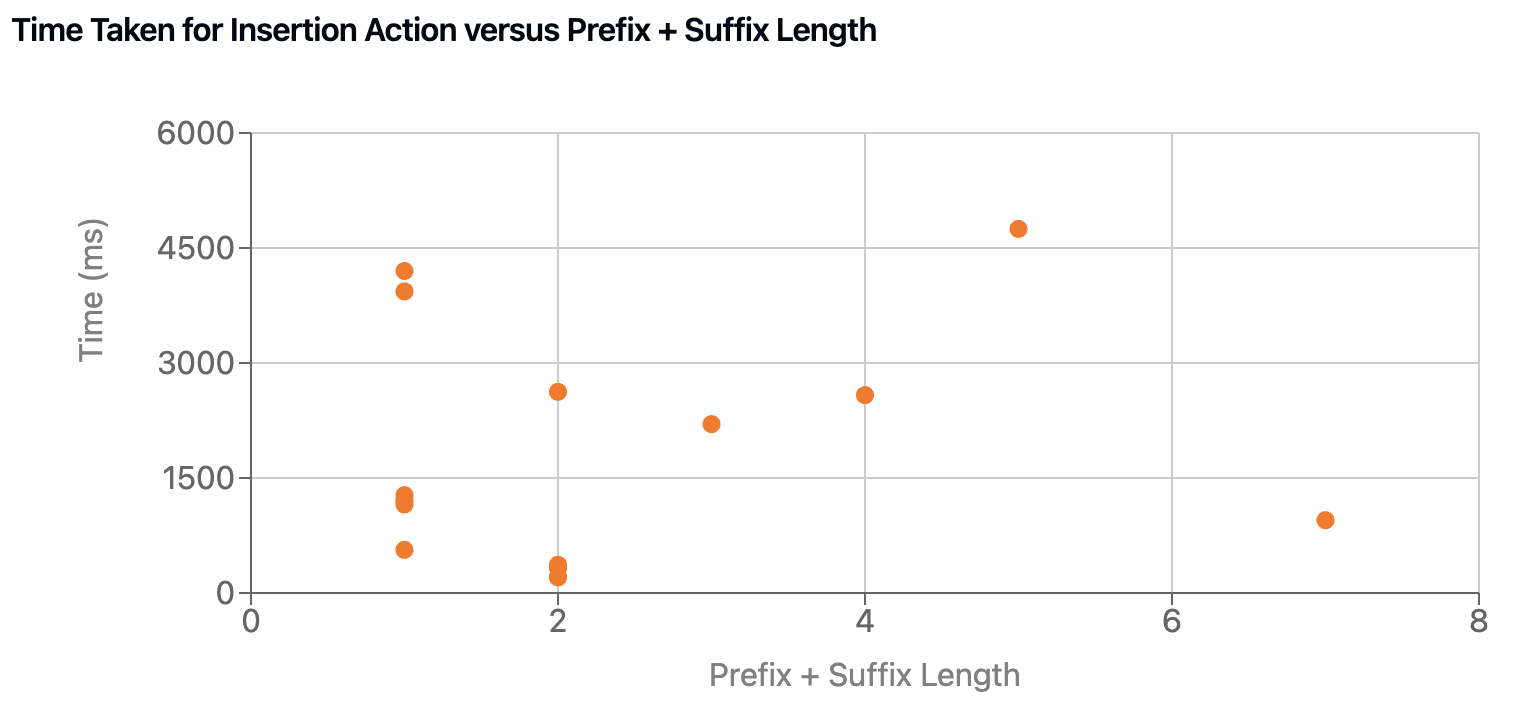}
    \caption{Time taken to perform 200 single-character token insertions in an operand hole of the specified operand sequence length}
    \label{fig:insertion-length}
  \end{minipage}
\end{figure}


\end{document}